\documentclass[10pt]{article}
\usepackage[english]{babel}
\usepackage{listings}
\usepackage{rotating}
\usepackage[T1]{fontenc}
\usepackage{amssymb,amsmath,amsfonts,amsthm,scrextend}
\usepackage{mathrsfs}
\usepackage{enumitem}
\usepackage{stmaryrd}
\usepackage{graphics}
\usepackage{graphicx}
\usepackage{subfig}
\usepackage{fullpage}
\usepackage{hyperref}
\usepackage{appendix}
\usepackage{dsfont}
\usepackage{mdframed}
\usepackage{mathtools}
\usepackage{algorithm}
\usepackage{algpseudocode}
\usepackage{scalerel}[2016/12/29]
\usepackage{multicol}
\usepackage{makecell}
\usepackage{tikz}
\usetikzlibrary{calc}

\algrenewcommand\algorithmicrequire{\textbf{Input:}}

\theoremstyle{definition}
\newtheorem{Def}{Definition}[section]
\theoremstyle{plain}
\newtheorem{Thm}[Def]{Theorem}
\newtheorem{Prop}[Def]{Proposition}
\newtheorem{Corol}[Def]{Corollary}
\newtheorem{Lemma}[Def]{Lemma}

\newtheorem{Pb}[Def]{Problem}

\theoremstyle{remark}
\newtheorem{Ex}[Def]{Example}
\newtheorem{Rq}[Def]{Remark}

\SetSymbolFont{stmry}{bold}{U}{stmry}{m}{n}

\newcommand{\ind}[1]{\mathds{1}_{#1}}
\newcommand{\tm}{\text{-}}

\DeclareMathOperator{\ev}{ev}
\DeclareMathOperator{\fibresp}{fs}
\newcommand{\gbc}[3]{\begin{bmatrix}#2\\#1\end{bmatrix}_{\!#3}}
\DeclareMathOperator{\GL}{GL}
\DeclareMathOperator{\im}{im}
\newcommand{\qdeg}{q\text{-}\!\deg}
\DeclareMathOperator{\rank}{rank}
\DeclareMathOperator{\rk}{rk}
\newcommand{\sgbc}[3]{\left[\begin{smallmatrix}#2\\#1\end{smallmatrix}\right]_{#3}}
\DeclareMathOperator{\slicesp}{ss}
\DeclareMathOperator{\Span}{Span}
\DeclareMathOperator{\Supp}{Supp}
\DeclareMathOperator{\tr}{tr}
\DeclareMathOperator{\trank}{trank}
\newcommand{\wt}{\textup{w}}

\newcommand{\F}{\mathbb{F}}
\newcommand{\N}{\mathbb{N}}
\newcommand{\Z}{\mathbb{Z}}
\newcommand{\GG}{\mathscr{G}}
\newcommand{\MM}{\mathscr{M}}
\newcommand{\PP}{\mathscr{P}}
\newcommand{\AAA}{\mathcal{A}}
\newcommand{\BBB}{\mathcal{B}}
\newcommand{\CCC}{\mathcal{C}}
\newcommand{\EEE}{\mathcal{E}}
\newcommand{\III}{\mathcal{I}}
\newcommand{\JJJ}{\mathcal{J}}
\newcommand{\KKK}{\mathcal{K}}
\newcommand{\OOO}{\mathcal{O}}
\newcommand{\SSS}{\mathcal{S}}
\newcommand{\TTT}{\mathcal{T}}
\newcommand{\UUU}{\mathcal{U}}
\newcommand{\Rr}{\mathfrak{R}}
\newcommand{\Aaa}{\mathfrak{a}}
\newcommand{\Ddd}{\mathfrak{d}}
\newcommand{\Mmm}{\mathfrak{m}}
\newcommand{\Sss}{\mathfrak{s}}

\newcommand*\samethanks[1][\value{footnote}]{\footnotemark[#1]}

\title{Decoding Algorithms for Tensor Codes}
\author{Eimear Byrne\thanks{School of Mathematics and Statistics, University College Dublin, \texttt{ebyrne@ucd.ie} and \texttt{lucien.francois@ucd.ie}}, Alain Couvreur \thanks{Inria, Laboratoire LIX, École polytechnique, Institut Polytechnique de Paris, \texttt{alain.couvreur@inria.fr}}, Lucien François\samethanks[1]\ \thanks{This article is the long version of the extended abstract \cite{BCF25} presented at the conference ISIT'25. It includes the missing proofs of the mentioned propositions, additional results and details on the generalisation to higher order tensors.}}

\begin{document}
\maketitle

\begin{center}
    \textbf{Abstract} \\[0.4cm]
    \begin{minipage}{15cm} \bf \small 
    Tensor codes are a generalisation of matrix codes. Such codes are defined as subspaces of order-$r$ tensors for which the ambient space is endowed with the tensor-rank as a metric. A class of these codes was introduced by Roth, who also outlined a decoding algorithm for low tensor-rank errors that can be generalised to an algorithm with exponential complexity in the decoding radius. They may be viewed as a generalisation of the well-known Delsarte-Gabidulin-Roth maximum rank distance codes. We study a generalised class of these codes. We investigate their properties and outline decoding techniques for different metrics that leverage their tensor structure. We first consider a fibre-wise  decoding approach, as each fibre of a codeword corresponds to a Gabidulin codeword. We then give a generalisation of Loidreau-Overbeck's decoding method that corrects errors with properties constrained by the dimensions of the slice spaces and fibre spaces. The metrics we consider are bounded from above by the tensor-rank metric, and therefore these algorithms also decode tensor-rank weight errors. 
    \end{minipage}
\end{center}

\setcounter{tocdepth}{1}
\section{Introduction}

Error-correcting codes are critical for ensuring data reliability in modern communication systems. Rank-metric codes such as the well-known Delsarte-Gabidulin-Roth codes have attracted significant attention due to their applicability in network coding. Such codes have seen generalisations, such as the twisted Gabidulin codes \cite{LUNARDON201879,sheekey} and the tensor codes. Tensor codes as a type of code have been introduced in \cite{RothMaximumArrayCodes} with a view towards criss-cross error correction, and a particular class has been studied in \cite{ROTHTensorCodesForRankMetric}. Roth's tensor codes are a family of subspaces of the vector space of tensors endowed with the tensor-rank as a metric and were constructed to have a known lower bound on the minimum tensor-rank distance. For this family of codes, Roth gave decoding algorithms for tensor-rank one errors and for tensor-rank two errors reaching the decoding radius of the codes, with polynomial complexity for the first and exponential complexity in the decoding radius for the second. He also suggested generalisations to higher radius. 

In this paper, we derive decoding algorithms for a generalisation of the tensor codes that were introduced in \cite{ROTHTensorCodesForRankMetric}. In our treatment we consider a number of new metrics, each of which provides a lower bound on the tensor rank as a weight function. We give bounded-distance decoders for these metrics which can hence be considered as bounded-distance decoders for the tensor rank as a metric with an additional range of decodable errors. The algorithms presented here all have polynomial complexity in the size of the tensors, for each decoding radius chosen. Although the presentation focuses on order-3 tensors, the generalisation of the algorithms to higher order tensors is straightforward and preserves this property. We also give a detailed comparison of the different algorithms. We base the first type of algorithm on the direct sum of codes in which the code is contained, and the second type on the Welch-Berlekamp-like algorithm introduced by Loidreau and Overbeck (see \cite{welch1986error}, \cite[Section~6.8]{citekHandbookCodingTheoryDecodeCyclicCodesey} and \cite{LoidreauWelshBerlekampGabidulin}).

The outline of this paper is as follows.
In Section \ref{sec:prelim}, we introduce basic properties of codes and tensors. In Section \ref{sec:evaluationcodes}, we introduce relevant properties of linearised and bilinearised polynomials and describe the class of tensor codes introduced by Roth in \cite{ROTHTensorCodesForRankMetric}. Furthermore, we introduce a class of evaluation codes that extend Roth's codes, describe its properties and its link with the original codes. We introduce a number of metrics relevant to the decoding and describe the decoding capabilities of the codes in respect of these distance functions. In Section \ref{sec:decoding}, we present four different decoding algorithms. The first two are based on Gabidulin decoders and exploit the fact that the family of tensor codes we study are subcodes of direct sums of Gabidulin codes. The last two algorithms use an approach similar to Loidreau-Overbeck using a left-Euclidean like factoring algorithm on solutions of a linear equation system. After studying the error correction capability of these algorithms, we consider the extended range of decodable errors provided by the application of existing probabilistic methods, namely supercode decoding methods and interleaving. In Section \ref{sec:Comparison}, we compare our different algorithms in terms of their decoding capability and their asymptotic complexity. Finally, in Section \ref{sec:Generalisation}, we give details on the generalisation of the aforementioned codes and algorithms to higher order tensor-codes.

A complementary GitHub repository (\url{https://github.com/lucienfrancois/RothTensorCodes}) provides MAGMA implementations of the algorithms discussed in this paper. Computations of the examples mentioned below can be produced using this MAGMA code.

\section{Preliminaries}\label{sec:prelim}

We will use the following notation. We refer the reader to \cite{lidl1994introduction,GABIDULINRankDistanceCodes,Burgisser1997ch14} for further background reading on finite fields, $q$-polynomials, rank-metric codes, and tensors. 

\subsection{Notation}
Throughout, we let $q$ be a fixed prime power and denote by $\mathbb{F}_q$ the finite field of order $q$. We let $\N_0$ be the set of non-negative integers and $\mathbb N$ be the set of positive integers. We let $n$ be a fixed positive integer. For a predicate $\PP$, we denote by $\ind{\PP}$ the indicator function, which is equal to one if $\PP$ is true, and equal to zero otherwise. We let $X,Y,Z$ be indeterminates. We fix a basis $(\alpha_1,\dots,\alpha_n)$ of $\mathbb F_{q^n}/\mathbb F_q$ and denote by $\alpha$ the vector $\alpha := (\alpha_1,\dots,\alpha_n)$. For each integer $j$, we define 
$\alpha^{q^j}:=(\alpha_1^{q^j},\dots,\alpha_n^{q^j})$.

We define $\llbracket L_1,L_2\rrbracket:=\{L_1,L_1+1,\dots,L_2-1,L_2\}$ for any positive integers $L_1,L_2$ with $L_1\leq L_2$. 
For any tuple ${x}$ of length $L$, we denote by $x_\ell$ its $\ell^{th}$ entry, for each $\ell \in \llbracket1,L\rrbracket$. For any $j \in \{ 1,2\}$, we write $\pi_j$ to denote the projection map 
\[\pi_j : \mathbb{N}_0^2 \to \mathbb{N}_0, x \mapsto x_j.\]  For any $i \in \mathbb{N}_0^2$ and $r \in \mathbb{N}_0$, we define $i + r :=(i_1 + r,i_2 +r)$. Similarly, for 
$I \subseteq {\mathbb{N}_0}^2$ and $R \subseteq {\mathbb{N}_0}$, we define $I + R := \{i +r \ |\  i \in  I,\ r \in R\}$.

For any $\mathbb{F}_q$-bilinear map $b : \mathbb{F}_{q^n} \times \mathbb{F}_{q^n} \to \mathbb{F}_{q^n}$  we denote the {left} (\emph{resp.} right) \textbf{radical} of $b$ by \[\Rr\Aaa\Ddd_1(b) := \{x \in \mathbb F_{q^n} \ | \ \forall y \in \mathbb F_{q^n}: b(x,y)= 0\}\quad \text{\emph{resp.}} \quad \Rr\Aaa\Ddd_2(b):= \{y \in \mathbb F_{q^n} \ | \ \forall x \in \mathbb F_{q^n}: b(x,y)= 0\};\]
see \cite[Def 8.3]{cooperstein2010advanced}.

\subsection{Rank-metric codes and Gabidulin codes}
\begin{Def}
    For any $w \in \F_{q^n}^n$ we define the rank weight or the $\F_q$\textbf{-rank} of $w$ to be $\rank_{\F_q}(w) := \dim_{\F_q}\Span_{\F_q}(w_1,\dots,w_n)$. 
    The rank distance between $u,v \in \F_{q^n}^n$ is defined to be $\rank_{\F_q}(u-v)$.
    The map $d_{\rank}:\F_{q^n}^n \times \F_{q^n}^n \to \N_{ 0}, (u,v)\mapsto \rank_{\F_q}(u-v)$ is a distance function on $\F_{q^n}^n$ called the \textbf{rank-metric}.
\end{Def}

\begin{Def}
    Endowed with the metric $d_{\rank}$, an $\F_{q^n}\tm[n,k]$ \textbf{linear rank-metric code} is an $\F_{q^n}$-linear subspace  $\CCC$ of $\F_{q^n}^n$ of dimension $k$ over $\F_{q^n}$. If $\CCC$ is non-zero, we say that $\CCC$ is an $\F_{q^n}\tm[n,k,d]$ code with $d := d_{\rank}(\CCC) := \min\{ d(u,v) \ : \ u,v \in \CCC \}$.
\end{Def}

\begin{Def}\label{def:GabidulinCode}
    The length $n$ \textbf{Gabidulin code} of dimension $k$ over $\mathbb{F}_{q^n}$ evaluated at the basis $\alpha$ is the $\F_{q^n}\tm[n,k,n-k+1]$ code:
    $$\GG_k(\alpha):=\left\{ (V(\alpha_1),\dots,V(\alpha_n)) \ \left| \ V(Z) =\sum_{\ell = 0}^{k-1}v_\ell Z^{q^\ell}, v_i \in \mathbb{F}_{q^n} \right\}\right. .$$
\end{Def}
 We remind the reader that Gabidulin codes are maximal rank distance codes, \emph{i.e.}  that the minimum $\mathbb{F}_q$-rank distance of $\GG_k(\alpha)$ is $n-k+1$, see \cite{GABIDULINRankDistanceCodes}. 
 \begin{Def}\label{def:MooreMatrixDefinition}
     For each integer $k \in \llbracket 1,n\rrbracket$, we denote respectively by $M(\alpha)$ and $M_k(\alpha)$ the \textbf{Moore matrix} of the basis $\alpha$ and its $k\times n$ submatrix defined by
     
     \begin{equation}
    \label{eq:MooreMatrix}
    M(\alpha) = \begin{bmatrix}
        \alpha_1 & \hdots & \alpha_n \\
        \alpha_1^{q} & \hdots & \alpha_n^{q} \\ 
        \vdots & \ddots & \vdots \\
        \alpha_1^{q^{n-1}} & \hdots & \alpha_n^{q^{n-1}} \\ 
    \end{bmatrix}
    \quad \text{ and } \quad 
    M_k(\alpha) = \begin{bmatrix}
        \alpha_1 & \hdots & \alpha_n \\
        \alpha_1^{q} & \hdots & \alpha_n^{q} \\ 
        \vdots & \ddots & \vdots \\
        \alpha_1^{q^{k-1}} & \hdots & \alpha_n^{q^{k-1}} \\ 
    \end{bmatrix}.
\end{equation}
\end{Def}
Note that $M_k(\alpha)$ is a generator matrix of $\GG_k(\alpha)$ as an $\F_{q^n}$-linear code. In addition, note that $M(\alpha)$ is non-singular over $\F_{q^n}$ and that its inverse is $M(\alpha^{\bot})^\top$ by definition of the dual basis $\alpha^\bot$ of $\alpha$, see \cite[Def 2.30]{lidl1994introduction}.

\subsection{Tensors and matrices} \label{subsec:tensorsandmatrices}
We denote by $\mathbb{F}_{q}^{n\times n}$ the vector space of matrices of size $n \times n$ over the field with $q$ elements. A $\boldsymbol{3}$\textbf{-tensor} of size $n \times n \times n$ over a field $\mathbb{F}_q$ is any element of the $n^3$-dimensional $\mathbb{F}_q$-vector space $\mathbb{F}_q^n \otimes_{\F_q} \mathbb{F}_q^n \otimes_{\F_q} \mathbb{F}_q^n$, also denoted by $(\mathbb{F}_q^{n})^{\otimes 3}$. An element $T \in (\mathbb{F}_q^{n})^{\otimes 3}$ is uniquely associated with its coordinate expression $(T[i])_{i \in \llbracket 1,n\rrbracket^3}$  as a $3$-dimensional array with respect to the basis $(e_{i_1} \otimes e_{i_2}\otimes e_{i_3})_{i \in \llbracket 1,n\rrbracket^3}$, where we denote by $(e_\iota)_{\iota \in \llbracket 1,n\rrbracket}$ the standard basis vector of $\mathbb{F}_q^n$, see \cite[Corollary 2.24]{Lang2002TensorProduct}. In other words, $T$ is uniquely associated with the sequence $(T[i])_{i \in \llbracket 1,n\rrbracket^3}$ of scalars in $\F_q$ satisfying $ T = \sum_{ i \in \llbracket 1,n \rrbracket^3} T[i] e_{i_1} \otimes e_{i_2}\otimes e_{i_3}$. 
{\renewcommand{\thesection}{\arabic{section}}

Let $\omega = \{\omega_1,\dots, \omega_n\}$ be a basis of $\mathbb{F}_{q^n}/\mathbb{F}_q$. 
Then the map $\mathfrak{s}_\omega: (\mathbb{F}_q^n)^{\otimes 3} {\longrightarrow} \mathbb{F}_{q^n}^{n\times n} $ defined by
\[\mathfrak{s}_\omega(\Gamma) := \sum_{i_{3} = 1}^n \omega_{i_{3}} (\Gamma[i_1,i_2,i_{3}])_{(i_1,i_2) \in \llbracket 1,n\rrbracket^2}\] for all 
$\Gamma \in (\mathbb{F}_q^n)^{\otimes 3}$ is an $\mathbb{F}_q$-isomorphism that yields an expression of every $3$-tensor over $\mathbb{F}_{q}$ as an $n\times n$ matrix over $\mathbb{F}_{q^n}$, as shown in Figure~\ref{fig:Isomorphismsomega}. Therefore, every $\mathbb{F}_{q}$-vector subspace $\mathcal{C}$ of $(\mathbb{F}_{q}^n)^{\otimes 3}$ can be viewed as an $\mathbb{F}_q$-vector subspace $\mathfrak{s}_{\omega}(\mathcal{C})$ of $\mathbb{F}_{q^n}^{n\times n}$. We say that $\mathcal{C}$ is the \textbf{associated tensor code} of $\mathfrak{s}_{\omega}(\mathcal{C})$ with respect to the basis $\omega$. This is a generalisation of the associated matrix codes of $\mathbb{F}_{q^n}$-linear codes introduced in \cite{RothMaximumArrayCodes}.  For further details, we refer the reader to \cite{Burgisser1997ch14} and \cite{ConciseCodingTheoryRankMetricCodes}.

\begin{Def}
    An $\F_{q^n}\tm[n\times n,k]$ \textbf{tensor code} is an $\F_{q^n}$-linear subspace of $\F_{q^n}^{n\times n}$ of dimension $k$ over $\F_{q^n}$.
    
\end{Def}
We will consider tensor codes endowed with the tensor-rank metric and other related metrics; see Subsection~\ref{subsec:distancesinthecode}.

\begin{Ex}
    For $n = 3$, the tensor $e_1 \otimes (e_1+e_2+e_3) \otimes e_1 + (e_1 \otimes e_1 + e_2 \otimes e_2 + e_3 \otimes e_3)\otimes e_2 + (e_1+e_2+e_3)\otimes e_1 \otimes e_3 $ in the space $(\F_q^3)^{\otimes 3}$ corresponds to the tensor represented on the left-hand side of Figure~\ref{fig:Isomorphismsomega}.
    
    \begin{figure}[ht]
    \centering
    \begin{tikzpicture}[thick,scale=0.8, every node/.style={scale=0.8}]
    \draw[thick,->] (-3,1) -- (-3,0) node [midway, left] {1};
    \draw[thick,->] (-3,1) -- (-2,1) node [midway, below] {2};
    \draw[thick,->] (-3,1) -- (-2.25,1.5) node [midway, above] {3};
        \node[draw,rectangle] (M) at (0,0) {$\begin{matrix}
            1&1&1\\
            0&0&0\\
            0&0&0
        \end{matrix}$};
        \node[draw,rectangle] (N) at (1.8,0.66) {$\begin{matrix}
            1&0&0\\
            0&1&0\\
            0&0&1
        \end{matrix}$};
        \node[draw,rectangle] (P) at (3.6,1.33) {$\begin{matrix}
            1&0&0\\
            1&0&0\\
            1&0&0
        \end{matrix}$};
        \draw (M)--(N)--(P);
        \draw (M.north west) -- (N.north west)-- (P.north west);
        \draw (M.south east) -- (N.south east)-- (P.south east);

        \node[right] at (M.south east) {$\quad\bullet\  k = 1$};
        \node[right] at (N.south east) {$\quad\bullet\  k = 2$};
        \node[right] at (P.south east) {$\ \bullet\  k = 3$};

        \node[scale = 1.5] at (7,1) {$\overset{\mathfrak{s}_\omega}\longmapsto$};

        \node at (10,1) {$\begin{bmatrix}
            \omega_1 + \omega_2 + \omega_3 & \omega_1 &\omega_1 \\
            \omega_3 &\omega_2 &0\\
            \omega_3 & 0&\omega_2
        \end{bmatrix}$};
    \end{tikzpicture}
    \caption{Illustration of $\mathfrak{s}_{\omega}$ for $n = 3$.}
    \label{fig:Isomorphismsomega}
\end{figure}
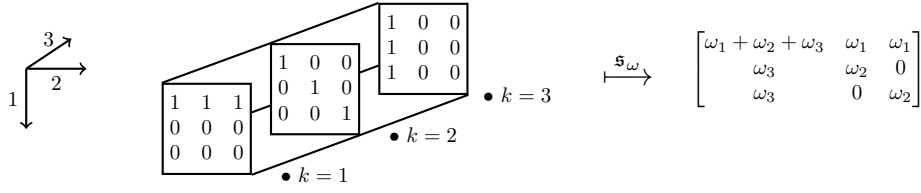

\end{Ex}}

\subsubsection*{Coefficients, submatrices, and subtensors}
Given a matrix $M \in \mathbb{F}_q^{n\times n}$ for any $i_1,i_2 \in \llbracket 1,n\rrbracket$, we denote equivalently by $M_{i_1,i_2}$,$M_{(i_1,i_2)}$, or $M[i_1,i_2]$ the entry of $M$ in the $i_1^{th}$ row and $i_2^{th}$ column of $M$. We use colon notation to denote the element (vector, matrix, tensor, array) obtained if the index at the location of the colon varies. For instance, if $M\in \mathbb{F}_q^{n\times n}$ and if $i_1 \in \llbracket 1,n\rrbracket$, then $M[i_1,:] = {(M_{i_1,i_2})}_{i_2 \in \llbracket 1,n\rrbracket}$ is the $i_1^{th}$ row of the matrix $M$, and if $T \in (\mathbb{F}_q^{n})^{\otimes 3}$ is a $3$-tensor, then $T[i_1,:,:] = (T[i_1,i_2,i_3])_{(i_2,i_3) \in \llbracket 1,n\rrbracket^2}$ is the $n^2$-tuple/matrix of size $n \times n$ whose entry at $(i_2,i_3)$ is $T[i_1,i_2,i_3]$. 

\subsubsection*{Slices and fibres}

For each $j \in \{1,2,3\}$, we consider the \emph{$j$-{slices}} of a tensor in $(\mathbb{F}_q^n)^{\otimes 3}$ to be the matrices obtained by fixing \emph{only the $j^{th}$ coordinate} in its set of coordinates. In parallel, we consider the \emph{$j$-{fibres}} of a tensor in $(\mathbb{F}_q^n)^{\otimes 3}$ to be the vectors obtained by fixing \emph{all but the $j^{th}$ coordinate} in its set of coordinates. More precisely, we introduce the following definitions.

\begin{Def}
    Let $\Gamma$ be a tensor in $(\mathbb{F}_q^n)^{\otimes 3}$ and fix $i_1,i_2,i_3 \in \llbracket 1,n\rrbracket$.\begin{itemize}
       \item  The \emph{$i_1$-th $1$-\textbf{slice}} of $\Gamma$ is the matrix $\Gamma[i_1,:,:] = (\Gamma[i_1,\tilde i_2,\tilde i_3])_{(\tilde i_2,\tilde i_3) \in \llbracket 1,n\rrbracket^2}$. 
     \item The \emph{$i_2$-th $2$-\textbf{slice}} of $\Gamma$ is the matrix $\Gamma[:,i_2,:] = (\Gamma[\tilde i_1, i_2,\tilde i_3])_{(\tilde i_1,\tilde i_3) \in \llbracket 1,n\rrbracket^2}$.
     \item The \emph{$i_3$-th $3$-\textbf{slice}} of $\Gamma$ is the matrix $\Gamma[:,:,i_3] = (\Gamma[\tilde i_1,\tilde i_2, i_3])_{(\tilde i_1,\tilde i_2) \in \llbracket 1,n\rrbracket^2}$.
    \end{itemize}
    For each $j \in \{1,2,3\}$, we denote by $\slicesp_j(\Gamma)$ the $j^{th}$ slice-space of $\Gamma$, that is the $\F_q$-vector subspace of $\F_{q}^{n\times n}$ spanned by its $j$-slices.
\end{Def}

\begin{Def}
    Let $\Gamma$ be a tensor in $(\mathbb{F}_q^n)^{\otimes 3}$ and fix $i_1,i_2,i_3 \in \llbracket 1,n\rrbracket$.\begin{itemize}
       \item  The \emph{$1$-\textbf{fibre}} of $\Gamma$ at $(i_2,i_3)$ is the vector $\Gamma[:,i_2,i_3] = (\Gamma[\tilde i_1,i_2, i_3])_{\tilde i_1 \in \llbracket 1,n\rrbracket}$. 
       \item The \emph{$2$-\textbf{fibre}} of $\Gamma$ at $(i_1,i_3)$ is the vector $\Gamma[i_1,:,i_3] = (\Gamma[ i_1,\tilde i_2, i_3])_{\tilde i_2 \in \llbracket 1,n\rrbracket}$. 
       \item The \emph{$3$-\textbf{fibre}} of $\Gamma$ at $(i_1,i_2)$ is the vector $\Gamma[i_1,i_2,:] = (\Gamma[i_1,i_2,\tilde i_3])_{\tilde i_3 \in \llbracket 1,n\rrbracket}$. 
    \end{itemize}
    For each $j \in \{1,2,3\}$, we denote by $\fibresp_j(\Gamma)$ the $j^{th}$ fibre-space of $\Gamma$, that is the $\F_q$-vector subspace of $\F_q^n$ spanned of all of its $j$-fibres.
\end{Def}

\begin{Rq}
Let $R \in \F_{q^n}^{n\times n}$ and let $\Gamma = \Sss_3^{-1}(R) \in (\F_{q}^n)^{\otimes 3}$ be its associated tensor. 
Note that a row or a column of $R$, corresponds to a slice of 
$\Gamma$, by construction of the map $\Sss_3$. In Section~\ref{sec:decoding}, we will introduce algorithms that we refer to as ``fibre-wise decoding'' as they apply a Gabidulin decoder on every column or every row of such an element $R$, hence on every fibre along one mode. As we discuss in Section~\ref{sec:Generalisation}, this correspondence between fibres of the matrix over the extension field and slices of the corresponding tensor over the base field is no longer true for higher order tensors. The generalisations of Algorithms~\ref{alg:fibrewise},~\ref{alg:fibrewisebis}' or~\ref{alg:fibreWiseTwoDirrs} decode a received word (as a tensor with coefficients over $\F_{q^n}$) fibre-wise, in other words by decoding successively the vectors over $\F_{q^n}$ obtained by fixing all but one coordinate of the received word, which is a fibre of this word. This  fibre corresponds to a matrix over $\F_q$ of the associated tensor over $\F_q$ obtained by fixing all but two coordinates. 
\end{Rq}

\section{Evaluation codes of bilinearised polynomials}
\label{sec:evaluationcodes}
\subsection{Linearised and bilinearised polynomials}

\textbf{Linearised polynomials} or $q$-polynomials over $\mathbb{F}_{q^n}$ are univariate polynomials in the $\mathbb{F}_{q^n}$-linear subspace of $\mathbb{F}_{q^n}[Z]$ spanned by the monomials of the form $Z^{q^r}$ where $r$ is a non-negative integer. These polynomials are endowed with the addition and the composition map, which gives them a structure of non-commutative ring denoted $\mathcal{M}_{q,\mathbb{F}_{q^n}}[Z]$. It is known that such a polynomial induces an $\mathbb{F}_q$--linear map since its monomials have exponents equal to powers of $q$. Moreover, the $\mathbb{F}_q$-dimension of the kernel of a non-zero linearised $q$-polynomial is bounded from above by its $q$-degree, see \cite[Chapter 4]{lidl1994introduction}. We will now generalise the notion and focus on additive vector subspaces of bivariate polynomials in the commutative algebra $\F_{q^n}[X,Y]$.

\begin{Def}
    The \textbf{bilinearised $\boldsymbol q$-polynomials} over 
     $\mathbb{F}_{q^n}$ are the elements of the $\F_q$-subspace
    $$\mathcal{M}_{q,\mathbb{F}_{q^n}}[X,Y] := \Span_{\mathbb{F}_{q^n}}\left\{ \left. X^{q^{ i_1}} Y^{q^{ i_2}}  \right| i \in {\mathbb{N}_0}^2 \right\} \subseteq \mathbb{F}_{q^n}[X,Y].$$
    
    An element $F(X,Y) \in \mathcal{M}_{q,\mathbb{F}_{q^n}}[X,Y]$ has an expression as 
    $F(X,Y) = \sum_{i \in S} F_i X^{q^{i_1}}Y^{q^{ i_2}}$ for some finite set $S \subset {{\mathbb{N}_0}^2}$. 
    We then define the support of $F$ to be $\Supp(F):=\{i \in S : F_{{i}} \neq 0 \}$.
    If $F(X,Y)$ has non-empty support $S$, we define the \textbf{partial $q$-degree} of $F(X,Y)$ along  $X$ (respectively $Y$), expressed as $\qdeg_X F(X,Y)$, (respectively $\qdeg_Y F(X,Y)$) to be $\max\{i_1 : i \in S\}$ (respectively $\max\{i_2 : i \in S\}$). 
\end{Def}

\begin{Ex}
    If $F(X,Y) = X^{q}Y^{q^5} + XY^{q} + X^{q^3}Y$, then $\qdeg_X F(X,Y) = 3$ and $\qdeg_Y F(X,Y) = 5$.
\end{Ex}

By the $\F_q$-linearity of the Frobenius map on $\mathbb{F}_{q^n}$ defined by $x \mapsto x^q$, we have the following observation.

\begin{Prop}
    Let $F(X,Y) \in \mathscr{M}_{q,\mathbb{F}_{q^n}}[X,Y]$. Then the evaluation map given by $\mathbb{F}_{q^n}^2 \to \mathbb{F}_{q^n}, (x,y) \mapsto F(x,y)$ is bilinear over $\mathbb{F}_q$.
\end{Prop}

The subspace of bilinearised $q$-polynomials is not a subalgebra of $\mathbb{F}_{q^n}[X,Y]$ with respect to usual polynomial multiplication. 
However, the $q^{th}$ power of a bilinearised $q$-polynomial is also a bilinearised $q$-polynomial, therefore, the space is stable upon left-composition by a (single variable) linearised $q$-polynomial. In particular, the space $\MM_{q,\F_{q^n}}[X,Y]$ has the structure of an $(\MM_{q,\F_{q^n}}[Z],+,\circ)$-module, as stated below.

\begin{Prop}
    The abelian additive group $\mathscr{M}_{q,\mathbb{F}_{q^n}}[X,Y]$ is a left-module over the non-commutative ring $(\mathscr{M}_{q,\mathbb{F}_{q^n}}[Z],+,\circ)$ with the external binary operation:
    $$\circ : \left\{\begin{array}{rcl}
         \mathscr{M}_{q,\mathbb{F}_{q^n}}[Z] \times \mathscr{M}_{q,\mathbb{F}_{q^n}}[X,Y] &\longrightarrow&\mathscr{M}_{q,\mathbb{F}_{q^n}}[X,Y] \\
         (V(Z) , F(X,Y)) & \longmapsto & V(F(X,Y)).
    \end{array}\right.$$
\end{Prop}

\label{subsec:PolySpaceAndEvalSpaceINTRO}

Define 
\begin{equation}\label{eq:UXY}
U_{X,Y}: = \{f(X^{q^n}-X,Y)  + g(X,Y^{q^n}-Y)\ | \ f(X,Y),g(X,Y) \in \mathscr{M}_{q,\mathbb{F}_{q^n}}[X,Y]\}.
\end{equation}
This is a vector subspace of $\mathscr{M}_{q,\mathbb{F}_{q^n}}[X,Y]$, which is contained in the ideal $(X^{q^n}-X,Y^{q^n}-Y)$ of the ring $\F_{q^n}[X,Y]$. Indeed, for any $f(X,Y) \in \MM_{q,\F_{q^n}}[X,Y]$, since $X$ divides $f(X,Y)$, then $X^{q^n}-X$ divides  $f(X^{q^n}-X,Y)$. For any polynomial, denote by $\overline{F(X,Y)} = F(X,Y) + U_{X,Y}$ the coset containing $F(X,Y)$ in the quotient vector space $\mathscr{M}_{q,\mathbb{F}_{q^n}}[X,Y]/U_{X,Y}$.

\begin{Lemma} \label{lemma:BasisOfQutientSpace}
    The set $\{\overline{X^{q^{ i_1}}Y^{q^{ i_2}}}:i \in \llbracket 0,n-1 \rrbracket^2\}$ is a basis of $\mathscr{M}_{q,\mathbb{F}_{q^n}}[X,Y]/U_{X,Y}$.
\end{Lemma}
\begin{proof}
     For each $t,s\in {\mathbb{N}_0}$ we have $X^{q^{t+n}}Y^{q^s} = X^{q^t}Y^{q^s}+ (X^{q^n} - X)^{q^t}Y^{q^{s}}$, thus $\overline{X^{q^{t+n}}Y^{q^s}} = \overline{X^{q^{t}}Y^{q^s}}$, and likewise $\overline{X^{q^{t}}Y^{q^{s+n}}} = \overline{X^{q^{t}}Y^{q^s}}$. By induction, one can check that $\overline{X^{q^a}Y^{q^b}} = \overline{X^{q^{a\!\!\mod n}}Y^{q^{b\!\!\mod n}}}$ for each $a,b\in {\N_0}$. Therefore, the set $\{\overline{X^{q^{ i_1}}Y^{q^{ i_2}}} :i \in \llbracket 0,n-1 \rrbracket^2\}$ spans the vector space $\mathscr{M}_{q,\mathbb{F}_{q^n}}[X,Y]/U_{X,Y}$. 
     It is straightforward to see that any non-zero polynomial in
     $U_{X,Y}$ has partial $q$-degree along $X$ or $Y$ at least $n$ and hence
     for any non-trivial $S\subset \llbracket 0,n-1 \rrbracket^2$ we have $\sum_{i \in S } f_i\overline{X^{q^{i_1}}Y^{q^{i_2}}} \notin U_{X,Y}$ unless $f_i=0$ for each $i$. Therefore, the set $\{\overline{X^{q^{ i_1}}Y^{q^{ i_2}}} :i \in \llbracket 0,n-1 \rrbracket^2\}$ is linearly independent over $\F_{q^n}$.

\end{proof}

This establishes that any bilinearised $q$-polynomial has a unique expression modulo $U_{X,Y}$ whose monomials terms each have partial $q$-degrees less than $n$. Consequently, we can define the \textbf{support} of any element  $\overline{f(X,Y)}$ in the quotient space $\mathscr{M}_{q,\mathbb{F}_{q^n}}[X,Y]/U_{X,Y}$, denoted by $\Supp(\overline{f(X,Y)})$, as the support of this unique representative. Hence, this support is necessarily a subset of $\llbracket 0,n-1\rrbracket^2$.

Since $U_{X,Y}$ is contained in the ideal of the ring $\F_{q^n}[X,Y]$ generated by $\{X^{q^n}-X,Y^{q^n}-Y\}$, the evaluation of a bilinearised $q$-polynomial at any element of $\F_{q^n}^2$ depends only on its corresponding coset in the quotient space.

\begin{Prop}\label{prop:isomPolyAndTensors}
    For each $i \in \llbracket 1,n\rrbracket^2$, denote by $E_{i} = (\ind{\{\tilde i = i\}})_{\tilde i \in \llbracket 1,n\rrbracket^2}$ the elementary matrix with a single non-zero entry at $i = (i_1,i_2)$. Recall that $\alpha$ denotes an $\F_q$--basis of $\F_{q^n}$ an $M(\alpha)$ is the Moore matrix introduced in Definition~\ref{def:MooreMatrixDefinition}. Then the linear maps $I_1$ and $I_2$ uniquely defined by 
    $$\begin{array}{crcll}
         I_1: & \mathscr{M}_{q,\mathbb{F}_{q^n}}[X,Y]/U_{X,Y} &\longrightarrow& \mathbb{F}_{q^n}^{n \times n}, &\forall i \in \llbracket 1,n\rrbracket^2 : I_1\left(\overline{X^{q^{ i_1-1}}Y^{q^{ i_2-1}}}\right) = E_i\\
         I_2: &  \mathbb{F}_{q^n}^{n \times n}& \longrightarrow&  \mathbb{F}_{q^n}^{n \times n},& \forall K \in \mathbb{F}_{q^n}^{n\times n} : I_2(K) = M(\alpha)^\top K M(\alpha)
    \end{array} $$
    are isomorphisms of $\mathbb{F}_{q^n}$-linear spaces. In particular, the evaluation map
    $$\ev:\mathscr{M}_{q,\mathbb{F}_{q^n}}[X,Y]/U_{X,Y} \to \F_{q^n}^{n\times n}, \overline{f(X,Y)}\mapsto (f(\alpha_{i_1},\alpha_{i_2}))_{i\in \llbracket 1,n\rrbracket^2}$$ is an isomorphism of $\mathbb{F}_{q^n}$-linear spaces with $\ev = I_2 \circ I_1$.
\end{Prop}
\begin{proof}
    The maps $I_1$ and $I_2$ are $\F_{q^n}$-linear by construction. The map $I_1$ is bijective as the image of a basis is also a basis by Lemma~\ref{lemma:BasisOfQutientSpace}, and the map $I_2$ is too as the Moore matrices are invertible. 
    Then for each $i\in\llbracket 1,n\rrbracket^2$ we have
    \[
        (I_2\circ I_1)\left(\overline{X^{q^{ i_1-1}}Y^{q^{ i_2-1}}}\right)
        = M(\alpha)^\top E_i M(\alpha) 
        = \left(\alpha_{j_1}^{q^{i_1-1}}\alpha_{j_2}^{q^{i_2-1}}\right)_{j \in \llbracket 1,n\rrbracket^2}
        = \ev\left(X^{q^{ i_1-1}}Y^{q^{ i_2-1}}\right),
    \]
    which concludes the proof.

\end{proof}

\begin{Ex}
    Let $n = 3$ and let $\F_8 = \F_2(b)$ for $b$ satisfying $b^3 = b+1$. Let $a=b^3$. Then $\{a,a^2,a^4\}$ is a normal $\F_2$-basis of $\F_8$. 
    Let $f(X,Y) = a X^{2}Y^{4} + XY^{2} + a^2X^{2}Y$. Then with respect to the notation above we have:
    $$I_1(\overline{f(X,Y)}) = \begin{bmatrix}
        0 & 1 & 0\\
        a^2&0&a \\
        0&0&0
    \end{bmatrix} \qquad \text{and} \qquad ev(\overline{f(X,Y)}) = \begin{bmatrix}
        b^4 & b^5 &1 \\
        b^4 & 1&b^4 \\
        1 & 0 & b^5
    \end{bmatrix}.$$
\end{Ex}

\subsection{Tensor codes}
\label{subsec:CodeRoth} 

In \cite{ROTHTensorCodesForRankMetric}, Roth introduced and studied a family of codes that can be expressed as the image of the evaluation map $\ev$ defined in Proposition~\ref{prop:isomPolyAndTensors} of a well-chosen subspace of $\mathscr{M}_{q,\mathbb{F}_{q^n}}[X,Y]/U_{X,Y}$. We recall the definition of these codes and some of their properties and then continue with a generalisation of this family. Note that in \cite{ROTHTensorCodesForRankMetric}, the codes were introduced with respect to three (possibly different) bases of $\F_{q^n}$. Since the natural action of $\GL_n(\F_q)^3$ on $(\F_q^n)^{\otimes 3}$ give rise to a tensor-rank preserving automorphism of $(\F_q^n)^{\otimes 3}$ that maps a such a code to a code with the same parameters defined with the same identical bases, we will only consider here the case of a single defining basis $\alpha$. For each $\mu \in \llbracket 0,2n-1\rrbracket$, the code 
$\CCC(n,\mu,3;q)$ is defined to be the $\F_q$-linear subspace of $\F_{q^n}^{n\times n \times n}$ given by 
$$\CCC(n,\mu,3;q) := \left\{ \Gamma \in (\mathbb{F}_q^{n})^{\otimes 3} \ \left| \ \forall r \in \mathcal{S}(n,\mu,3;q) : \sum_{ i \in \llbracket 1,n \rrbracket^3} \Gamma[i] (\alpha^{\bot}_{i_1})^{q^{ r_1}}(\alpha^{\bot}_{i_2})^{q^{ r_2}}\alpha_{i_3} = 0 \right. \right\},$$
where the set $\mathcal{S}(n,\mu,3;q)$ is the subset of $\{r \in \llbracket 0,n-1\rrbracket^{2} \ | \ r_1 + r_{2} \leq \mu -2  \}$ for which there exists an $\F_q$-$[\mu-1,r_1+1,r_2+1]$ Hamming-metric code; 
see \cite[Section~3.1]{ROTHTensorCodesForRankMetric}. We denote by 
$\SSS^{(Roth)}_{\mu} = \llbracket 0,n-1\rrbracket^2 \setminus \SSS(n,\mu,3;q)$ the complement of the set $\mathcal{S}(n,\mu,3;q)$ in $\llbracket 0,n-1\rrbracket^2$. Then we have
\[
    \Sss_{\alpha}\left(\CCC(n,\mu,3;q)\right) = \ev\left(\Span_{\F_{q^n}} \left\{ \overline{X^{q^{s_1}}Y^{q^{s_2}}} \mid (i_1,i_2) \in \SSS^{(Roth)}_{\mu}\right\}\right).
\]

Roth showed that $\CCC(n,\mu,3;q)$ has redundancy $n|\SSS(n,\mu,3;q)| \leq \left(\begin{smallmatrix}  \mu\\2 \end{smallmatrix}\right)n$ as an $\F_q$-vector space, and hence $\F_q$-dimension $n|\SSS_\mu^{(Roth)}| \geq n^3 - n\left(\begin{smallmatrix} \mu\\2 \end{smallmatrix}\right)$. Moreover, he proved that that any element in $\CCC(n,\mu,3;q)$ has tensor-rank at least $\mu$ and \cite[Theorem 5]{ROTHTensorCodesForRankMetric}. Along with this lower bound on the tensor-rank, in \cite{ROTHTensorCodesForRankMetric} two decoding algorithms for the codes $\CCC(n,3,3;q)$ and $\CCC(n,5,3;q)$ are given that correct errors of tensor-rank at most one and two, respectively. There is a generalisation of the second algorithm to correct errors of higher tensor-rank for suitable codes. We will discuss the complexity of these algorithms in Section~\ref{sec:Comparison}. We now study the properties of similar codes defined with this second characterisation as evaluation codes of $q$-polynomials with support in an arbitrary subset of $\llbracket 0,n-1\rrbracket^2$.

We now consider the subspace of polynomials with a specified support in the quotient vector space, as well as its image under the isomorphism described in Proposition~\ref{prop:isomPolyAndTensors}, and study their properties. Throughout this section, let $\mathcal S$ be a subset of $\llbracket 0,n-1\rrbracket^2$.

\begin{Def} \label{def:DefnCodePolyExpression}
   We define the \textbf{(polynomial) Roth-tensor code} associated to $\mathcal{S}$ to be the  $\mathbb{F}_{q^n}$-vector subspace of $\mathscr{M}_{q,\mathbb{F}_{q^n}}[X,Y]/U_{X,Y}$
    $$\mathcal{C}(\mathcal{S}) := \left\{ \left.  \overline{f(X,Y)} \ \right| \ \overline{f(X,Y)} \in \mathscr{M}_{q,\mathbb{F}_{q^n}}[X,Y]/U_{X,Y} , \Supp(\overline{f(X,Y)}) \subseteq  \mathcal{S}\right\}.$$

    We define the \textbf{(evaluation) Roth-tensor code} associated to $\mathcal{S}$ and $\alpha$ to be the $\mathbb{F}_{q^n}$-linear code in $\mathbb{F}_{q^n}^{n\times n}$:
    $$\mathcal{C}_{\alpha}(\mathcal{S}) := \left\{ \left. \big(f(\alpha_{i_1},\alpha_{i_{2}}) \big)_{{i} \in \llbracket 1,n\rrbracket^2} \right| \begin{array}{l} f \in \mathcal{M}_{n,\mathbb{F}_{q^n}}[X,Y], \\\Supp(f) \subseteq  \mathcal{S}\end{array}\right\}.$$
\end{Def}

It follows immediately from the definitions of $\CCC(\SSS)$ and $\CCC_\alpha(\SSS)$ and Proposition~\ref{prop:isomPolyAndTensors} that $\ev(\CCC(\SSS)) = \CCC_\alpha(\SSS)$.

Furthermore, we observe that every column and row of any codeword of $\CCC_\alpha(\SSS)$ is a codeword of a Gabidulin code, as stated in the following proposition.

\begin{Prop}
\label{prop:GabidulinColumns}
    The $\mathbb{F}_{q^n}$-vector spaces $\mathcal{C}_\alpha(\mathcal{S})$ and $\mathcal{C}(\mathcal{S})$ are isomorphic and have $\mathbb{F}_{q^n}$-dimension $|\mathcal{S}|$. Moreover, the maps 
    
    $$\phi_2 : \left\{
        \begin{array}{rcl}    
            \mathcal{C}_\alpha(\mathcal{S}) & \longrightarrow& \Big( \GG_{\max(\pi_2(\mathcal{S}))+1}(\alpha) \Big)^{n} \\[0.2cm]
            C & \longmapsto & \Big( C[i_1,:] \Big)_{i_1 \in \llbracket 1,n\rrbracket}
        \end{array}\right. \quad \text{and} \quad 
    \phi_1 : \left\{
        \begin{array}{rcl}    
            \mathcal{C}_\alpha(\mathcal{S}) & \longrightarrow& \Big( \GG_{\max(\pi_1(\mathcal{S}))+1}(\alpha) \Big)^{n} \\[0.2cm]
            C & \longmapsto & \Big( C[:,i_2] \Big)_{i_2 \in \llbracket 1,n\rrbracket}
         \end{array}\right.$$

are $\mathbb{F}_{q^n}$-vector space monomorphisms, such that for each $j\in \{1,2\}$, $\phi_j$ is an isomorphism if and only if  $\mathcal{S} = \llbracket 0,n-1 \rrbracket \times \llbracket 0,\max(\pi_j(\mathcal{S})) \rrbracket$.
\end{Prop}

\begin{proof}
The monomials $\{X^{q^{s_1}}Y^{q^{s_2}}: s \in \mathcal{S}\}$ are linearly independent over $\F_{q^n}$ and generate $\CCC(\SSS)$ as an $\F_{q^n}$-vector space, thus $\dim_{\F_{q^n}}\CCC(\SSS) = |\mathcal{S}|$. Since $\mathcal{C}(\mathcal{S})$ and $\mathcal{C}_\alpha(\mathcal{S})$ are isomorphic, we have $\dim_{\F_{q^n}}\mathcal{C}_\alpha(\mathcal{S}) = |\SSS|$. We only need to prove the statement for $\phi_2$. For each $q$-polynomial $f(X,Y) \in \mathscr{M}_{q,\mathbb{F}_{q^n}}[X,Y]$ such that $\Supp(f(X,Y)) \subseteq  \mathcal{S}$ and each $i_1 \in \llbracket 1,n\rrbracket$ the polynomial $f(\alpha_{i_1},Y)$ is a linearised $q$-polynomial with $q$-degree (if non-zero) bounded from above by $\max (\pi_2(\SSS))$. Let $C := \big(f(\alpha_{i_1},\alpha_{i_2}) \big)_{{i} \in \llbracket 1,n\rrbracket^2}$ be its associated codeword in $\mathcal{C}_\alpha(\mathcal{S})$. Then for each $i_1 \in \llbracket 1,n\rrbracket$, $C[i_1,:]$ is the evaluation at $\alpha$ of $f(\alpha_{i_1},X)$ and thus is an element of $\GG_{\max(\pi_2(\SSS))+1}(\alpha)$. Hence, the map $\phi_2$ is well-defined.  Moreover, the map $\phi_2$ is $\F_{q^n}$-linear and injective since it is a restriction of the isomorphism  $\F_{q^n}^{n\times n} \to (\F_{q^n}^n)^n$ that splits any matrix into its rows. Hence, it is an isomorphism if and only if $|\mathcal{S}| = n(\max(\pi_2(\mathcal{S})) + 1)$, in other words, if and only  $\mathcal{S} = \llbracket 0,n-1 \rrbracket \times \llbracket 0,\max(\pi_2(\mathcal{S})) \rrbracket$.
\end{proof}

\begin{Corol}\label{corol:GabidulinTwoSidesAndTensor}
    Let $\mu_1,\mu_2 \in \llbracket 0,n-1\rrbracket$ be integers. We have
    $$\CCC_\alpha(\llbracket 0,\mu_1\rrbracket \times \llbracket 0,\mu_2\rrbracket) = \left\{ C \in \F_{q^n}^{n\times n} \ \left| \  \forall i \in \llbracket 1,n\rrbracket^2: C[i_1,:] \in  \GG_{\mu_2+1}(\alpha) \text{ and } C[:,i_2] \in  \GG_{\mu_1+1}(\alpha)  \right.  \right\},$$
    and we have the isomorphism of $\F_{q^n}$-vector spaces 
    $$\CCC_\alpha(\llbracket 0,\mu_1\rrbracket \times \llbracket 0,\mu_2\rrbracket) \simeq \GG_{\mu_2+1}(\alpha) \otimes_{\F_{q^n}} \GG_{\mu_1+1}(\alpha).$$
\end{Corol}
\begin{proof}
    Let $\tilde \CCC=\left\{ C \in \F_{q^n}^{n\times n} \ \left| \  \forall i \in \llbracket 1,n\rrbracket: C[i_1,:] \in  \GG_{\mu_2+1}(\alpha) \text{ and } C[:,i_2] \in  \GG_{\mu_1+1}(\alpha)  \right.  \right\}$. Set $\SSS=\llbracket 0,\mu_1\rrbracket \times \llbracket 0,\mu_2\rrbracket$. By definition we have $\tilde \CCC = \GG_{\mu_2+1}(\alpha) \otimes_{\F_{q^n}} \GG_{\mu_1+1}(\alpha)$. By Proposition~\ref{prop:GabidulinColumns}, $\CCC_\alpha(\SSS)\subseteq\tilde \CCC$, and since 
    \[
        \dim_{\F_q} \tilde \CCC = \dim_{\F_q}\left(\GG_{\mu_2+1}(\alpha)\right) \cdot \dim_{\F_q}\left(\GG_{\mu_1+1}(\alpha)\right) = (\mu_2 +1)(\mu+1 +1) = |\mathcal S| = \dim_{\F_q} \CCC_\alpha(\SSS),
    \]
    we have the wanted result.
\end{proof}

We recall the definition of the dual of a matrix code:

\begin{Def}
The \emph{dual} code of a matrix code 
$\CCC \subseteq \F_{q^n}^{n\times n}$, is defined to be
$$ \CCC^\perp:= \{ X \in \F_{q^n}^{n\times n} ~|~ \forall Y \in \CCC,\  \tr(XY^\top) = 0\}.$$
\end{Def}

As in \cite[Section~3.2]{ROTHTensorCodesForRankMetric}, we have an expression of the dual of $\mathcal{C}_\alpha(\mathcal{S})$ in terms of parity check equations. That is, 
$$\mathcal{C}_\alpha(\mathcal{S})^\perp= \left\{ \Gamma \in (\mathbb{F}_q^{n})^{\otimes 3} \ ~\left|~ \ \forall r \in \llbracket 0,n-1\rrbracket^2 \backslash \mathcal{S} ~|~ \sum_{ i \in \llbracket 1,n \rrbracket^3} \Gamma[i] (\alpha^{\bot}_{i_1})^{q^{ r_1}}(\alpha^{\bot}_{i_2})^{q^{ r_2}}\alpha_{i_3} = 0 \right. \right\}.$$
This expression is the consequence of the two following facts. 
    \begin{enumerate}
    \item For each $r \in \llbracket 0,n-1\rrbracket^2$, consider the matrix $M_r := \alpha^{q^{r_1}} \otimes \alpha^{q^{r_2}} \in \F_{q^n}^{n\times n}$, i.e. the matrix whose 
    $(i_1,i_2)$-component is $M_{r,(i_1,i_2)} =\alpha_{i_1}^{q^{r_1}}\alpha_{i_2}^{q^{r_2}}$ for each $i_1,i_2 \in  \llbracket 1,n\rrbracket$. Then the set $\{M_{r}: r\in \llbracket0,n-1\rrbracket^2\}$ is an $\F_{q^n}$-basis of $\F_{q^n}^{n\times n}$; see \cite[Lemma 4]{ROTHTensorCodesForRankMetric}.
    \item Let $\alpha^\bot$ be the dual basis of $\alpha$, so that $\sum_{\ell = 0}^{n-1} (\alpha_i \alpha^\bot_j)^{q^\ell} = \ind{i=j}$ for each $i,j \in \llbracket 1,n\rrbracket$. This is equivalent to stating that $M(\alpha)$ is invertible and has $M(\alpha^\bot)^\top$ as its inverse. Hence, one also has $\sum_{i = 1}^{n} \alpha_i^{q^u} (\alpha^\bot_i)^{q^v} = \ind{u=v}$ for each $u,v \in \llbracket 0,n-1\rrbracket$.
    \end{enumerate}

These two properties let us state that the dual code of an evaluated Roth-tensor code is also an evaluated Roth-tensor code evaluated at the dual basis; see \cite{RavagnaniRankMetricCodes}. 

\begin{Prop} 
    We have $\mathcal{C}_{\alpha}(\mathcal{S})^\bot = \mathcal{C}_{\alpha^\bot}(\llbracket 0,n-1 \rrbracket^2 \backslash S)$.
\end{Prop}
\begin{proof}
Since $\CCC^\perp$ is a dual relatively to the canonical dot product in $\F_{q^n}^{n\times n}$, we have $\dim_{\F_{q^n}}\CCC_\alpha(\mathcal S)^{\perp} = n^2 - |\mathcal S| = \dim_{\F_{q^n}} \mathcal{C}_{\alpha^\bot}(\llbracket 0,n-1 \rrbracket^2 \backslash S)$, therefore one inclusion suffices to conclude. Let $r = (r_1,r_2) \in \llbracket 0,n-1\rrbracket^2 \backslash \SSS$. Then for each $s = (s_1,s_2) \in \mathcal{S}$, if we denote by $D = \left((\alpha_{i_1}^\perp)^{q^{r_1}}(\alpha_{i_2}^{\perp})^{q^{r_2}}\right)_{i \in \llbracket 1,n\rrbracket^2}$ and $C = \left(\alpha_{i_1}^{q^{s_1}}\alpha_{i_2}^{q^{s_2}}\right)_{i \in \llbracket 1,n\rrbracket^2}$ the respective codewords of $\mathcal{C}_{\alpha^\bot}(\llbracket 0,n-1 \rrbracket^2 \backslash S)$ and $\mathcal C_\alpha(\SSS)$ associated to the monomials $X^{q^{r_1}}Y^{q^{r_2}}$ and $X^{q^{s_1}}Y^{q^{s_2}}$, we have
\begin{align*}
    \tr(CD^\top) &= \sum_{i \in \llbracket 1,n\rrbracket^2} \alpha_{i_1}^{q^{s_1}}\alpha_{i_2}^{q^{s_2}}(\alpha_{i_1}^\perp)^{q^{r_1}}(\alpha_{i_2}^{\perp})^{q^{r_2}}\\
    &= \left(\sum_{i_1 = 1}^n \alpha_{i_1}^{q^{s_1}}(\alpha_{i_1}^\perp)^{q^{r_1}}\right) \left(\sum_{i_2 = 1}^n \alpha_{i_2}^{q^{s_2}}(\alpha_{i_2}^\perp)^{q^{r_2}}\right)\\
    &= \ind{r_1 = s_1}\ind{r_2 = s_2}\\
    &= 0,
\end{align*}
since $r \not\in \mathcal S$ by assumption . Therefore we have $\mathcal{C}_{\alpha}(\mathcal{S})^\bot \supseteq \mathcal{C}_{\alpha^\bot}(\llbracket 0,n-1 \rrbracket^2 \backslash \mathcal{S})$ which concludes the proof.
\end{proof}

\begin{Rq}
    In particular, $\CCC_\alpha(\SSS)^\bot$ and $\CCC_{\alpha}(\llbracket 0,n-1\rrbracket^2\backslash \SSS)$ are equivalent as matrix $\F_{q^n}\tm[n\times n,k]$ rank-metric codes, 
    \emph{i.e.} there exist invertible matrices $P,Q \in \GL_n(\F_{q^n})$ such that $\CCC_\alpha(\SSS)^\bot = P\CCC_{\alpha}(\llbracket 0,n-1\rrbracket^2\backslash \SSS)Q$. This is a direct consequence of the form of the isomorphism $I_2$ in Proposition~\ref{prop:isomPolyAndTensors}.
    
\end{Rq}

Given a basis $\alpha$ of $\F_{q^n}/\F_q$, the families of Roth-tensor codes parametrised with the sets $\SSS$ are ascending families of codes with respect to inclusion. In other words, we have the following statement.

\begin{Lemma}
     Let $\SSS,\SSS' \subseteq \llbracket 0,n-1\rrbracket$ with $\SSS \subseteq \SSS'$. Then $\CCC(\SSS) \subseteq \CCC(\SSS')$ and hence $\CCC_\alpha(\SSS) \subseteq \CCC_\alpha(\SSS')$.
\end{Lemma}
\begin{proof}
    For each $\overline{f(X,Y)} \in \mathscr{M}_{q,\mathbb{F}_{q^n}}[X,Y]/U_{X,Y}$ such that $\Supp(\overline{f(X,Y)}) \subseteq \SSS$, we have that $\Supp(\overline{f(X,Y)}) \subseteq \SSS'$, thus $\CCC(\SSS) \subseteq \CCC(\SSS')$. 
\end{proof}

\begin{Prop}\label{prop:TranslationSets}
    Let $r_1,r_2 \in \Z$. Let $\mathcal{S}' = \{((s_1 + r_1) \mod n, (s_2 + r_2) \mod n) \ | \ (s_1,s_2) \in \mathcal{S}\}$. Then there exist matrices $L,R \in \GL(\F_q)$ such that $L \CCC_\alpha(\mathcal{S}) R = \CCC_\alpha(\mathcal{S}')$.
\end{Prop}
\begin{proof}
    The elements of $\CCC_\alpha(\mathcal{S}')$ are the evaluations at $\alpha$ of the polynomials with monomials of the form $X^{q^{(s_1 + r_1) \mod n}}Y^{q^{(s_2 + r_2) \mod n}} = (X^{q^{r_1}})^{q^{s_1}} (Y^{q^{r_2}})^{q^{s_2}}$. Therefore, with the notations of Proposition~\ref{prop:isomPolyAndTensors}, the code $\CCC_\alpha(\mathcal{S}')$ is the image of $I_1(\CCC(\SSS))$ under the linear map $K \mapsto M(\alpha^{q^{r_1}})^\top K M(\alpha^{q^{r_2}})$. Since the Frobenius map is an $\F_q$-automorphism of $\F_{q^n}$, then $\alpha^{q^{r_1}}$ forms an $\F_q$-basis of $\F_{q^n}$ and there exists a matrix $A \in \GL_n(\F_q)$ such that $(\alpha_1^{q^{r_1}},\dots,\alpha_n^{q^{r_1}}) = (\alpha_1,\dots,\alpha_n)A$. In particular we have $M(\alpha^{q^{r_1}}) = M(\alpha)A$. Likewise there exists a matrix $B\in\GL_n(\F_q)$ such that $M(\alpha^{q^{r_2}}) = M(\alpha)B$ and therefore
        \begin{align*}
        \CCC_\alpha(\mathcal{S}') &= \{  (M(\alpha)A)^\top K M(\alpha) B \ | \ K \in I_1(\CCC(\SSS))\} \\
        &= A^\top (I_2 \circ I_1)(\CCC(\SSS)) B \\
        &= A^\top \CCC_\alpha(\mathcal{S}) B.
    \end{align*}
    
\end{proof}

\subsection{Properties of bilinearised $q$-polynomials}

\begin{Lemma} \label{lemma:CoefCompositionVof}
Let $\mathcal{S}_V \subseteq \llbracket 0,n-1\rrbracket$ and let 
$\mathcal{S}_f  \subseteq \N_0^2$.
Let $V(Z)= \sum_{\ell \in \mathcal{S}_V} v_\ell Z^{q^\ell} \in \mathscr{M}_{q,\mathbb{F}_{q^n}}[Z]$ 
and let $f(X,Y)= \sum_{ s \in \mathcal{S}_F} f_{ s} X^{q^{ s_1}}Y^{q^{ s_2}} \in \mathscr{M}_{q,\mathbb{F}_{q^n}}[X,Y]$. 
Then $\Supp(V(f(X,Y)))\subseteq \mathcal{S}_f + \mathcal{S}_V$ and for each $t \in \mathcal{S}_f + \mathcal{S}_V$, the coefficient in $V(f(X,Y))$ of the monomial $X^{q^{ t_1}}Y^{q^{ t_2}}$ is given by
$\displaystyle \sum_{\substack{\ell \in \mathcal{S}_V\\  t -\ell \in \mathcal{S}_f}} v_\ell f_{ t - \ell}^{q^\ell}. $
\end{Lemma}
\begin{proof}
    Observe that there is a bijection between the sets $\mathcal{S}_f \times \mathcal{S}_V$ and $\{( t, \ell) \in (\mathcal{S}_f + \mathcal{S}_V)\times \mathcal{S}_V \ | \  t - \ell \in \mathcal{S}_f\}$. We have 
    $$V(f(X,Y)) = \sum_{\ell \in \mathcal{S}_V}\sum_{ s \in \mathcal{S}_f} v_\ell f_{ s}^{q^\ell}X^{q^{ s_1+\ell}} Y^{q^{ s_2+\ell}} = \sum_{ t \in \mathcal{S}_f + \mathcal{S}_{V}} \sum_{\substack{\ell \in \mathcal{S}_V\\  t -\ell \in \mathcal{S}_f}} v_\ell f_{ t - \ell}^{q^\ell}X^{q^{ t_1}} Y^{q^{ t_2}} , $$
    from which the statement follows.
\end{proof}

For each polynomial $F(X,Y)\in \MM_{q,\F_{q^n}}[X,Y]$, there exist unique linearized polynomials $P_\ell(Z), Q_\ell(Z) \in \MM_{q,\F_{q^n}}[Z]$ such that $F(X,Y) = \sum_{{\ell}  \in {\N_0} } P_{{\ell}}(X)Y^{q^{ \ell}}= \sum_{{\ell}  \in {\N_0} } X^{q^{ \ell}}Q_{{\ell}}(Y)$.

\begin{Prop} \label{prop:RadicalPolynomialMultivariable}
    Let $F(X,Y)\in \mathscr{M}_{q,\mathbb{F}_{q^n}}[X,Y]$. 
    For each $\ell \in \N_0$, let $P_\ell(Z), Q_\ell(Z) \in \MM_{q,\F_{q^n}}[Z]$ be such that
    \begin{equation}\label{eq:UniqueWritingMultilinearPolyFactorVariable}
        F(X,Y) = \sum_{{\ell}  \in {\N_0} } P_{{\ell}}(X)Y^{q^{ \ell}}= \sum_{{\ell}  \in {\N_0} } X^{q^{ \ell}}Q_{{\ell}}(Y).
    \end{equation}
     
     Then $\displaystyle\bigcap_{{\ell} \in {\N_0}} \ker P_{{\ell}} \subseteq \Rr\Aaa\Ddd_1(F)$ and we have equality if $\deg_{Y} F(X,Y) < q^n$ . Similarly, $\displaystyle\bigcap_{{\ell} \in {\N_0}} \ker Q_{{\ell}} \subseteq \Rr\Aaa\Ddd_2(F)$ and we have equality if $\deg_{X} F(X,Y) < q^n$ .
\end{Prop}

\begin{proof}
  Let $x\in\bigcap_{ \ell \in {\N_0}} \ker P_\ell$. For each $y \in \mathbb{F}_{q^n}$, we have $F(x,y) = 0$, therefore $x \in \Rr\Aaa\Ddd_1(F)$. Conversely, let $x \in \Rr\Aaa\Ddd_1(F)$ and assume that $\deg_{Y} F(X,Y) < q^n$. Then the polynomial $F(x,Y)$ is a linearised $q$-polynomial of $q$-degree less than $n$ that annihilates every element of $\mathbb{F}_{q^n}$, thus $F(x,Y) = 0$. Hence, every coefficient of $F(x,Y)$ is zero \emph{i.e.} $P_{{\ell}}(x) = 0$ for each $ \ell \in {\mathbb{N}_0}$ and so $x \in \bigcap_{{\ell} \in{\N_0}} \ker P_{ \ell}$. A similar argument applied to the polynomials $Q_\ell, \ell \in {\N_0}$ completes the proof.
\end{proof}

\begin{Corol} \label{corol:CorolRadicalPolynomialMulivar}
    Let $F(X,Y)\in \mathscr{M}_{q,\mathbb{F}_{q^n}}[X,Y]$ be non-zero. 
    If $q$-$\deg_{Y} F(X,Y) < n$, then 
    $\dim_{\mathbb{F}_q} \Rr\Aaa\Ddd_1(F) \leq \qdeg_{X} F(X,Y)$. Similarly, if $q$-$\deg_{X} F(X,Y) < n$, then $\dim_{\mathbb{F}_q} \Rr\Aaa\Ddd_2(F) \leq \qdeg_{Y} F(X,Y)$.
\end{Corol}

\begin{proof}
     Assume that $q$-$\deg_{Y} F(X,Y) < n$. Let $P_{{\ell}}(X)$ be the unique linearised $q$-polynomials over $\mathbb{F}_{q^n}$  satisfying (\ref{eq:UniqueWritingMultilinearPolyFactorVariable}). For each $\ell \in {\mathbb{N}_0}$, denote by $P_{{\ell}}$ the corresponding $\mathbb{F}_q$-linear endomorphism of $\mathbb{F}_{q^n}$. Proposition~\ref{prop:RadicalPolynomialMultivariable} ensures that $\bigcap_{ \ell \in {\N_0}} \ker P_{ \ell} = \Rr\Aaa\Ddd_1(F)$, hence we have:

    $$\dim_{\mathbb{F}_q}\Rr\Aaa\Ddd_1(F) = \dim_{\mathbb{F}_q}\bigcap_{ \ell \in {\N_0}} \ker P_{ \ell} \leq \min_{ \ell \in {\N_0}} \dim_{\mathbb{F}_q} \ker P_{ \ell}  \leq \min_{\substack{ \ell \in {\N_0}\\ P_{ \ell}(X) \neq 0}} \qdeg P_{ \ell}(X).$$

    Since $F(X,Y)$ is non-zero, there exists $\ell \in {\mathbb{N}_0}$ such that $P_{ \ell}(X) \neq 0$. Since $\qdeg_{X}P_{ \ell}(X)$ is bounded from above by $\qdeg_{X} F(X,Y)$, we have $\dim_{\mathbb{F}_q} \Rr\Aaa\Ddd_1(F) \leq \qdeg_{X} F(X,Y)$.
    The proof of the second statement is analogous.
\end{proof}

\subsection{Distance functions}
\label{subsec:distancesinthecode}

\begin{Def}
    We define the \textbf{fibre weight} of a matrix $T \in \mathbb{F}_{q^n}^{n\times n}$, denoted $\wt_{\fibresp_{3}}(T)$, to be the dimension of the $\mathbb{F}_q$-span of the entries of $T$, in other words, $\wt_{\fibresp_3}(T) = \dim_{\F_q}\Span_{\F_q}(\{T_{i_1,i_2} \ | \ i_1,i_2 \in \llbracket 1,n\rrbracket\})$.
\end{Def}

Let $\sigma$ be a bijection from $\llbracket1,n^2\rrbracket$ to $\llbracket1,n\rrbracket^2$. 
This yields an isomorphism $\hat{\sigma} :\mathbb{F}_{q^n}^{n\times n} \rightarrow \mathbb{F}_{q^n}^{n^2}, M \mapsto (M[\sigma(\ell)])_{\ell \in \llbracket 1,n^2\rrbracket}$ for all $M \in \mathbb{F}_{q^n}^{n\times n}$. Therefore, we have that $\wt_{\fibresp_3}(M) = \rank_{\F_q}(\hat{\sigma}(M))$ for all $M \in \mathbb{F}_{q^n}^{n\times n}$.

In particular, decoding any tensor code with respect to this metric is  equivalent to decoding its isomorphic image in $\mathbb{F}_{q^n}^{n^2}$ for the rank-metric. 

For each $T \in \mathbb{F}_{q^n}^{n\times n}$, we denote by $\UUU_1(T)$ (\emph{resp.} $\UUU_2(T)$) the $\mathbb{F}_q$-subspace of $\mathbb{F}_{q^n}^n$ spanned by all the rows (\emph{resp.} the columns) of $T$ over $\mathbb{F}_q$.

\begin{Def}
    Let $j \in \{1,2\}$. We define the \textbf{$j$-slice space weight} of a tensor $T \in \mathbb{F}_{q^n}^{n\times n}$, denoted $\wt_{\slicesp_j}(T)$, to be the dimension of the $\F_q$-span of the $j$-slices of the associated tensor of $T$, in other words, $\wt_{\slicesp_j}(T) = \dim_{\mathbb{F}_q}\UUU_j(T)$.
\end{Def}

\begin{Prop}\label{prop:correspondancetensormatrix}
    Let $T\in\mathbb{F}_{q^n}^{n\times n}$ and let $\Gamma \in (\mathbb{F}_q^n)^{\otimes 3}$ be its associated tensor with respect to a basis $\omega$ of $\mathbb{F}_{q^n}/\mathbb{F}_{q}$. Then $\wt_{\slicesp_1}(T)$, $\wt_{\slicesp_2}(T)$ and $\wt_{\fibresp_3}(T)$ are exactly the $\mathbb{F}_{q}$-dimensions of $\slicesp_1(\Gamma)$, $\slicesp_2(\Gamma)$ and $\fibresp_3(\Gamma)$ respectively, which do not depend on the choice of $\omega$.
\end{Prop}

\begin{proof}
    Since $\omega$ is a basis of $\mathbb{F}_{q^n}/\mathbb{F}_q$, the maps $\phi,\psi$ defined by  
    \[
    \phi : \left\{\begin{array}{ccc}
    \mathbb{F}_{q}^{n\times n} &\longrightarrow& \mathbb{F}_{q^n}^n\\
    M &\longmapsto& \sum_{k = 1}^n M[:,k]\omega_k
    \end{array}\right. \qquad \text{and} \qquad \psi : \left\{
    \begin{array}{ccc}
    \mathbb{F}_{q}^n &\longrightarrow& \mathbb{F}_{q^n}\\
    u &\longmapsto& \sum_{k = 1}^n u_k\omega_k
    \end{array}\right.
    \]
    are $\mathbb{F}_q$-isomorphisms. Moreover, one can verify that $\UUU_1(T) = \phi(\slicesp_1(\Gamma))$, that $\UUU_2(T) = \phi(\slicesp_2(\Gamma))$ 
    
     and that $\psi(\fibresp_3(\Gamma))$ is exactly the $\mathbb{F}_q$-span of the entries of $T$. Since $\phi$ and $\psi$ are isomorphisms,
    we have that the $\mathbb{F}_q$-dimensions of $\slicesp_i(\Gamma)$ and  $\fibresp_3(\Gamma)$ are preserved under these isomorphisms, respectively.
    Since the weights $\wt_{\slicesp_i}(T)$ and $\wt_{\fibresp_3}(T)$ are independent of $\omega$, then so are the dimensions of $\slicesp_i(\Gamma)$ and $\fibresp_3(\Gamma)$.
    
\end{proof}

We note the following link between these weights and the $\mathbb{F}_q$-rank of the fibres of the tensors.

\begin{Lemma} \label{lemma:maxrankandslicespaceweight}
    Let $T \in \mathbb{F}_{q^n}^{n\times n}$. We have the following properties.
    \begin{enumerate}
        \item The $\mathbb{F}_q$-rank of each column of $T$ is bounded from above  by the dimension of the $\mathbb{F}_q$-span of the rows of $T$. In particular, we have the inequality $\max_{i_2 \in \llbracket 1,n\rrbracket} \rank_{\mathbb{F}_q} T[:,i_2] \leq \wt_{\slicesp_1}(T)$.
        \item The $\mathbb{F}_q$-rank of each row of $T$ is bounded from above by the dimension of the $\mathbb{F}_q$-span of its columns. In particular, we have the inequality $\max_{i_1 \in \llbracket 1,n\rrbracket} \rank_{\mathbb{F}_q} T[i_1,:] \leq \wt_{\slicesp_2}(T)$.
    \end{enumerate}
\end{Lemma}

\begin{proof}
    Let $i_2 \in \llbracket 1,n\rrbracket$. Consider the $\mathbb{F}_q$-linear epimorphism  $p_{i_2} : \mathbb{F}_{q^n}^n \to \mathbb{F}_{q^n}, x \mapsto x_{i_2}$ and for each $i_1 \in \llbracket 1,n\rrbracket$ we have $p_{i_2}(T[i_1,:]) = T_{i_1,i_2}$. Thus $\Span_{\F_q}(\{T_{i_1,i_2} : i_1 \in \llbracket 1,n\rrbracket\}) = p_{i_2}({\mathcal U}_1(T))$. Finally, since $p_{i_2}$ is surjective we have $\rank_{\F_q}  T[:,i_2] = \dim_{{\mathbb F}_q}(p_{i_2}({\mathcal U}_1(T))) \leq \dim_{\mathbb F_q}({\mathcal U}_1(T)) = w_{\slicesp_1}(T)$. 
    
    The second inequality follows similarly.
\end{proof}

These bounds are sharp. 
For instance, consider $T \in \F_{q^n}^{n\times n}$ such that $T[i_1,i_2] = 0$ for each $i \in \llbracket 1,n\rrbracket^2$ with $i_2 \neq 1$, \emph{i.e.} only the first column of $T$ is non-zero. Then, the $\F_q$-span of the entries of $T$, which is $\Span_{\F_q}\{ T[i_1,1] \ | \ i_1 \in \llbracket 1,n\rrbracket \}$, has the same $\F_q$-dimension as $\UUU_1(T) = \Span_{\F_q}\{ (T[i_1,1],0,\dots,0) \ | \ i_1 \in \llbracket 1,n\rrbracket \}$.

Clearly, the aforementioned weights are both subadditive and such that the zero tensor has weight zero. Therefore, for any $j \in \{1,2\}$, the maps $d_{\fibresp_3} : (T_1,T_2)\in (\mathbb{F}_{q^n}^{n\times n} )^2 \mapsto \wt_{\fibresp_{3}}(T_1 - T_2) $ and $d_{\slicesp_j} : (T_1,T_2)\in (\mathbb{F}_{q^n}^{n\times n} )^2 \mapsto \wt_{\slicesp_{j}} (T_1 - T_2) $ are distance functions on $\F_{q^n}^{n\times n}$. If $\CCC$ is a non-zero tensor code, we denote by $d_{\slicesp_j}(\CCC) := \min\{d_{\slicesp_j}(C,C'): C,C' \in \CCC, C \neq C'\}$ and $d_{\fibresp_3}(\CCC) := \min\{d_{\fibresp_3}(C,C'):C,C' \in \CCC, C \neq C'\}$ the \textbf{minimum distances} of the code with respect to those metrics. 

Likewise, since the $\mathbb{F}_q$-rank yields a metric on $\mathbb{F}_{q^n}^n$, the maximum $\mathbb{F}_q$-rank of all columns or all rows of an element in $\mathbb{F}_{q^n}^{n\times n}$ yields a metric on 
$\mathbb{F}_{q^n}^{n\times n}$. More precisely, the maps $d_{\rk_1},d_{\rk_2} : \mathbb{F}_{q^n}^{n\times n} \times \mathbb{F}_{q^n}^{n\times n} \to \N_0$ defined by \begin{align*}
d_{\rk_1}(T_1,T_2) &:=\max\{\rank_{\mathbb{F}_q} (T_1-T_2)[:,i_2] ~|~ i_{2} \in \llbracket 1,n\rrbracket\}\\
\text{and}\qquad  d_{\rk_2}(T_1,T_2) &:=\max\{ \rank_{\mathbb{F}_q} (T_1-T_2)[i_1,:] ~|~ i_{1} \in \llbracket 1,n\rrbracket\},
\end{align*}
 for all $(T_1,T_2) \in  \mathbb{F}_{q^n}^{n\times n} \times \mathbb{F}_{q^n}^{n\times n}$ 
define metrics on the vector space $\mathbb{F}_{q^n}^{n\times n}$.

\begin{Prop}\label{prop:UpperBoundDistance}
    Let $\mathcal{S}$ be a non-empty subset of $\llbracket 0,n -1\rrbracket^2$. Then the code $\CCC = \mathcal{C}_{\alpha}(\mathcal{S})$ has the following properties. 
    \begin{enumerate}
        \item We have $d_{\fibresp_3}(\CCC) \geq n - \max\{s_j \ : \ (s_1,s_2) \in \SSS, j \in \{1,2\}\}$.
        \item We have $d_{\slicesp_j}(\CCC) \geq n - \max \{\pi_j(\mathcal{S})\}$, $j \in \{1,2\}$.
    \end{enumerate}
\end{Prop}

\begin{proof}
Let $C \in \mathcal{C}_{\alpha}(\mathcal{S})$ be a non-zero codeword and let $j \in \{1,2\}$.  Firstly, Proposition~\ref{prop:GabidulinColumns} states that each column (\emph{resp.} row) of $C$ is an element of the MRD code $\GG_{\max \pi_1(\mathcal{S}) +1 }(\alpha)$ (\emph{resp.} $\GG_{\max \pi_2(\mathcal{S}) +1 }(\alpha)$) and thus is either zero or has rank at least $n - \max \pi_1(\mathcal{S})$ (\emph{resp.} $n - \max \pi_2(\mathcal{S})$). Since $C$ is non-zero, there exists at least one such non-zero column (\emph{resp.} row) in the matrix $C$. 
In particular the rank of the space spanned by all the entries in $C$ is at least $n - \max_{j \in\{ 1,2\}}\max \pi_j(\mathcal{S})$. The second inequality is an application of Lemma~\ref{lemma:maxrankandslicespaceweight} and Proposition~\ref{prop:GabidulinColumns} using the same argument.
\end{proof}

In the sequel, we will remark that the fibre and slice weights introduced above all give lower bounds on the tensor-rank. The \textbf{tensor-rank} of an element $\Gamma$ in $(\mathbb{F}_q^{n})^{\otimes 3}$, denoted by $\trank(\Gamma)$, is the least integer $R \in {\mathbb{N}_0}$ such that there exist 
vectors $a_r, b_r, c_r \in \mathbb{F}_q^n$ for each $ r \in \llbracket 1,R\rrbracket$ such that
$$\Gamma = \sum_{r = 1}^R a_r \otimes b_r \otimes c_r$$
in which case the expression above is called a \emph{minimal tensor-rank form} of $\Gamma$.
As a coordinate tensor, we have $\Gamma[i_1,i_2,i_3] = \sum_{r = 1}^R a_{r,i_1}b_{r,i_2}c_{r,i_3}$ for each $i\in\llbracket 1,n\rrbracket^3$. 
This generalises the notion of matrix rank since the rank of a matrix $M$ is the smallest integer $R$ such that $M$ can be expressed as a sum of $R$ rank-$1$ matrices, \emph{i.e.} matrices of the form $a \otimes b = a^\top b$. Note that the map $(\Gamma_1, \Gamma_2) \mapsto \trank(\Gamma_1 - \Gamma_2)$ is a metric on the vector space $(\mathbb{F}_q^n)^{\otimes 3}$.

\begin{Lemma}\label{lemma:tensorrankdecompositionasmatrices}
    Let $\omega$ be a basis of $\mathbb{F}_{q^n}/\mathbb{F}_q$. Let $T \in \mathbb{F}_{q^n}^{n\times n}$ and let $\Gamma \in (\mathbb{F}_{q}^n)^{\otimes 3}$ be the tensor associated to $T$ for the basis $\omega$. Then $\trank(\Gamma)$ does not depend on $\omega$ and is the smallest integer $R \in {\mathbb{N}_0}$ such that there exists $A_1,\dots,A_R \in \mathbb{F}_q^{n\times n}$ of rank one and $\gamma_1,\dots,\gamma_R \in \mathbb{F}_{q^n}$ such that: 
    $$T = \sum_{r = 1}^R \gamma_rA_r.$$
    In particular, for each $T\in \F_{q^n}^{n\times n}$ and each $L,R \in \F_q^{n\times n}$ we have $\trank_{\F_q}(LTR) \leq \trank_{\F_q}(T)$, and this inequality is an equality if $L,R \in \GL_n(\F_q)$.
\end{Lemma}

\begin{proof}
        Let $R \in {\mathbb{N}_0}$. Let $(\omega_1^*,\dots,\omega_n^*)$ be
        the dual basis of $\omega$, \emph{i.e.}
        $\omega_i^*(\omega_j) = \ind{i = j}$ for each $i,j \in \llbracket 1,n\rrbracket$. Let $A_1,\dots,A_R \in \mathbb{F}_q^{n\times n}$ be a collection  of rank-$1$ matrices and let $\gamma_1,\dots,\gamma_R \in \mathbb{F}_{q^n}$ such that $T = \sum_{r = 1}^R A_r \gamma_r$. Let $a_r,b_r, c_r \in \mathbb{F}_{q}^n$ such that for each $r$ we have $A_r = a_r^\top b_r$ and $c_r = (\omega_1^*(\gamma_r),\dots,\omega_n^*(\gamma_r))$. Then $\Gamma = \sum_{r = 1}^R a_r \otimes b_r \otimes c_r$, thus $\trank(\Gamma) \leq R$.     We have the converse inequality by definition of the $\mathbb{F}_q$-isomorphism $(\mathbb{F}_{q}^{n})^{\otimes 3} \to \mathbb{F}_{q^n}^{n\times n}$ associated to the basis $\omega$. Additionally, denote by $L,R \in \F_q^{n\times n}$ two matrices over $\F_q$ and denote by $T\in \F_{q^n}^{n\times n}$ a matrix with a decomposition $T = \sum_{r = 1}^{R} \gamma_rA_r$ with $\gamma_1,\dots,\gamma_R \in \F_q$, with $A_1,\dots,A_R \in \F_q^{n\times n}$ matrices of rank one and with $R$ a positive integer. Then we have $LTC = \sum_{r = 1}^{R} \gamma_{r}(LA_rR)$ and $LA_rR \in \F_{q}^{n\times n}$ has rank at most one for each integer $r \in \llbracket 1,r\rrbracket$. Hence $\trank_{\F_q}(LTR) \leq R$ and in particular $\trank_{\F_q}(LTR) \leq \trank_{\F_q}(T)$, and assuming that $L$ and $R$ are invertible, the other inequality can be obtained symmetrically.
\end{proof}

\begin{Rq}
The second part of the lemma is a particular case of a known result: the group action of $\GL_n(\F_q)^3$ on $(\F_q^n)^{\otimes 3}$ preserves the tensor-rank, as seen in \cite[Section 14.4]{Burgisser1997ch14}.
\end{Rq}

\begin{Def}
    Let $T \in 
    \mathbb{F}_{q^n}^{n\times n}$. 
    The \textbf{tensor-rank} of $T$ is defined to be $\trank_{\mathbb{F}_q}(T):=\trank(\Gamma)$ of its associated tensor $\Gamma \in (\mathbb{F}_q^n)^{\otimes 3}$.
\end{Def}

By abuse of notation, we will denote by $d_{\trank}(\Gamma,\Gamma'):=\trank(\Gamma-\Gamma')$ the tensor-rank distance between any $\Gamma,\Gamma' \in (\F_{q}^n)^{\otimes 3} $ and likewise
$d_{\trank}(T,T'):=\trank_{\mathbb{F}_q}(T-T')$ for any $T,T' \in \F_{q^n}^{n\times n}$. 

\begin{Def}
    Let $\CCC$ be a non-zero $\F_{q^n}\tm[n\times n,k]$ {tensor code}. We define the \textbf{minimum distance} of the code $\CCC$ to be 
    $d_{\trank}(\CCC) := \min\{d_{\trank}(C,C') \;|\; C,C' \in \CCC, C \neq C'\}  = \min\{\trank_{\F_q}(C)\;|\;C \in \CCC, C \neq 0\}$. 
    If $d_{\trank}(\CCC)=d$, we say that $\CCC$ is an $\F_{q^n}\tm[n\times n,k,d]$ {tensor code}.
\end{Def}

\begin{Ex} \label{ex:GabidulinInCode}
    Consider $\mu \in \llbracket 0,n-1\rrbracket$.     For all $i_2 \in \llbracket 1,\mu+1\rrbracket$, consider a linearised $q$-polynomial $\phi_{i_2}(X)$ of $q$-degree $\mu$ such that $\phi_{i_2}(\alpha_{t}) = \ind{i_2=t}$ for each $t \in \llbracket 1,\mu+1\rrbracket$. Such a polynomial exists for one can consider for instance 
    \[\phi_{i_2}(X) =  \prod_{u \in \Span_{\F_q}(\{\alpha_t | i_2 \neq t, 1\leq t \leq \mu +1\})}(X-u)(\alpha_{i_2} - u)^{-1},\] see \cite[Theorem 3.52]{lidl1994introduction}. Then, the polynomial $f(X,Y) = \sum_{i=0}^\mu \phi_{i+1}(X)Y^{q^i}$ is a bilinearised $q$-polynomial such that  $\Supp(f)\subseteq \llbracket 0,\mu\rrbracket^2$ and such that $f(\alpha_{i_1+1},Y) = Y^{q^{i_1}}$ for each $i_1 \in \llbracket 0,\mu\rrbracket$.
    
    Denote by $C_f:= \ev(\overline{f(X,Y)})$ the codeword of $\CCC_\alpha(\llbracket 0,\mu\rrbracket^2)$ associated to $f$, and denote by $C_f' \in \F_{q^n}^{(\mu+1)\times n}$ the matrix comprising the first $\mu + 1$ rows of $C_f$.     Then $C_f'$ is a generator matrix of the $\F_{q^n}$-linear code $\GG_{\mu +1}(\alpha)$. Thus, if we denote by $\Gamma := \Sss_{\omega}^{-1}(C_f) \in (\F_{q}^n)^{\otimes 3}$ and if we denote by $\Gamma' \in \F_{q}^{(\mu+1)} \otimes \F_q^n \otimes \F_{q}^n$ the tensor constructed from the first $\mu+1$ $1$-slices of $\Gamma$, then the tensor $\Gamma'$ is a generator tensor of the matrix code associated to $\GG_{\mu+1}(\alpha)$, in the sense of \cite[Definition 11.1.3]{ConciseCodingTheoryRankMetricCodes} and \cite[Definition 4.1]{ByrneNeriRavagnaniSheekeyTensorRepresentation2019}. 
    Moreover, the surjective function $p:(\F_{q}^n)^{\otimes 3} \to \F_{q}^{(\mu+1)} \otimes \F_q^n \otimes \F_{q}^n$ that associates to any $\Omega \in (\F_{q}^n)^{\otimes 3}$ the tensor constructed from the first $\mu+1$ $1$-slices of $\Omega$, is tensor-rank decreasing. 
    Consequently, $\trank_{\F_q}(C_f)$ is bounded from below by the tensor-rank of the matrix code associated to $\GG_{\mu +1}(\alpha)$, in the sense of \cite[Definition 4.3]{ByrneNeriRavagnaniSheekeyTensorRepresentation2019}. 
\end{Ex}

\begin{Rq}
    Example~\ref{ex:GabidulinInCode} shows that $d_{\trank}( \CCC_\alpha(\llbracket 0,\mu\rrbracket^2))$ is bounded from above by the tensor-rank of the Gabidulin code $\GG_{\mu+1}(\alpha)$. Obtaining the tensor-rank of an element in $(\mathbb{F}_q^n)^{\otimes 3}$ is NP-complete \cite{HASTAD1990644}. The families of tensors for which the tensor-rank is known are few. In particular, computing the minimum tensor-rank of a code $\CCC_\alpha(\SSS)$ is a complicated problem, and Example~\ref{ex:GabidulinInCode} illustrates a link between the minimum tensor-rank of such a code and the tensor-rank of a Gabidulin code, which is not something known, except for particular cases \cite{BYRNE20241}.
\end{Rq}

\begin{Lemma} \label{lemma:SliceSpacefibreSpaceDimension}
   Let $\Gamma \in (\mathbb{F}_q^n)^{\otimes 3}$. Then for each $j \in \llbracket 1,3\rrbracket$ we have $\dim_{\mathbb{F}_{q}} \slicesp_j(\Gamma) = \dim_{\mathbb{F}_q} \fibresp_j(\Gamma)$. 
\end{Lemma}

\begin{proof}
We will prove the property for $j =3$; the case $j=1$ is similar. Let $\eta$ be the least integer such that there exist $M_1,\dots,M_\eta \in \mathbb{F}_q^{n \times n}$ and $c_1,\dots,c_\eta \in \mathbb{F}_q^n$ such that
   \begin{equation}
   \label{MinimalMcWriting}
       \Gamma = \sum_{r =1}^\eta M_r \otimes c_r, \qquad \emph{i.e.} \qquad  \forall  (i,j,k)  \in \llbracket 1,n\rrbracket^{3} : \Gamma[ i,j,k] = \sum_{r = 1}^\eta M_{r,(i,j)}c_{r,k}.
    \end{equation}
   
Let $\lambda_1,\dots,\lambda_\eta \in \mathbb{F}_q$ be scalars such that $\sum_{r = 1}^\eta \lambda_rM_r = 0$. If there exists $\ell \in \llbracket 1,\eta \rrbracket$ such that $\lambda_\ell \neq 0$, then $\Gamma = \sum_{r \in \llbracket 1,\eta \rrbracket \backslash \{\ell\}} M_r \otimes (c_r - \lambda_r\lambda_\ell^{-1}c_{\ell})$, which contradicts the minimality of $\eta$. Therefore, $\lambda_r = 0$ for each $r \in \llbracket 1,\eta \rrbracket$ proving that the $M_1,\dots,M_\eta$ are linearly independent over $\mathbb{F}_q$. Similarly, the $c_1,\dots,c_\eta$ are also linearly independent. We will prove that $\{M_1,\dots,M_\eta\}$ and $\{ c_1,\dots,c_\eta\}$ generate 
$ \slicesp_j(\Gamma)$ and $\fibresp_j(\Gamma)$, respectively.

Since $\Gamma[i,j,:] = \sum_{r = 1}^\eta M_{r,(i,j)}c_r$ for each $i,j\in \llbracket 1,n\rrbracket$, we have the inclusion $\Span_{\mathbb{F}_q}(c_1,\dots,c_\eta) \supseteq \fibresp_{3}(\Gamma)$. Conversely, let $\phi : \llbracket 1,n^2 \rrbracket \to \llbracket 1,n\rrbracket^2$ be a bijection. Since $(M_1,\dots,M_{\eta})$ are linearly independent over $\mathbb{F}_q$, the matrix of size $n^2 \times \eta$ below has $\eta$ linearly independent columns over $\F_q$ and thus has rank $\eta$:
          $$\begin{pmatrix}
             M_{1,\phi(1)} & \hdots &  M_{\eta,\phi(1)}\\
            \vdots & \ddots &  \vdots\\
             M_{1,\phi(n^2)} & \hdots &  M_{\eta,\phi(n^2)}
         \end{pmatrix}.$$
In particular, there exist distinct integers $\ell_1,\dots,\ell_\eta \in \llbracket 1,n^2\rrbracket$  such that $(M_{r,\phi(\ell_\delta)})_{(\delta,r) \in \llbracket 1,\eta\rrbracket^2}$ is a non-singular matrix in $\F_{q}^{\eta \times \eta}$. Hence, there exists an invertible matrix $Q \in \GL_{\eta}(\mathbb{F}_q)$ with $\sum_{\delta = 1}^\eta Q_{i,\delta}M_{r,\phi(\ell_\delta)} = \ind{i = r}$ for each $i,r \in \llbracket 1,\eta\rrbracket$. In particular, for each $i \in \llbracket 1,\eta \rrbracket$, we can write the following:
         $$c_i = \sum_{r = 1}^n \ind{r = i}c_r =  \sum_{\delta = 1}^\eta \sum_{r = 1}^\eta Q_{i,\delta}M_{r,\phi(\ell_\delta)}c_r = \sum_{\delta = 1}^\eta Q_{i,\delta}\Gamma[\phi(\ell_\delta),:].$$
Then $c_1,\dots,c_{\eta}\in \fibresp_{3}(\Gamma)$, which proves the converse inclusion. Similarly, we have $\Span_{\mathbb{F}_q}(M_1,\dots,M_\eta) \supseteq \slicesp_{3}(\Gamma)$ by construction while the converse inclusion is holds by the $\F_q$-linear independence of the $c_1,\dots,c_\eta$. This ensures that there exist distinct $\ell_1,\dots,\ell_\eta \in \llbracket 1,n\rrbracket$ such that $(c_{r,\ell_\delta})_{(\delta,r) \in \llbracket 1,\eta\rrbracket^2}$ is invertible with inverse $P \in \GL_\eta(\mathbb{F}_q)$ say, so that for each $i \in \llbracket 1,\eta\rrbracket$ we have $M_i = \sum_{\delta = 1}^\eta P_{i,\delta}\Gamma[:,:,\ell_\delta]$. Therefore, $\{c_1,\dots,c_\eta\}$ and $\{M_1,\dots,M_\eta\}$ are respective bases of the $\mathbb{F}_q$-vector spaces $\fibresp_{3}(T)$ and $\slicesp_{3}(T)$, which proves the lemma.
\end{proof}

With the lemma and the previous remarks we have the following immediate consequences.

\begin{Prop} \label{prop:weightsandtensorrank}
    For each  $T \in \mathbb{F}_{q^n}^{n\times n}$, we have $\wt_{\fibresp_3}(T) \leq \trank_{\mathbb{F}_q}(T)$ and we have $\wt_{\slicesp_j}(T) \leq \trank_{\mathbb{F}_q}(T)$ for each $j \in \{ 1,2\}$.
\end{Prop}

\begin{proof}
    A consequence of \cite[14.45]{Burgisser1997ch14} is that any slice-space of $\Gamma$ is contained in a space generated by $\trank(\Gamma)$ matrices of rank 1, which proves that the tensor-rank is bounded from below by the slice-space dimension. Then, the proposition is a direct consequence of Proposition~\ref{prop:correspondancetensormatrix} and Lemma~\ref{lemma:SliceSpacefibreSpaceDimension}.
\end{proof}

\begin{Corol} \label{corol:maxrankslicesvstrank}
    For each $T \in \mathbb{F}_{q^n}^{n\times n}$ we have
    $$ \max_{i_1 \in \llbracket 1,n\rrbracket} \rank_{\mathbb{F}_q} T[i_1,:] \leq \trank_{\mathbb{F}_q}(T) \qquad \text{and} \qquad \max_{i_2 \in \llbracket 1,n\rrbracket} \rank_{\mathbb{F}_q} T[:,i_2] \leq \trank_{\mathbb{F}_q}(T). $$
    Hence, if $\mathcal{S} \subseteq \llbracket 0,n-1\rrbracket^2$ is non-empty, then  $d_{\trank}(\CCC_\alpha(\SSS)) \geq n-\max\{s_j \ | \ (s_1,s_2) \in \SSS, j \in \{1,2\}\}$.
\end{Corol}
\begin{proof}
    This is a direct consequence of Proposition~\ref{prop:weightsandtensorrank}, of Lemma~\ref{lemma:maxrankandslicespaceweight} and of Proposition~\ref{prop:UpperBoundDistance}.
\end{proof}

\begin{Ex}
Fix $q = 3$, $n = 2$ and $\omega = (\omega_1,\omega_2)$ a basis of $\mathbb{F}_9/\mathbb{F}_3$. Consider the tensors $\Gamma_1,\Gamma_2,\Gamma_3 \in (\F_3^2)^{\otimes 3}$ defined below and also illustrated in Figure~\ref{fig:IllustrTRANK123}.
\begin{itemize}
    \item Let $\Gamma_1 := (1,1) \otimes (1,2) \otimes (1,2)$.
    \item Let $\Gamma_2 := (1,1) \otimes (1,1) \otimes (1,1) + (0,2) \otimes (1,1) \otimes (1,0)$.
    \item Let $\Gamma_3 := (1,0) \otimes (1,0) \otimes (1,0) + (0,1) \otimes (0,1) \otimes (1,0) + (1,0) \otimes (0,1) \otimes (0,1)$.
\end{itemize}
For example, we have $\Gamma_1[1,1,1] = 1$, $\Gamma_1[1,1,2] = 2$, and $\Gamma_1[2,1,1] = 0$. Moreover, one can check that $\mathfrak{s}_\omega(\Gamma_i)=M_i, i=\{1,2.3\}$, for the matrices $M_1$, $M_2$ and $M_3$ shown in Figure~\ref{fig:IllustrTRANK123}. 
We note the following properties.
\begin{itemize}
    \item The expression of $\Gamma_1$ contains only one elementary tensor, and is non-zero. Consequently, we have $\trank(\Gamma_1) = \trank_{\F_q}(M_1)= 1$. In accordance with Proposition~\ref{prop:weightsandtensorrank}, we can check that the rows (\emph{resp.} columns) of $M_1$ are $\F_q$-proportional, and that the span of the entries of $M_1$ is one dimensional, hence $\wt_{\slicesp_1}(M_1) = \wt_{\slicesp_2}(M_1) = \wt_{\fibresp_3}(M_1) = 1$. 
    \item The $\F_q$-space $\UUU_1(M_2)$ is spanned by the two vectors $(\omega_1 + \omega_2, \omega_1 + \omega_2)$ and $(\omega_2,\omega_2)$, hence $\wt_{\slicesp_1}(M_2) = 2$. Similarly, $\UUU_2(M_2)$ is spanned by $(\omega_1 + \omega_1,\omega_2)$ hence $\wt_{\slicesp_2}(M_2) = 1$. Finally $\wt_{\fibresp_3}(M_2) = 2$. By Proposition~\ref{prop:weightsandtensorrank}, the expression of $\Gamma_2$ above is minimal and we have $\trank(\Gamma_2) = \trank_{\F_q}(M_2) = 2$.
    \item We have $\wt_{\slicesp_1}(M_3) = \wt_{\slicesp_2}(M_3) = \wt_{\fibresp_3}(M_3) = 2$ and $\trank(\Gamma_3) = 3$ while none of the inequalities listed by the proposition are met. Indeed, if $\Gamma_3$ had tensor-rank 2, then there would exist a basis of $\slicesp_1(\Gamma_3) = \Span_{\F_q} (\left[ \begin{smallmatrix} 1&0\\0&0\end{smallmatrix}\right],\left[ \begin{smallmatrix} 0&1\\1&0\end{smallmatrix}\right])$ of the form $(A_1,A_2)$ with $\rank A_1 = \rank A_2 = 1$, by definition of the tensor-rank. But such a basis cannot exist since the only matrices of rank one in this space are the scalar multiples of $\left[ \begin{smallmatrix} 1&0\\0&0\end{smallmatrix}\right]$. 
\end{itemize}

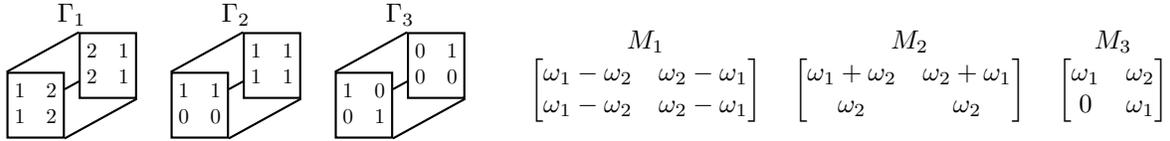
\begin{figure}[ht]
\centering
\begin{tabular}{c c c} $\Gamma_1$&$\Gamma_2$&$\Gamma_3$\\
 \begin{tikzpicture}[thick,scale=0.8, every node/.style={scale=0.8}]
        \node[draw,rectangle] (M) at (0,0) {$\begin{matrix}
            1&2\\1&2
        \end{matrix}$};
        \node[draw,rectangle] (N) at (1.2,0.66) {$\begin{matrix}
            2&1\\2&1
        \end{matrix}$};
        \draw (M)--(N);
        \draw (M.north west) -- (N.north west);
        \draw (M.south east) -- (N.south east);
    \end{tikzpicture}
    &\begin{tikzpicture}[thick,scale=0.8, every node/.style={scale=0.8}]
        \node[draw,rectangle] (M) at (0,0) {$\begin{matrix}
            1&1\\0&0
        \end{matrix}$};
        \node[draw,rectangle] (N) at (1.2,0.66) {$\begin{matrix}
            1&1\\1&1
        \end{matrix}$};
        \draw (M)--(N);
        \draw (M.north west) -- (N.north west);
        \draw (M.south east) -- (N.south east);
    \end{tikzpicture}
    &\begin{tikzpicture}[thick,scale=0.8, every node/.style={scale=0.8}]
        \node[draw,rectangle] (M) at (0,0) {$\begin{matrix}
            1&0\\0&1
        \end{matrix}$};
        \node[draw,rectangle] (N) at (1.2,0.66) {$\begin{matrix}
            0&1\\0&0
        \end{matrix}$};
        \draw (M)--(N);
        \draw (M.north west) -- (N.north west);
        \draw (M.south east) -- (N.south east);
    \end{tikzpicture}
\end{tabular} \phantom{x}
\begin{tabular}{ccc} $M_1$&$M_2$&$M_3$\\
 $\begin{bmatrix}
        \omega_1 - \omega_2 & \omega_2 - \omega_1 \\ \omega_1 - \omega_2 & \omega_2 - \omega_1
    \end{bmatrix}$
    & $\begin{bmatrix}
        \omega_1 + \omega_2 & \omega_2 + \omega_1 \\  \omega_2 &\omega_2
    \end{bmatrix}$ &
     $\begin{bmatrix}
        \omega_1 & \omega_2 \\ 0& \omega_1 
    \end{bmatrix}$
\end{tabular}
\caption{Illustration of tensor-rank one, two, and three elements in $\F_{3^2}^{2 \times 2}$.}
\label{fig:IllustrTRANK123}
\end{figure}
\end{Ex}

\begin{Prop}\label{prop:trianglesincomplement}
Let $\SSS \subseteq \llbracket 0,n-1 \rrbracket^2$ be non-empty. Then the code $\CCC_\alpha(\SSS)$ is an $\F_{q^n}\tm[n\times n, |\SSS|]$ tensor code of minimum distance $d_{\trank}(\CCC_\alpha(\SSS)) \geq \max\{\mu \in \llbracket 0,n-1\rrbracket \mid \exists x,y \in \llbracket 0,n-1\rrbracket \text{ s.t. } \TTT_{x,y,\mu} \cap \SSS = \emptyset\} $ where we denote by $\TTT_{x,y,\mu} = \{(r_1,r_2) \in \llbracket 0,n-1\rrbracket^2 \ | \ r_1 \geq x,r_2 \geq y, r_1 + r_2 \leq \mu + x + y -2\}$ for each $x, y \in \llbracket 0,n-1 \rrbracket$. 
\end{Prop}
\begin{proof} Let $\mu \in \llbracket 0,n-1\rrbracket$ be such that $\SSS \cap \TTT_{x,y,\mu} = \emptyset$ for some $x,y \in \llbracket 0,n-1\rrbracket$. Proposition~\ref{prop:TranslationSets} ensures that there exists matrices $L,R \in \GL_n(\F_q)$ such that $\CCC_\alpha(\SSS') = L \CCC_\alpha(\SSS) R$, where $\SSS' = \{(s_1-x,s_2-y) \mid s \in \SSS \}$, and hence, by Lemma~\ref{lemma:tensorrankdecompositionasmatrices} we have $d_{\trank}(\CCC_\alpha(\SSS')) = d_{\trank}(\CCC_\alpha(\SSS))$. Moreover, since $\SSS'\cap \{(r_1,r_2) \in \llbracket 0,n-1\rrbracket^2 \mid r_1 + r_2 \leq \mu - 2\} = \emptyset$, we have $\SSS' \subseteq \SSS^{(Roth)}_\mu$ and in particular $d_{\trank}(\CCC_\alpha(\SSS')) \geq \mu$.
\end{proof}

\begin{Corol}
    Let $\mu_1,\mu_2\in\llbracket 0,n-1\rrbracket$ and let $\SSS = \llbracket 0,\mu_1\rrbracket \times \llbracket 0,\mu_2\rrbracket$. Then we have $d_{\trank_{\F_q}}(\CCC_\alpha(\SSS)) \geq 2n-1-\mu_1-\mu_2$.
\end{Corol}
\begin{proof}
    Let $\SSS' = \llbracket n-1-\mu_1,n-1\rrbracket \times \llbracket n-1-\mu_2,n-1\rrbracket$. By Proposition~\ref{prop:TranslationSets} and Lemma~\ref{lemma:tensorrankdecompositionasmatrices}, we have $d_{\trank_{\F_q}}(\CCC_\alpha(\SSS)) = d_{\trank_{\F_q}}(\CCC_\alpha(\SSS'))$. Let $\eta = 2n-1-\mu_1-\mu_2$, and note that $\TTT_{0,0,\eta} \cap \SSS' = \emptyset$. Indeed, if $(r_1,r_2) \in \TTT_{0,0,\eta}$ with $r_1 \geq n-1-\mu_1$, then we have $r_2 \leq \eta -2 - r_1 \leq 2n-1-\mu_1-\mu_2-2 -n+1+\mu_1 = n-\mu_2-2$. Therefore, $d_{\trank_{\F_q}}(\CCC_\alpha(\SSS')) \geq \eta$ by Proposition \ref{prop:trianglesincomplement} and we have the wanted result.
\end{proof}

\section{Decoding}
\label{sec:decoding}
\subsection{Decoding in a subspace of a direct sum of Gabidulin codes.}
\label{subsec:decodingfibrewise}

In \cite{GABIDULINRankDistanceCodes}, Gabidulin introduced a decoding algorithm for the Gabidulin code $\GG_k(\alpha)$ of length $n$ that corrects any error of rank at most $\lfloor{\frac{n-k}{2}}\rfloor$. Since Proposition~\ref{prop:GabidulinColumns} states that every column and row of a codeword in $\mathcal{C}_{\alpha}(\mathcal{S})$ is a Gabidulin codeword, we can infer a decoding algorithm that uses the fact that $\mathcal{C}_{\alpha}(\mathcal{S})$ is embedded in a direct sum of Gabidulin codes. This motivates us to study this code with respect to the metric inherited from the structure of direct sums of rank-metric codes, a metric linked to the tensor-rank metric by Proposition~\ref{prop:weightsandtensorrank}.
We denote by \texttt{GabDec}$(r,k,\alpha)$ the function that returns the unique Gabidulin codeword $c \in \GG_{k}(\alpha)$ for $r = c +e$ with $e$ a vector in $\mathbb{F}_{q^n}^n$ of rank at most $\left\lfloor\frac{n-k}{2}\right\rfloor$. In the event of a decoding failure, the output of \texttt{GabDec}$(r,k,\alpha)$ is the received word $r$.

\begin{algorithm}[ht]
\caption{Column-wise decoding (\emph{i.e.} fibre-wise decoding for $j = 1$)}\label{alg:fibrewise}
\begin{algorithmic}
\Require $n$ an integer, $q$ a prime power, $\mathcal{S} \subseteq \llbracket 0,n -1 \rrbracket^2$ and $R \in \mathbb{F}_{q^n}^{n \times n}$.
\State $C\leftarrow 0 \in \mathbb{F}_{q^n}^{n \times n}$.
\For{$i_2\in \llbracket 1,n\rrbracket$}
\State $C[:,i_2]  \leftarrow \texttt{GabDec}(R[:,i_2], \max(\pi_1(\mathcal{S}))+1, \alpha)$
\EndFor
\State \textbf{return} $C$.
\end{algorithmic}
\end{algorithm} 

\setcounter{algorithm}{0}

\begin{algorithm}[ht]
\caption{\!\!\!{\footnotesize\textbf{'}}\ Row-wise decoding (\emph{i.e.} fibre-wise decoding for $j = 2$)}\label{alg:fibrewisebis}
\begin{algorithmic}
\Require $n$ an integer, $q$ a prime power, $\mathcal{S} \subseteq \llbracket 0,n -1 \rrbracket^2$ and $R \in \mathbb{F}_{q^n}^{n \times n}$.
\State $C\leftarrow 0 \in \mathbb{F}_{q^n}^{n \times n}$.

\For{$i_1\in \llbracket 1,n\rrbracket$}
\State $C[i_1,:]  \leftarrow \texttt{GabDec}(R[i_1,:], \max(\pi_2(\mathcal{S}))+1, \alpha)$
\EndFor
\State \textbf{return} $C$.
\end{algorithmic}
\end{algorithm} 

Given a received word $R$ as input, Algorithm~\ref{alg:fibrewise} returns a unique codeword $C \in \mathcal{C}_\alpha(\mathcal{S})$ satisfying $R = C + E$ for some error matrix $E$ such that the $\mathbb{F}_q$-rank of every $j$-fibre (\emph{i.e.} column if $j =1$, and row if $j = 2$) of $E$ has rank at most the decoding radius of the Gabidulin code of dimension $\max \pi_j(\mathcal{S})$. That is, it can correct any error $E$ subject to the following constraints:
\begin{equation}
    \label{eq:CriterionFibreWiseNaive}
    \begin{array}{cc}
         {\displaystyle \max_{i_{2} \in \llbracket 1,n\rrbracket}} \rank_{\mathbb{F}_q} E[:,i_2] \leq \left\lfloor{\frac{n-\max \pi_1(\mathcal{S}) - 1}{2}}\right\rfloor &  \text{if} \ j = 1;\\[0.5cm]
         {\displaystyle\max_{i_{1} \in \llbracket 1,n\rrbracket}} \rank_{\mathbb{F}_q} E[i_1,:] \leq \left\lfloor{\frac{n-\max \pi_2(\mathcal{S}) - 1}{2}}\right\rfloor&  \text{if} \ j=2.
    \end{array}
\end{equation}

We can improve this initial approach with the additional information that the codeword sent is not any codeword of a product of Gabidulin codes. We remind the reader that the Hamming weight of a vector $u \in \mathbb{F}_{q^n}^n$ is its number of non-zero entries, \emph{i.e.} $\textup{w}_H(u):=|\left\{i \in \llbracket 1,n\rrbracket\ | \ u_i \neq 0\right\}|$. In addition, the Hamming weight of $u$ is bounded from below by its $\F_q$-rank. Indeed, the dimension of $\Span_{\F_q}(\{u_1,\dots,u_n\})$ is at most the number of non-zero elements among $u_1,\dots,u_n$. 

\begin{Lemma} \label{lemma:ErrorMaxRankcolumnsHammingRows}
Let $\theta,\kappa$ be integers. Let $E \in \mathbb{F}_{q^n}^{n \times n}$ such that
\begin{equation}
    \label{eq:maxrankupperboundedbutnoteverywhere} 
    \min_{\substack{\JJJ \subseteq \llbracket 1,n\rrbracket \\|\JJJ| = \kappa }} \max_{i_2 \in \JJJ} \ \rank_{\mathbb{F}_q} E[:,i_2] \leq \theta.
\end{equation}
Assume that Algorithm~\ref{alg:fibrewise} with input a received word of the form $R = C + E$ with $C \in \mathcal{C}_\alpha(\mathcal{S})$ returns a matrix $\Tilde{C}$ such that 
$$\forall i_2 \in \llbracket 1,n\rrbracket :  \rank_{\F_q}E[:,i_2] \leq \theta \implies C[:,i_2] = \Tilde C[:,i_2].$$
Then $\Tilde{C} = C + \Tilde E$ where every row of $\Tilde{E}$ has Hamming weight at most $n - \kappa$.
\end{Lemma}

\begin{proof}
    Let $\JJJ \subseteq \llbracket 1,n \rrbracket$ be a set of cardinality $\kappa$. Assume that $\JJJ$ is an argument of the minima in (\ref{eq:maxrankupperboundedbutnoteverywhere}). Then for each $i_2 \in \JJJ$, the $\mathbb{F}_q$-rank of $E[:,i_2]$ is at most $\theta$ thus $\Tilde{E}[:,i_2] = 0$ since the Gabidulin decoder corrects each such $E[:,i_2]$. In particular, for each $(i_1,i_2) \in \llbracket 1,n\rrbracket \times \JJJ$ we have $E[i_1,i_2] = 0$, thus $E[i_1,:]$ has at most $n - |\JJJ| = n - \kappa$ non-zero entries.
\end{proof}

Assume that $\tilde C$ is the output of Algorithm~\ref{alg:fibrewise} for the input $R = C+ E$ for which $C$ a codeword of $\CCC_\alpha(\SSS)$ with $E$ satisfying (\ref{eq:maxrankupperboundedbutnoteverywhere}). 
Any row of $\Tilde C$ is a Gabidulin codeword, or is a row of the received word $R$ with an error given by the corresponding row of $\Tilde C-C$. In particular, the $i_2$-th row of $\Tilde C-C$ is non-zero only if a decoding failure occurred for the $i_2$-th row of $R$.   
Furthermore, the Hamming weight of any row of $\Tilde{C}-C$ is at most 
$n-\kappa$. Therefore, applying a Gabidulin decoder to each row of the output of Algorithm~\ref{alg:fibrewise} can correct the remaining errors where 
$n-\kappa$ is below a certain threshold. In what follows, we will assume that $\mathcal{S} = \llbracket 0,\mu_1\rrbracket \times \llbracket 0,\mu_2\rrbracket$.

\begin{algorithm}[ht]
\caption{Two-way fibre-wise decoding}\label{alg:fibreWiseTwoDirrs}
\begin{algorithmic}
\Require $n$ an integer, $q$ a prime power, $\mu_1,\mu_2 \in \llbracket 0,n-1\rrbracket$, and $R \in \mathbb{F}_{q^n}^{n\times n}$
\State $C\leftarrow 0 \in \mathbb{F}_{q^n}^{n\times n}$
\State $\Tilde C\leftarrow 0 \in \mathbb{F}_{q^n}^{n\times n}$
\For{$i_2\in \llbracket 1,n\rrbracket$}
\State $\Tilde C[:,i_2] \leftarrow \texttt{GabDec}(R[:,i_2],\mu_1+1,\alpha)$
\EndFor
\For{$i_1 \in \llbracket 1,n\rrbracket$}
\State $C[i_1,:] \leftarrow \texttt{GabDec}(\Tilde{C}[i_1,:],\mu_2 +1,\alpha)$
\EndFor
\State \textbf{return} $C$.
\end{algorithmic}
\end{algorithm}

In particular, setting $\theta = \left\lfloor \frac{n - \mu_1 - 1}{2} \right\rfloor$ and $\kappa=\left\lceil \frac{n + \mu_2 + 1}{2} \right\rceil$, we get the following statement.

\begin{Thm} \label{thm:twowayscriterion}
    Let $\mu_1,\mu_2 \in \llbracket 0,n-1\rrbracket$. Let $R = C + E$ where $C \in \mathcal{C}_{\alpha}(\llbracket 0,\mu_1\rrbracket \times \llbracket 0,\mu_2\rrbracket)$ and $E \in \mathbb{F}_{q^n}^{n \times n}$ satisfies
    \begin{equation}
    \label{eq:maxrankupperboundedbutnoteverywherePRACTICAL} 
    \min_{\substack{\JJJ \subseteq \llbracket 1,n\rrbracket \\|\JJJ| =   \left\lceil \frac{n + \mu_2 + 1}{2} \right\rceil}} \max_{i_2 \in \JJJ} \ \rank_{\mathbb{F}_q} E[:,i_2] \leq  \left\lfloor \frac{n - \mu_1 - 1}{2} \right\rfloor.
\end{equation}
 Then running Algorithm~\ref{alg:fibreWiseTwoDirrs} with input $R$ returns $C$.
\end{Thm}
\begin{proof}
    Let $R = C + E$ with $E \in \mathbb{F}_{q^n}^{n \times n}$ satisfying inequality  (\ref{eq:maxrankupperboundedbutnoteverywherePRACTICAL}). Let $\tilde{C} = C +\tilde{E}$ be the output of Algorithm~\ref{alg:fibrewise} for the input $R$. For each $i_2\in \llbracket 1,n\rrbracket$, the column $R[:,i_2]$ of $R$ is of the form $R[:,i_2] = c + e$ where $c \in \GG_{\mu_1+1}(\alpha)$ and $e \in \mathbb{F}_{q^n}^n$. By Proposition~\ref{prop:GabidulinColumns}, $\rank_{\F_q}(e) \leq \left\lfloor \frac{n - \mu_1 - 1}{2} \right\rfloor$ for at least $\left\lceil \frac{n + \mu_2 + 1}{2} \right\rceil$ different choices of $i_2$. Note that $n - \left\lceil\frac{n + \mu_2 + 1}{2} \right\rceil = \left\lfloor \frac{n - \mu_2 -1}{2}\right\rfloor $. By Lemma~\ref{lemma:ErrorMaxRankcolumnsHammingRows} and by the result above, $\Tilde C$ is such that its rows are of the form $\tilde C[i,:] = \tilde c + \tilde e$ where $c \in \GG_{\mu_2 + 1}(\alpha)$ and $e \in \F_{q^n}^n$ of Hamming weight at most $\left\lfloor\frac{n - \mu_2 -1}{2}\right\rfloor$, and in particular of $\mathbb{F}_q$-rank at most $\frac{n - \mu_2 -1}{2}$. Thus applying a Gabidulin decoder on each row of $\Tilde C$ will return $C$, therefore, so will Algorithm~\ref{alg:fibreWiseTwoDirrs}. 
\end{proof}

\begin{Ex}\label{ex:DiffAlg1Alg2}
    Fix $n = 5$, $q = 3$ and $\mu=\mu_1=\mu_2 = 2$. Note that $\left\lceil \frac{n + \mu + 1}{2} \right\rceil= 4 $ and $ \left\lfloor \frac{n - \mu - 1}{2} \right\rfloor = 1$. Let $M = 0 \in \CCC_\alpha(\llbracket 0,2\rrbracket^2)$ and let     $$E_1 := \begin{bmatrix}
        \alpha_1 & \alpha_2 & \alpha_3 & \alpha_4 & \alpha_5 \\
        \alpha_1 & \alpha_2 & 0 & 0 & \alpha_5 \\ 
        \alpha_1 & \alpha_2 & 0 & 0 & \alpha_5 \\ 
        \alpha_1 & \alpha_2 & 0 & 0 & \alpha_5 \\ 
        \alpha_1 & \alpha_2 & 0 & 0 & \alpha_5 \\ 
    \end{bmatrix},
    \quad  E_2:= \begin{bmatrix}
        \alpha_1 & \alpha_2 & \alpha_3 & \alpha_4 & \alpha_5 \\
        \alpha_1 & \alpha_2 & 0 & 0 & \alpha_4 \\ 
        \alpha_1 & \alpha_2 & 0 & 0 & \alpha_3 \\ 
        \alpha_1 & \alpha_2 & 0 & 0 & \alpha_2 \\ 
        \alpha_1 & \alpha_2 & 0 & 0 & \alpha_1 \\ 
    \end{bmatrix}.$$
    Since every column of $E_1$ has a rank 1, running Algorithm~\ref{alg:fibrewise} (for $j = 1$) on $\CCC_\alpha(\llbracket 0,2\rrbracket^2)$ using $R = M + E_1$ as input will indeed return $M$ since \texttt{GabDec} returns the zero codeword if its input is any column of $E_1$. The last column of $E_2$ has rank $5$. Therefore, running \texttt{GabDec} on this column may return (a decoding failure or) a random codeword in the Gabidulin code. Since $E_2$ satisfies (\ref{eq:maxrankupperboundedbutnoteverywherePRACTICAL}), with an argument of the minima $\JJJ = \{1,2,3,4\}$, $E_2$ will be decoded by Algorithm~\ref{alg:fibreWiseTwoDirrs}.
\end{Ex}

Denote by $\EEE_i$ the set of errors corrected by Algorithm $i$   applied to the code $\CCC_\alpha(\llbracket 0,\mu_1\rrbracket \times \llbracket 0,\mu_2\rrbracket)$. In other words,
$$\begin{array}{lrcl}
     &\EEE_1 &:=& \{ E\in \F_{q^n}^{n\times n} \ | \ \forall C \in \CCC_\alpha(\llbracket 0,\mu_1\rrbracket\times \llbracket0,\mu_2\rrbracket) : \texttt{Alg1}(C+E) = C\},\\
    &\EEE_2 &:=& \{ E\in \F_{q^n}^{n\times n} \ | \ \forall C \in \CCC_\alpha(\llbracket 0,\mu_1\rrbracket\times \llbracket0,\mu_2\rrbracket) : \texttt{Alg2}(C+E) = C\}.
\end{array}$$
where we denote by \texttt{Alg1}$(R)$ (\emph{resp.} \texttt{Alg2}$(R)$) for each $R \in \F_{q^n}^{n\times n}$ the output of Algorithm~\ref{alg:fibrewise} (\emph{resp.} Algorithm~\ref{alg:fibreWiseTwoDirrs}) with input $R$ for the parameters $\alpha$, $\mu_1$ and $\mu_2$.

Clearly, we have $\EEE_1 \subseteq \EEE_2$. Moreover, if we denote by 
$\mathfrak{n}(\mathbb{F}_{q}^{n\times n},r)$ the number of matrices $n\times n$ over $\mathbb{F}_q$ of rank exactly $r$, then we have 
\begin{equation}
    \label{eq:numbererrorcolumnwise}|\EEE_1| \geq 
    \left(\sum_{r = 0}^{\left\lfloor \frac{n - \mu_1 - 1}{2} \right\rfloor} \mathfrak{n}(\mathbb{F}_{q}^{n\times n},r)\right)^n
\end{equation}
since Algorithm~\ref{alg:fibrewise} corrects any error for which every column has rank at most $\left\lfloor \frac{n - \mu_1 - 1}{2} \right\rfloor$, and
\begin{equation}
    \label{eq:numbererrorfibrewise} |\EEE_2| \geq 
\sum_{k = \left\lceil \frac{n + \mu_2 + 1}{2} \right\rceil}^n \binom{n}{k} \left(\sum_{r = 0}^{\left\lfloor \frac{n - \mu_1 - 1}{2} \right\rfloor} \mathfrak{n}(\mathbb{F}_{q}^{n\times n},r)\right)^k\left(q^{n^2}  - \sum_{r = 0}^{\left\lfloor \frac{n- \mu_1 - 1}{2} \right\rfloor} \mathfrak{n}(\mathbb{F}_{q}^{n\times n},r)\right)^{n-k}
\end{equation}
since Algorithm~\ref{alg:fibreWiseTwoDirrs} corrects any error with at least $\left\lceil \frac{n + \mu_2 + 1}{2} \right\rceil$ columns that have rank at most $\left\lfloor \frac{n - \mu_1 - 1}{2} \right\rfloor$. 
We will compare the integers in the upper-bounds (\ref{eq:numbererrorcolumnwise}) and (\ref{eq:numbererrorfibrewise}) in Section~\ref{sec:Comparison}.
\label{location:numbersofcorrectederrors}

\begin{Rq}
    Note that the set of errors that can be corrected by Algorithm~\ref{alg:fibreWiseTwoDirrs} depends on the choice of first decoding along the columns of the matrix and then decoding along its rows or vice versa. Indeed, if the algorithm first decodes along the rows of and the error $E_2$ of Example~\ref{ex:DiffAlg1Alg2} occurs, then no row would be decoded correctly at this step as each row of $E_2$ has rank at least 3.
\end{Rq}

\subsection{A factorisation algorithm}
\label{subsec:FactoringOnTheLeft}

We now focus on resolving the following problem, which is relevant to obtaining an extension of the Loidreau-Overbeck decoder.
\begin{Pb}\label{pb:leftdivisionbilinear}
    Let $t \in \N$, $V(Z) \in \MM_{q,\F_{q^n}}[Z]$ and $N(X,Y) \in \MM_{q,\F_{q^n}}[X,Y]$ such that $\qdeg V(Z) \leq t$, $\Supp N(X,Y) \subseteq \SSS + \llbracket 0,t\rrbracket$.
    Compute $f(X,Y) \in \CCC(\SSS)$ satisfying $V(f(X,Y)) = N(X,Y)$.
    \end{Pb}

This is a generalisation of symbolic left-hand division of polynomials, \emph{i.e.} given two $q$-polynomials $f(Z)$ and $g(Z)$ in $\MM_{q,\F_{q^n}}[Z]$ compute $h(Z) \in \MM_{q,\F_{q^n}}[Z]$ satisfying $f(h(Z)) = g(Z)$; see \cite{Ore1933OnASpecialClassOfPolynomials}. With the observation below, we can use an existing division algorithm on the ring of linearised polynomials to solve Problem \ref{pb:leftdivisionbilinear}. 

\begin{Lemma}\label{lemma:diagonaldecompositionbilinearpolys}
        Let $V(Z) \in \MM_{q,\F_{q^n}}[Z]$ and let $N(X,Y) = \sum_{s \in \N_0^2} n_sX^{q^{s_1}}Y^{q^{s_2}} \in \MM_{q,\F_{q^n}}[X,Y]$ be $q$-polynomials with $(n_s)_{s\in\N_0^2}$ a sequence in $\F_{q^n}$ with finitely many non-zero terms.\\
    For each integer $\delta \in \Z$, define     $\displaystyle N_\delta(Z) := \sum_{\substack{\tau \in \N_0\\\tau +\delta \geq 0}} n_{(\tau,\delta+\tau)} Z^{q^\tau}.$
    The following are equivalent.
    \begin{enumerate}[label = (\roman*)]
        \item There exists $f(X,Y) \in \MM_{q,\F_{q^n}}[X,Y]$ such that $V(f(X,Y)) = N(X,Y)$.
        \item There exists $(f_\delta(Z))_{\delta \in \Z}$ a sequence of $q$-polynomials in $\MM_{q,\F_{q^n}}[X,Y]$ with finitely many non-zero terms such that all monomials in $f_\delta(Z)$ have $q$-degree at least $n-\delta$, and $V(f_\delta(Z)) = N_\delta(Z)$ for each $\delta \in \Z$.
    \end{enumerate}
    In particular, if (ii) is satisfied for a sequence $(f_\delta(Z))_{\delta \in \Z}$, then $f(X,Y) = \sum_{\delta \in \Z} f_\delta(XY^{q^\delta})$ satisfies (i).
\end{Lemma}
\begin{proof}
    Since the sets $(\{(x,x+\delta) \ | \ x \in \Z\})_{\delta \in \Z}$ form a partition of $\Z^2$, we have
    $$N(X,Y) = \sum_{s \in \N_0^2} n_{s}X^{q^s}Y^{q^s} = \sum_{\delta \in \Z}  \sum_{\substack{\tau \in \N_0\\\tau +\delta \geq 0}} n_{(\tau,\delta+\tau)} (XY^{q^\delta})^{q^\tau} = \sum_{\delta \in \Z} N_\delta(XY^{q^\delta}).$$
    Note that for each $\delta \in \Z$, $N_\delta(XY^{q^{\delta}})$ is indeed a bilinearised $q$-polynomial, as the definition of $N_\delta(Z)$ implies that, if not zero, its monomials have $q$-degree at least $-\delta$. 
    
    Let $f(X,Y) \in \MM_{q,\F_{q^n}}[X,Y]$ be a bilinearised $q$-polynomial and let $(f_\delta(Z))_{\delta \in \Z}$ be the sequence of $q$-polynomials satisfying 
    $f(X,Y) = \sum_{\delta \in \Z} f_{\delta}(XY^{q^{\tau}})$. 
    Suppose that $V(f(X,Y)) = N(X,Y)$. Since we have
    $$\sum_{\delta \in \Z} N_\delta(XY^{q^\delta}) = N(X,Y) = V(f(X,Y)) = \sum_{\delta \in \Z} V(f_\delta(XY^{q^\delta})), $$
    and since $\Supp(V(f_\delta(XY^{q^{\delta}})))\subseteq \{(x,x+\delta) \ | \ x \in \Z\}$, we obtain that $N_\delta(XY^{q^\delta}) = V(f(XY^{q^\delta}))$ and thus $N_\delta(Z) = V(f_\delta(Z))$.

    Conversely, if there exists a sequence $(f_\delta(Z))_{\delta \in \Z}$ of $q$-polynomials in $\MM_{q,\F_{q^n}}[X,Y]$ with finitely many non-zero terms such that all monomials in $f_\delta(Z)$ have $q$-degree at least $-\delta$, and $V(f_\delta(Z)) = N_\delta(Z)$ for each $\delta \in \Z$, then setting $f(X,Y) = \sum_{\delta \in \Z} f_{\delta}(XY^{q^{\tau}})$, we obtain $V(f(X,Y)) = N(X,Y)$. 
\end{proof}

Using the existing left-Euclidean division algorithm as described in \cite{Ore1933NonCommutativePolynomials}, we can solve Problem~\ref{pb:leftdivisionbilinear} as follows: given $t \in \N$, $V(Z) \in \MM_{q,\F_{q^n}}[Z]$, and $N(X,Y) \in \MM_{q,\F_{q^n}}[X,Y]$ such that $\qdeg V(Z) \leq t$, $\Supp N(X,Y) \subseteq \SSS + \llbracket 0,t\rrbracket$, we apply the left-Euclidean division algorithm to factor $V(Z)$ from $N_\delta(Z)$ for each $\delta\in \Z$ such that $N_\delta(Z)\neq 0$ to obtain $(f_\delta(Z))_{\delta \in \Z}$, and then compute $f(X,Y)$ as in Lemma~\ref{lemma:diagonaldecompositionbilinearpolys}. We will refer to this algorithm as \textbf{Algorithm~\ref{alg:divisionalgorithm}}; see Appendix~\ref{sec:DetailsFactoringAlgorithm}.

\subsection{An extension of the Loidreau-Overbeck decoding algorithm}\label{subsec:decodinglinearisedproblem}
We will now use the technique introduced in \cite{LoidreauDecodingBeyondECC_ACCT} to decode a different range of errors using the $j$-slice space weight. Let $\mu \in \llbracket 0,n-2\rrbracket$ be an integer. Throughout this section, we fix $\mathcal{S}=\llbracket 0,\mu \rrbracket^2$, we fix a codeword $C \in \CCC_{\alpha}(\SSS)$, and we fix $R,E \in \F_{q^n}^{n\times n}$ such that $R = C + E$.

\begin{Pb} \label{pb:Linearisedproblem} Given $\SSS$, $t \in {\mathbb{N}_0}$ and $R$, find $(V(Z),N(X,Y)) \in \mathscr{M}_{q,\mathbb{F}_{q^n}}[Z]\times  \mathscr{M}_{q,\mathbb{F}_{q^n}}[X,Y]$ satisfying $\qdeg V(Z) \leq t$ and $\Supp(N(X,Y)) \subseteq \mathcal{S} + \llbracket 0,t\rrbracket$ in 
\begin{equation}
    \label{eq:linearisedProblem}
      V(R_{i_1,i_2}) = N(\alpha_{i_1},\alpha_{i_2}),\:\forall i \in \llbracket 1,n\rrbracket^2.
\end{equation}
\end{Pb}

\begin{Rq}
    Observe that if $P(X,Y)$ is a representative of the polynomial coset associated to $R$, then equation (\ref{eq:linearisedProblem}) is equivalent to $\overline{V(Z)} \circ \overline{P(X,Y)} = \overline{N(X,Y)}$ in $\mathscr{M}_{q,\mathbb{F}_{q^n}}[X,Y]/U_{X,Y}$ (where $U_{X,Y}$ was defined in \eqref{eq:UXY}) since we have an isomorphism $\mathscr{M}_{q,\mathbb{F}_{q^n}}[X,Y]/U_{X,Y} {\to}\mathbb{F}_{q^n}^{n\times n}$ , see Proposition~\ref{prop:isomPolyAndTensors}.
\end{Rq}

In this situation, $(V(Z),V(Z) \circ f(X,Y))$ is a solution of equation (\ref{eq:linearisedProblem}) for any $q$-polynomial $V(Z)$ of $q$-degree at most $t$ such that $V(Z)$ vanishes at any entry of $E =R-C$ with $C$ the associated codeword of $f(X,Y)$, and the following theorem states a criterion that establishes when all solutions of the problem have this form.

\begin{Thm}\label{thm:RadicalCriterion}
    Let $f(X,Y) \in \mathcal{C}(\mathcal{S})$ be such that $\ev_\alpha(f) = C$. Let $\Theta \in \llbracket 0,n-\mu-2\rrbracket$ and assume that $\min_{j \in \{1,2\}}\wt_{\slicesp_j}(E) \leq \Theta$.
    Then any solution $(V(Z),N(X,Y))$ of Problem~\ref{pb:Linearisedproblem} for $\SSS$, $t = n-\mu-1-\Theta$ and $R=C + E$ is such that $V(f(X,Y)) = N(X,Y)$. 
    \end{Thm}

\begin{proof} 
Let $(V(Z),N(X,Y))$ be a solution of Problem~\ref{pb:Linearisedproblem} with parameters $\SSS$, $t = n-\mu-1-\Theta$ and $R = C+ E$. By the linearity of $V(Z)$, for each $(i_1,i_2) \in \llbracket 1,n \rrbracket^2$ we have
$$(V\circ f - N)(\alpha_{i_1},\alpha_{i_2}) = V((R - E)[i_1,i_2]) - N(\alpha_{i_1},\alpha_{i_2})= - V(E[i_1,i_2]).$$

Let $W := (V(E[i_1,i_2]))_{(i_1,i_2) \in \llbracket 1,n\rrbracket^2} \in \mathbb{F}_{q^n}^{n \times n}$ be such that $W$ is the matrix defining the bilinear map $N-V\circ f$ over $\mathbb{F}_q$ with respect to the basis $\alpha$. Assume without loss of generality that $\wt_{\slicesp_{1}}(E) \leq \Theta$, \emph{i.e.} $\UUU_{1}(E)$ is of dimension at most $\Theta$. Let $\ell : \F_q^n \to \F_{q^n}^n$ be the $\mathbb{F}_{q}$-linear map  given by $\ell(x) = xW$. Since $W[i_1,i_2] = V(E[i_1,i_2])$ for each $i_1,i_2 \in \llbracket 1,n\rrbracket^2$, then $\dim_{\F_q} \UUU_1(W) = \dim_{\F_q} \im \ell \leq \dim_{\F_q} \UUU_1(E)$. Moreover, since $-W$ is the matrix associated to the $\F_q$-bilinear form $V \circ f - N$, then we have $\Rr\Aaa\Ddd_1(V\circ f - N) \simeq \ker \ell$. Therefore, by the rank nullity theorem we have
$$\dim \mathfrak{Rad}_1(V\circ f - N) = \dim_{\mathbb{F}_q}  \ker (\ell)  = \dim_{\mathbb{F}_q} \mathbb{F}_q^n - \wt_{\slicesp_1}(W) \geq n - \wt_{\slicesp_1}(E) \geq n - \Theta.$$

Finally, assume that $(V \circ f  - N)(X,Y)$ is not the zero polynomial. Both $V(f(X,Y))$ and $N(X,Y)$ have support in $\mathcal{S} + \llbracket 0, t \rrbracket$ by Lemma~\ref{lemma:CoefCompositionVof} and by the assumed constraints on the solutions of Problem~\ref{pb:Linearisedproblem}, so does $(V \circ f  - N)(X,Y)$. Therefore, we have:
    $$ \qdeg_{X}(V \circ f - N)(X,Y) \leq \max(\pi_{j}(\mathcal{S})) + t \leq n-\Theta-1.$$
This contradicts the assertions of Corollary~\ref{corol:CorolRadicalPolynomialMulivar}. Therefore, $(V \circ f - N)(X,Y) = 0$, and in other words, $V(f(X,Y)) = N(X,Y)$.
\end{proof}

\begin{Corol} \label{corol:DecodingConditionRadical}
    Assume that we have
    \begin{equation}\label{eq:ConditionRadicalDecoding}
        \wt_{\Sigma\!\slicesp}(E) := \wt_{\fibresp_{3}}(E) +  \min_{j \in \{1,2\}}\wt_{\slicesp_j}(E) \leq n - \mu -1.
    \end{equation}
    Then there exists $t \in \llbracket 0,n-1\rrbracket$ such that every solution $(V(Z),N(X,Y))$ of Problem~\ref{pb:Linearisedproblem} with parameters $\SSS$, $t$ and $R$ is of the form $N(X,Y) = V(Z) \circ f(X,Y)$. 
\end{Corol}

\begin{proof}
    Let $t = \wt_{\fibresp_{3}}(E)$. Then we have $\min_{j \in \{1,2\}}\wt_{\slicesp_j}(E) \leq \Theta := n - \mu -1-t$ and the statement follows from Theorem~\ref{thm:RadicalCriterion}.
\end{proof}

Therefore, assuming that Equation~(\ref{eq:ConditionRadicalDecoding}) is satisfied, solving the linear system (\ref{eq:linearisedProblem}) for a good choice of $t$, and factoring any non-zero solution $(V(Z),N(X,Y))$ using Algorithm~\ref{alg:divisionalgorithm} allows one to recover $f(X,Y)$ and thus the transmitted codeword. Algorithm~\ref{alg:RadicalDecoding} can retrieve any codeword in $\mathcal{C}_\alpha(\mathcal{S})$ with $\mathcal{S} = \llbracket 0,\mu\rrbracket^2$ with an error $E\in \mathbb{F}_{q^n}^{n\times n}$ such that $\min_{j \in \{1,2\}} \wt_{\slicesp_j}(E)$ at most $n - \mu- t - 1$ and $\wt_{\fibresp_{3}}(\mathcal{S})$ at most $t$, for a given $t \in \llbracket 0,n-1\rrbracket$. 

\begin{algorithm}
\caption{Radical decoding with specified fibre-weight.\label{alg:RadicalDecoding}}
\begin{algorithmic}

\Require $n$ an integer, $q$ a prime power, $t \in \llbracket 0,n-1\rrbracket$, $\mu \in \llbracket 0,n-1\rrbracket$ defining $\mathcal{S}$, $R \in \mathbb{F}_{q^n}^{n \times n}$ received.
\If{Problem~\ref{pb:Linearisedproblem} has no non-zero solution for $\SSS$,$t$ and $R$}
\State \textbf{return} ``Decoding failure''
\EndIf
\State Pick a non-zero solution $(V(Z),N(X,Y))$ of Problem~\ref{pb:Linearisedproblem}.\State Run \textbf{Algorithm~\ref{alg:divisionalgorithm}} to factorise $N(X,Y)$ into $V(f(X,Y))$.
\State Generate the codeword $C$ associated to $f(X,Y)$.
\State \textbf{return} $C$
\end{algorithmic}
\end{algorithm}

Moreover, we can use the following fact to retrieve a solution that can be factored from the resolution of Problem \ref{pb:Linearisedproblem} for a greater value of $t$ by checking the support of the solutions.

\begin{Corol}\label{corol:findsmallestsupport}
    Assume that $\wt_{\Sigma\!\slicesp}(E) \leq n-\mu -1$ and let $t \in \llbracket 1,n-\mu -2 \rrbracket$. Then Problem~\ref{pb:Linearisedproblem} with parameters $\SSS$, $t$ and $R$ has a non-zero solution if and only if $t \geq \wt_{\fibresp_3}(E)$. Moreover, if $(V(Z),N(X,Y))$ is a non-zero solution of Problem~\ref{pb:Linearisedproblem} with parameters $\SSS$, $t$ and $R$ such that $\qdeg V(Z) \leq \delta$, $\Supp(N(X,Y)) \subseteq \SSS + \llbracket 0,\delta\rrbracket$, and such that $\delta \in \llbracket 0,n-1\rrbracket$ is minimal with this property, then $N(X,Y) = V(Z) \circ f(X,Y)$.
\end{Corol}
\begin{proof}
    If $E = 0$, the statement is obvious. Assume now that $E \neq 0$. First, let $(V(Z),N(X,Y))$ be a solution of Problem~\ref{pb:Linearisedproblem} with parameters $\SSS$, $t$ and $R$ with $t < \wt_{\fibresp_3}(E)$. By Theorem \ref{thm:RadicalCriterion}, since $\max_{j \in \{1,2\}}\wt_{\slicesp_j}(E) < n-\mu -1 -t$, we have $N(X,Y) = V(f(X,Y))$, and thus $V(E[i]) = 0$ for each $i \in \llbracket 1,n\rrbracket^2$. 
    Since $\qdeg V(Z) < \wt_{\fibresp_3}(E)$, we have $V(Z) = 0$ and $N(X,Y) = 0$. Conversely, $(V(Z),V(f(X,Y)))$ is a solution of Problem~\ref{pb:Linearisedproblem} with parameters $\SSS$, $t\geq \wt_{\fibresp_3}(E)$ and $R$   with $V(Z) \in \MM_{q,\F_{q^n}}[Z]$ such that $\qdeg V(Z) = \wt_{\fibresp_3}(E)$ and $V(E[i]) = 0$ for each $i \in \llbracket 1,n\rrbracket$, which proves the first part of the statement. Now, if $(V(Z),N(X,Y))$ is a non-zero solution of Problem~\ref{pb:Linearisedproblem} with parameters $\SSS$, $t$ and $R$ such that $\qdeg V(Z) \leq \delta$, $\Supp(N(X,Y)) \subseteq \SSS + \llbracket 0,\delta\rrbracket$, and such that $\delta \in \llbracket 0,n-1\rrbracket$ is minimal with this property, then it is also a solution of Problem~\ref{pb:Linearisedproblem} with parameters $\SSS$, $\delta$ and $R$, and $\delta = \wt_{\fibresp_3}(E)$ by the first part. Thus the solution is factorisable by Theorem \ref{thm:RadicalCriterion}, since $\max_{j \in \{1,2\}}\wt_{\slicesp_j}(E) \leq n-\mu -1-\delta$.
\end{proof}

{Therefore, we can amend Algorithm~\ref{alg:RadicalDecoding} by searching for 
a solution $(V(Z), N(X,Y))$ of Problem~\ref{pb:Linearisedproblem} for which
$V(Z)$ has least degree $\delta$ and $N(X,Y)$ has support contained in $ \SSS + \llbracket 0,\delta \rrbracket$. That is, we find a solution corresponding to the least degree of $V(Z)$ such that $N(X,Y)$ that can be factored.}
In this way, Algorithm~\ref{alg:RadicalDecodingFindMin} can correct any error such that $\wt_{\fibresp_{3}}(E) + \min_{j \in \{1,2\}} \wt_{\slicesp_j}(E) \leq n - \mu -1$.

\begin{algorithm}
\caption{Radical decoding.\label{alg:RadicalDecodingFindMin}}
\begin{algorithmic}

\Require $n$ an integer, $q$ a prime power, $\mu \in \llbracket 0,n-1\rrbracket$ defining $\mathcal{S}$, $R \in \mathbb{F}_{q^n}^{n\times n}$ received.
\If{Problem~\ref{pb:Linearisedproblem} has no non-zero solution for $\SSS$,$(n-\mu-2)$ and $R$}
\State \textbf{return} ``Decoding failure''
\EndIf
\State Find the least $\delta \geq 1$ such that Problem~\ref{pb:Linearisedproblem} with parameters $\SSS$, $n-\mu - 2$ and $R$ has a non-zero solution $(V(Z),N(X,Y))$ such that $\qdeg V(Z) \leq \delta$ and $\Supp N(X,Y) \subseteq \SSS + \llbracket 0,\delta \rrbracket$.
\State Pick $(V(Z),N(X,Y))$ such a solution.
\State Run \textbf{Algorithm~\ref{alg:divisionalgorithm}} to factorise $N(X,Y)$ into $V(f(X,Y))$.
\State Generate the codeword $C$ associated to $f(X,Y)$.
\State \textbf{return} $C$
\end{algorithmic}
\end{algorithm}

\subsection{Further comments on radical decoding}
Here we get more into the linear system to be solved in order to show that
if Algorithm~\ref{alg:RadicalDecoding} or Algorithm~\ref{alg:RadicalDecodingFindMin} corrects any error pattern up to a given threshold given by Corollary~\ref{corol:DecodingConditionRadical}, then it is possible to adapt the algorithm and, with a high probability, correct errors $E$ such that 
\[
\wt_{\fibresp_{3}}(E) \leq n-\mu-1.
\]

\begin{Rq} Note that the claimed decoding radius is $d-1$ where $d$ denotes the minimum distance of $\CCC_{\alpha}(\SSS)$ with respect to $\wt_{\fibresp_{3}}(\cdot)$. Hence, this can be up to twice the decoding radius that Algorithms~\ref{alg:fibrewise} and~\ref{alg:fibrewisebis}' can achieve.
Even Algorithm~\ref{alg:fibreWiseTwoDirrs} cannot succeed if the error $E \in \F_{q^n}^{n \times n}$ has too many rows of rank weight greater than $ \frac{n-\mu-1}{2}$.
\end{Rq}

The idea is comparable to that of \cite[Section~4]{BC21} where the decoding of
supercodes of Gabidulin codes is studied. Namely,
let $t\leq n-\mu-1$ and suppose we are given $R = C+E$ where $C \in \CCC_{\alpha}(\mathcal S)$ with $\SSS = \llbracket 0, \mu \rrbracket^2$ and $E \in \F_{q^n}^{n\times n}$ such that $\wt_{\fibresp_3}(E) \leq t$.
Then,
\begin{enumerate}
\item Solve Problem~\ref{pb:Linearisedproblem} with parameters $\SSS$, $t$ and $R$. 
\item Given a nonzero solution $(V(Z), N(X,Y))$ of that problem, compute the kernel $K$ of $V$ (when regarded as an $\F_q$--endomorphism of $\F_{q^n}$);
\item Solve the affine system whose unknowns are the entries of $E \in \F_{q^n}^{n\times n}$ and such that
\begin{equation}\label{eq:final_sys}
    \left\{
    \begin{array}{crl}
        \forall (i_1,i_2) \in \llbracket 1, n \rrbracket^2, & E[i_1,i_2] &\in K\\
        &R - E & \in \CCC_{\alpha}(\SSS).
    \end{array}
    \right.
\end{equation}
\end{enumerate}

By interpolation, there exists $V \in \F_{q^n}[Z]$ of $q$-degree $\leq t$ that vanishes on the entries of $E$. Hence a solution $(V(Z), N(X,Y))$ of Problem~\ref{pb:Linearisedproblem}, with those parameters, exists. What is not guaranteed is that any such solution $(V(Z),N(X,Y))$ satisfies
\[
\forall (i_1,i_2) \in \llbracket 1, n \rrbracket^2,\quad V(E[i_1,i_2]) = 0.
\]

We have the following statement inspired from \cite[Lem.~1]{BC21}.
To state it, let us introduce a notation. For any $V \in \F_{q^n}[Z]$, we denote $V(E) \in \F_{q^n}^{n \times n}$ the matrix $(V(E[i_1,i_2]))_{i \in \llbracket 1,n\rrbracket^2}$. Then, we set
\begin{equation}\label{eq:VE}
    \mathcal{V}_{t,E} := \{V(E) ~|~ V \in \F_{q^n}[Z], \ \qdeg(V) \leq t\}
\end{equation}

\begin{Lemma}
    Suppose that
    \begin{equation}\label{eq:assump_intersection}\CCC_{\alpha}(\SSS + \llbracket 0, t \rrbracket) \cap \mathcal{V}_{t,E} = \{0\},
    \end{equation}
    then any solution $(V(Z),N(X,Y))$ of Problem~\ref{pb:Linearisedproblem}, for given $\SSS$, $t$, and $R$ satisfies $V(E) = 0$.
\end{Lemma}

\begin{proof}
    Let $(V(Z),N(X,Y))$ be a solution of Problem~\ref{pb:Linearisedproblem} with parameters $\SSS$, $t$ and $R$ such that $V(E) \neq 0$. Then, \[V(E) = V(R-C) = N - V(C) \in \CCC_{\alpha}(\SSS + \llbracket 0, t \rrbracket).\]
    This contradicts Assumption \eqref{eq:assump_intersection}.
\end{proof}

Of course, since $E$ is unknown, Assumption \eqref{eq:assump_intersection} is 
not checkable. Still,
it is very likely that 
\[\CCC_{\alpha}(\SSS + \llbracket 0, t \rrbracket) \cap 
    \mathcal{V}_{t,E} = \{0\}\]
as long as the sum of the two spaces remains below $n^2$ (which is the dimension of the ambient space). It turns out that with $t = n-\mu-1$ then
\begin{equation}\label{eq:dimension_condition}
\dim_{\F_{q^n}} \CCC_\alpha (\SSS+ \llbracket 0, t \rrbracket) + \dim_{\F_{q^n}} \mathcal{V}_{t,E} \leq n^2. 
\end{equation}
Indeed, the dimension of $\CCC_\alpha (\SSS+ \llbracket 0, t \rrbracket)$
equals $(\mu+1)^2 + t(2\mu-1)$. The dimension of $\mathcal{V}_{t,E}$ is that
of the codomain of the $\F_{q^n}$--linear map
\[
\begin{array}{ccc}
\{V \in \F_{q^n}[Z],\ \qdeg V \leq t\} & \longrightarrow & \mathcal{V}_{t,E}\\
V & \longmapsto & V(E).
\end{array}
\]
The domain $\{V \in \F_{q^n}[Z],\ \qdeg V \leq t\}$ of the map has $\F_{q^n}$-dimension $t+1$ and since $\wt_{\fibresp_3}(E) \leq t$, the above map has a non trivial kernel. By the rank nullity theorem, we deduce that $\dim_{\F_{q^n}} \mathcal{V}_{t,E}\leq t$.

Therefore, \eqref{eq:dimension_condition} is satisfied if and only if,
\begin{align*}
(\mu+1)^2 + t(2\mu-1) + t & \leq n^2\\
(\mu+1)^2 + 2t\mu & \leq n^2 \\
t & \leq \frac{n^2 - (\mu+1)^2}{2\mu}\\
t & \leq (n-\mu-1) \frac{(n+\mu+1)}{2\mu}
\end{align*}
Since $n > \mu$ then $\frac{n+\mu+1}{2\mu} > 1$ and the above inequality
is actually always satisfied since $t \leq n-\mu-1$. Consequently, condition \eqref{eq:assump_intersection} is very likely to hold and hence whenever $\wt_{\fibresp_3}(E) \leq n - \mu - 1$ our algorithm succeeds with a good probability.

\subsection{Comments on interleaving techniques.} \label{subsec:interleaving}
In \cite{Sidorenko2020OnIR}, for a given integer $L \in {\N}$, and for a list $\mathcal{C}^{(1)},\dots,\mathcal{C}^{(L)}$ of $\mathbb{F}_{q^n}$-linear rank-metric codes of the same length $n$ and having respective dimensions $k^{(1)},\dots,k^{(L)}$, the \textbf{vertically} $L$\textbf{-interleaved codes} of this sequence of codes is the $\mathbb{F}_{q^n}\tm[L \times n,k^{(1)}+\cdots + k^{(L)}]$ matrix code $\mathcal{C}_V$ given by
\begin{equation}\label{eq:CV}\mathcal{C}_V := \left.\left\{ \begin{pmatrix}
    c^{(1)} \\ \vdots \\c^{(L)}\end{pmatrix} \right| \forall \ell \in \llbracket 1,L \rrbracket : c^{(\ell)} \in \mathcal{C}^{(\ell)}
\right\} \subseteq \F_{q^n}^{L\times n}.
\end{equation}

As we stated before, every fibre of a codeword in $\mathcal{C}_{\alpha}(\mathcal{S})$ is a Gabidulin codeword, and more precisely an evaluation codeword of a space of $q$-polynomials. In particular, $\mathcal{C}_{\alpha}(\mathcal{S})$ is contained in such a vertically interleaved code as well as in its transposition (\emph{i.e.} the code of rectangular matrices whose columns are the codewords of the $\mathcal{C}^{(\ell)}$'s) as stated below.

\begin{Prop} \label{prop:InterleavingIsomorphisms}
   Let $\mathcal{S} \subseteq \llbracket 0, n-1\rrbracket^2$. Denote by $C_1 = \{(P(\alpha_1),\dots,P(\alpha_n)) \ | \ P \in \Span_{\mathbb{F}_{q^n}}(\{X^{q^{s_1}} |\ s_1 \in \pi_1(\mathcal{S})\}\})$ and $C_2 = \{(P(\alpha_1),\dots,P(\alpha_n)) \ | \ P \in \Span_{\mathbb{F}_{q^n}}(\{Y^{q^{s_2}} |\ s_2 \in \pi_2(\mathcal{S})\}\})$ the evaluation (vector) codes of linearised polynomials with support in $\pi_1(\mathcal{S})$ and $\pi_2(\mathcal{S})$ respectively.  
    \begin{enumerate}
        \item The code $\mathcal{C}_{\alpha}(\mathcal{S})$ is a subcode of the transpose of the vertically $n$-interleaved code for $C_1 = \CCC^{(1)} = \cdots = \CCC^{(n)}$. Moreover, the two codes are equal if and only if $\mathcal{S} = \pi_1(\mathcal{S}) \times \llbracket 0,n-1\rrbracket$. 
        
        \item The code $\mathcal{C}_\alpha(\mathcal{S})$ is a subcode of the vertically $n$-interleaved code for $C_2 = \CCC^{(1)} = \cdots = \CCC^{(n)}$. Moreover, the two codes are equal if and only if $\mathcal{S} =  \llbracket 0,n-1\rrbracket \times \pi_2(\mathcal{S})$. 

        \item Let $\mu_1,\mu_2 \in \llbracket 0,n-1\rrbracket$. Assume that $\mathcal{S} = \llbracket 0,\mu_1\rrbracket \times \llbracket 0,\mu_2\rrbracket$. Then $\mathcal{C}_\alpha(\mathcal{S})$ is $\mathbb{F}_{q^n}$-isomorphic to the vertically $(\mu_2 + 1)$-interleaved code for $\GG_{\mu_1+1}(\alpha) = \CCC^{(1)} = \cdots = \CCC^{(\mu_2 + 1)}$, as well as to the vertically $(\mu_1 + 1)$-interleaved code for $\GG_{\mu_2+1}(\alpha) = \CCC^{(1)} = \cdots = \CCC^{(\mu_1 + 1)}$.

        \item The map $\psi: \F_{q^n}^{n\times n}\times \F_{q^n}^{2n^2}, C \mapsto [C[1,:],\dots,C[n,:],C[:,1],\dots,C[:,n]]$ is a monomorphism such that $\psi(\CCC_{\alpha}(\SSS))$ is a subcode of the horizontal $2n$-interleaved code for $\CCC^{(1)} = \cdots = \CCC^{(n)} =  \GG_{\mu_2+1}(\alpha)$ and  $\CCC^{(n+1)} = \cdots = \CCC^{(2n)}= \GG_{\mu_1+1}(\alpha)$, and such that $\rank_{\F_q} \psi(C) = \wt_{\fibresp_{3}}(C)$ for each $C \in \F_{q^n}^{n\times n}$.
    \end{enumerate}  
    \end{Prop}

\begin{proof}
    Let $C\in\mathcal{C}_\alpha(\mathcal{S})$. Proposition~\ref{prop:GabidulinColumns} indicates that the columns of $C$ are codewords of $C_1$ and that the rows of $c$ are codewords of $C_2$. Therefore, we have the inclusions mentioned in the first and second point. Moreover, the last point is also direct consequence of Proposition~\ref{prop:GabidulinColumns}. Additionally, the dimension of the vertically $n$-interleaved code of the repetition of $C_i$ is $n |\pi_i(\mathcal{S})|$, for $i\in \{1,2\}$. So $\mathcal{C}_{\alpha}(\mathcal{S})$ is exactly the vertically $n$-interleaved code if and only if $|\mathcal{S}| = n |\pi_1(\mathcal{S})|$ (\emph{resp.} $|\mathcal{S}| = n |\pi_2(\mathcal{S})|$) and this is equivalent to $\mathcal{S} = \pi_1(\mathcal{S}) \times \llbracket 0,n-1\rrbracket$ (\emph{resp.} $\mathcal{S} = \llbracket 0,n-1\rrbracket \times \pi_2(\mathcal{S})$). Now assume that $\mathcal{S} = \llbracket 0,\mu_1\rrbracket \times \llbracket 0,\mu_2\rrbracket$.     Recall that $M_{k+1}(\alpha)(M_{k+1}(\alpha^\bot)^\top)=I_{k+1}$.     Let $f(X,Y) = \sum_{i_1 = 0}^{\mu_1}\sum_{i_2 = 0}^{\mu_2} f_{(i_1,i_2)}X^{q^{i_1}}Y^{q^{i_2}} \in \mathscr{M}_{q,\mathbb{F}_{q^n}}[X,Y]$ with $\Supp(f) \subseteq \mathcal{S}$ be associated to $C$.         Then
    \begin{align*}
        C = \begin{bmatrix}
        f(\alpha_1,\alpha_1) & \hdots & f(\alpha_1,\alpha_n) \\
        \vdots & \ddots & \vdots \\
        f(\alpha_n,\alpha_1) & \hdots & f(\alpha_n,\alpha_n)
    \end{bmatrix} &= M_{\mu_1+1}(\alpha)^\top \begin{bmatrix}
        f_{(0,0)} & \hdots & f_{(0,\mu_2)} \\
        \vdots & \ddots & \vdots \\
        f_{(\mu_1,0)} & \hdots & f_{(\mu_1,\mu_2)}
    \end{bmatrix}M_{\mu_2+1}(\alpha) \\
    &= \begin{bmatrix}
        \sum_{i_1 = 0}^{\mu_1}f_{(i_1,0)}\alpha_1^{q^{i_1}} & \hdots & \sum_{i_1 = 0}^{\mu_1}f_{(i_1,\mu_2)}\alpha_1^{q^{i_1}} \\
        \vdots & \ddots & \vdots \\
        \sum_{i_1 = 0}^{\mu_1}f_{(i_1,0)}\alpha_n^{q^{i_1}} & \hdots & \sum_{i_1 = 0}^{\mu_1}f_{(i_1,\mu_2)}\alpha_n^{q^{i_1}} \\
    \end{bmatrix}M_{\mu_2+1}(\alpha). 
    \end{align*}
    Therefore, the morphism $\phi : \mathcal{C}_\alpha(\mathcal{S}) \to \mathcal{C}_V,\ C \mapsto  M_{\mu_2+1}(\alpha^\bot) C^\top$ (where $\mathcal C_V$ is the vertically interleaved code of \eqref{eq:CV})
    is well-defined and bijective and hence is an isomorphism. By a similar argument, the map $C \mapsto M_{\mu_1+1}(\alpha^\bot)C$ yields an isomorphism between $\mathcal{C}_\alpha(\mathcal{S})$ and the $(\mu_1+1)$-interleaved code of the repetition of $\GG_{\mu_2 +1}(\alpha)$. 
\end{proof}

\begin{Rq}
    Note that the isomorphism $\phi: \F_{q^n}^{n\times n} \to \F_{q^n}^{n\times n}, M \mapsto M_{\mu_2 +1}(\alpha^\bot) M^\top$, from which is derived the isomorphism explained in the third point of Proposition~\ref{prop:InterleavingIsomorphisms}, is not an isometry for the metric given by $\wt_{\fibresp_3}$. In other words, there exists matrices $M\in \F_{q^n}^{n\times n}$ such that $\wt_{\fibresp_3}(M) < \wt_{\fibresp_3}(\phi(M))$, \emph{e.g.} we have $\wt_{\fibresp_3}(I_n) = 1$ while $\wt_{\fibresp_3}(M_{\mu_2 +1}(\alpha^\bot)) =n$ as the Moore matrix has the $\F_q$-basis $\alpha^\bot$ as its first row. Obviously, this is not the case of the inclusion maps described in the first two points.
\end{Rq}

\begin{Rq}
    Assume that $\mu = \mu_1 = \mu_2$, the embedding $\varphi$ of Proposition \ref{prop:InterleavingIsomorphisms} 4. shows that it is possible to apply interleaved decoder on the subcode $\CCC_\alpha(\SSS)$ of an interleaved Gabidulin code, see \cite{Sidorenko2020OnIR}. Given $C \in \CCC_\alpha(\llbracket 0,\mu\rrbracket^2)$, an interleaved Gabidulin code decoder can correct any error $E\in \F_{q^n}^{n\times n}$ on a received word $R = C+E$ such that $\wt_{\fibresp_3}(E) \leq \frac{1}{2}(n-\mu -1)$ and furthermore it can correct an error $E \in \F_{q^n}^{n\times n}$ such that $\wt_{\fibresp_3}(E) \leq \frac{2n}{2n+1}(n - \mu -1)$ with low     failure probability, see \cite[Theorem 5]{Sidorenko2020OnIR} and \cite[Section~VIII]{SidorenkoJiangBossert2011SkewFeedbackAndInterleaved}. We will discuss the comparison of that technique with the previous ones in the next section.

    \end{Rq}

\section{Comparisons} \label{sec:Comparison}

\subsection{Roth algorithms} \label{sec:RothCodes}

In his paper \cite{ROTHTensorCodesForRankMetric}, Roth gave two algorithms to decode the $3$-tensor codes ``$\mathcal{C}(n,3,3;q)$'' and ``$\mathcal{C}(n,5,3;q)$'' for errors respectively of tensor-rank at most one and two, and introduced the generalisation of the second for higher tensor-rank correction for suitable codes. In terms of computations over the field $\F_q$, the complexity of the first algorithm is polynomial in $n$, while the complexity of the second and of its generalisations is polynomial in $n$ and exponential in the correction radius. For the case of order-3 tensors, the codes are respectively isomorphic to the codes $\CCC_\alpha(\SSS_1)$ and $\CCC_{\alpha}(\SSS_2)$ with $\SSS_1\supseteq \{s \in \llbracket 0,n-1\rrbracket^{2} \ | \ s_1 + s_{2} > 1  \}$ and $\SSS_2 \supseteq  \{s \in \llbracket 0,n-1\rrbracket^2 \ | \ s_1 + s_2 > 3  \}$. Since the sets $\SSS_1$ and $\SSS_2$ have maximum $n-1$ in each direction, running Algorithms~\ref{alg:fibrewise} or~\ref{alg:RadicalDecoding} would result in no error correction capability. We summarize the correction capabilities of the codes as follows.

\begin{itemize}
    \item The code {$\mathcal{C}(n,3,3;q)$} has $\mathbb{F}_{q^n}$-dimension at least $n^2 - 3$ and can correct every error of tensor-rank at most $1$ in $(\mathbb{F}_{q}^{n})^{\otimes 3}$, in particular the decoding algorithm can correct at least $\frac{(q^n -1)^3}{(q-1)^2} + 1$ different errors. 
    \item The code {$\mathcal{C}(n,5,3;q)$} has $\mathbb{F}_{q^n}$-dimension at least $n^2 - 10$ and can correct every error of tensor-rank at most $2$ in $(\mathbb{F}_{q}^{n})^{\otimes 3}$, in particular the decoding algorithm can correct at least 
    \begin{equation}
        \frac{q(q^n-1)^3(q^{n-1}-1)^2}{(q-1)^3(q^2-1)} \left(\frac 12 q^2(q+1)(q^{n-1}-1) + 3(q-1)\right) + \frac{(q^n -1)^3}{(q-1)^2} + 1
        \label{eq:NbTrk2}
    \end{equation} 
    different errors.
\end{itemize}

The number of tensors of tensor-rank one in $(\mathbb{F}_q^{n})^2$ is easily obtained, while (\ref{eq:NbTrk2}) is a particular case of Lemma~\ref{lemma:CountingNbTrank2} below, (see appendix~\ref{proof:CountingNbTrank2} for a proof). In the following sections, we will compare those quantities to the known lower-bounds on the number of correctable errors with respect to the algorithms described previously.

\begin{Lemma} \label{lemma:CountingNbTrank2}
    Let $k,m,n \geq 2$. In the space of (non-necessarily square) tensors $\F_{q}^k \otimes \mathbb{F}_q^{m} \otimes \F_q^n$, the number of elements of tensor-rank exactly two is
    $$\textstyle \frac{q(q^n-1)(q^m-1)(q^k-1)}{(q-1)^3(q^2-1)} \left(\frac{(q^{n-1} - 1)(q^{m-1}-1)(q^{k-1}-1)q^2(q+1)}{2} + (q-1)\left(\frac{q^{n + m} +q^{k+n} + q^{k+m}}{q^2} -2 \frac{q^{k} + q^{m} +q^{n}}{q} +3 \right)\right).$$
\end{Lemma}

\subsection{Comparison of the fibre-wise decoders}
As mentioned in Section~\ref{location:numbersofcorrectederrors}, the first version of the fibre-wise decoder, \emph{i.e.} column- and row-wise decoders, can be upscaled into a second version that corrects more errors at the cost of doubling the complexity. Since the number of matrices of a given size and a given rank over a given field is known (see \cite[Theorem 25.2]{Van_Lint_A_Course_In_Combinatorics}), we can compute the aforementioned lower-bounds on the number of correctable errors by the respective algorithms and compute the improvement. In Figure~\ref{fig:ImprovBoundFibreDec} we display the logarithm of $N_2/N_1$ for different values of $\mu_1$ and $\mu_2$ and for the choice of $n = 10$ and $q= 2$, where $N_j$ is the known lower-bound on the number of errors that Algorithm $j$ can correct on the code $\CCC_\alpha(\llbracket 0,\mu_1\rrbracket\times \llbracket 0,\mu_2 \rrbracket)$.

\begin{Rq} 
Let $\kappa \in \llbracket 1,n-1\rrbracket$. We have $\trank_{\F_q}(T) \geq \min_{\substack{\JJJ ,|\JJJ| = \kappa }} \max_{i_2 \in \JJJ} \ \rank_{\mathbb{F}_q} T[:,i_2]$ for each $T \in \mathbb{F}_{q^n}^{n\times n}$, as a consequence of Corollary~\ref{corol:maxrankslicesvstrank}. Note that this inequality is met for any choice of $\kappa \in \llbracket 1,n\rrbracket$. For instance, consider $T = [c^\top,\dots,c^\top]$ where $c \in \F_{q^n}^n$ is non-zero, then $\trank_{\F_q}(T) = \rank_{\F_q}(c)$ and the right-hand side of the inequality is also equal to $\rank_{\F_q}(T)$ as the ranks of all columns is identical.

Moreover, consider $T_1,T_2 \in \mathbb{F}_{q^n}^{n \times n}$ defined for each $(i_1,i_2) \in \llbracket 0,n-1\rrbracket$ by $T_{1,(i_1,i_2)} = \alpha_{i_1}\ind{i_1\leq a}\ind{i_2 \leq n - \kappa}$  and $T_{2,(i_1,i_2)} = \alpha_{i_1}\ind{i_1\leq a}$. Then both have tensor-rank exactly $a$, and the corresponding values of the right-hand side term respectively are $0$ and $a$. 
\end{Rq}

Since Roth codes introduced in Section~\ref{sec:RothCodes} have in general dimension strictly larger than any Roth-tensor code $\CCC_\alpha(\SSS)$ for which at least one of the algorithms above corrects non-trivial errors, we can compare the former with the largest Roth-tensor codes for which Algorithms~\ref{alg:fibrewise} and~\ref{alg:fibreWiseTwoDirrs} can decode non-trivial errors.
\begin{Prop} \label{prop:NumberOfCorrectableErrorsMethods}
    Consider the algorithms and codes in $\mathbb{F}_{q^n}^{n\times n}$ introduced before, where  $n \geq 3$. Figure~\ref{fig:AssymptoticComparisionRothAndSimpleAlgs} lists their respective dimensions, as well as the asymptotic behaviour of the known lower-bounds of their respective sets of correctable errors.

\end{Prop}

\begin{table}
    \centering
    \begin{tabular}{|c|c|c|c|c|}\hline
        \textbf{Algorithm}&\makecell{Algorithm~\ref{alg:fibrewise} on\\$\mathcal{C}_\alpha(\llbracket 0,n-1\rrbracket \times \llbracket 0,n-3\rrbracket)$} & \makecell{Algorithm~\ref{alg:fibreWiseTwoDirrs}\\ on $\mathcal{C}_\alpha(\llbracket 0,n-3\rrbracket^2)$} &
        \makecell{Roth decoder\\on $\mathcal{C}(n,5,3,q)$} & \makecell{Roth decoder\\on $\mathcal{C}(n,3,3,q)$}\\ \hline
        \makecell{\textbf{Dimension of}\\\textbf{the code}} & $n(n-3)$ & $(n-3)^2$ & $ \geq n^2 - 10$ & $ \geq n^2 - 3$\\\hline
        \makecell{\textbf{Asymptotic}\\\textbf{lower bound on}\\\textbf{number of}\\\textbf{correctable}\\\textbf{errors}}& $q^{2n^2} (q-1)^{-n}$ & $n q^{3n^2 - 2n} (q-1)^{1-n}$ &  $\frac{(q+1) q^{6n}}{2(q-1)^3(q^2-1)}$ & $\frac{q^{3n}}{(q-1)^2} $ \\\hline
    \end{tabular}
    \caption{Comparison of the decoding algorithms.}
    \label{fig:AssymptoticComparisionRothAndSimpleAlgs}
\end{table}

\begin{proof}
    The dimensions of the codes have been specified before.         For each $n\in \N$, denote by $S_n := 1 + \frac{(q^n -1)^2}{q-1}$ the number of vectors in $\F_{q^n}^{n}$ of $\F_q$-rank at most one. Since Algorithm~\ref{alg:fibrewise} on $\mathcal{C}_\alpha(\llbracket 
    0,n-1\rrbracket\times \llbracket 0,n-3\rrbracket)$ can correct any error in $\F_{q^n}^{n\times n}$ whose rows have all rank at most one, a lower bound on the number of errors that are correctable using  Algorithm~\ref{alg:fibrewise} on $\mathcal{C}_\alpha(\llbracket 0,n-3\rrbracket^2)$ is given by $T_n^{(1)} = S_n^n$. Likewise, since Algorithm~\ref{alg:fibreWiseTwoDirrs} on $\mathcal C_\alpha(\llbracket 0,n-3\rrbracket^2)$ can correct any error in $\F_{q^n}^{n\times n}$ that has at least $n-1$ rows that have $\F_q$-rank at most one, a lower bound on the number of errors that are correctable by this algorithm is $T^{(2)}_n := n S_n^{n-1} (q^{n^2}- S_n) + S_n^n$. Note that
    $$S_n = \frac{q^{2n}}{q-1}(1 - q^{-n}(2 - q^{1-n})).$$
    Since $1-2q^{-n} \leq 1 - q^{-n}(2 - q^{1-n}) \leq 1-q^{-n}$ for $n$ sufficiently large, and since the limits of $(1-2q^{-n})^n$, of $(1-q^{-n})^n$, of $(1-2q^{-n})^{n-1}$ and of $(1-q^{-n})^{n-1}$ all exist and are equal to $1$ as $n$ tends to infinity, we have the following limits.
    $$\lim_{n \to \infty} (1 - q^{-n}(2 - q^{1-n}))^{n-1} =\lim_{n \to \infty} (1 - q^{-n}(2 - q^{1-n}))^{n} = 1.$$
    As a consequence, $T^{(1)}_n = S_n^n \underset{n \to \infty}{\sim} \frac{q^{2n^2}}{(q-1)^n}$. Additionally, since $T_n =nS_n^{n-1}(q^{n^2} -S_n + n^{-1}S_n)$ and since $S_n$ is asymptotically dominated by $q^{n^2}$, we obtain the following equivalences.
    $$T^{(2)}_n \underset{n \to \infty}{\sim} n S_n^{n-1} q^{n^2}  \underset{n \to \infty}{\sim} n \frac{q^{2n(n-1)}}{(q-1)^{n-1}}q^{n^2} = \frac{nq^{3n^2 - 2n}}{(q-1)^{n-1}}.$$
    The asymptotic equivalents for the lower bound of the number of correctable errors by Roth decoders can be computed from the formulas given in Section~\ref{sec:RothCodes}.
\end{proof}

\begin{Rq}\label{rq:ContextualisationTabular} Note that the codes given by Roth studied in Figure~\ref{fig:AssymptoticComparisionRothAndSimpleAlgs} have a lower redundancy than the codes considered for Algorithms~\ref{alg:fibrewise} and~\ref{alg:fibreWiseTwoDirrs}, and that it is possible to find a subcode of $\CCC(n,3,3;q)$ and an algorithm induced by Roth's decoding algorithm on the ambient code that has a higher asymptotic lower bound on the number of correctable errors than the one of Algorithm~\ref{alg:fibrewise} on $\CCC_\alpha(\llbracket 0,n-1\rrbracket \times \llbracket 0,n-3\rrbracket)$, while having the same redundancy. We outline our reasoning below.

Let $n\geq 4$ and denote by $\III = \{ (i_1,i_2) \ | \ i_1= n \text{ or } (i_1 = n-1 \text{ and } i_2 > 3)\}$. Note that $|\III|=2n-3$. Consider $\CCC^\star$, the $\F_{q^n}$-subcode of $\CCC(n,3,3;q)$ defined as 
\[
\CCC^\star = \{ C \in \CCC(n,3,3;q) \ | \ C[i] = 0, \forall i \in \III \}.
\]
Let us show that $\dim_{\F_q} \mathcal C^\star = n^2 -2n$. The code $\mathcal C^\star$ is the kernel the map $\mathcal C(n,3,3;q) \to \F_{q^n}^{\III}, C \mapsto (C[i])_{i \in \III}$, we will show that this map is surjective. Let $(y_i)_{i \in \III} \in \F_{q^n}^{\III}$ and consider the polynomial 
\[
f(X,Y) = L_{n-1}(Y)^{q^{n-2}} f_n(X) +  L_{n}(Y)^{q^{n-2}} f_{n-1}(X),
\]
where $L_{n-1}(Y), L_n(Y) \in \F_{q^n}[Y]$ and $f_n(X), f_{n-1}(X) \in \F_{q^n}[X]$ are linearised $q$-polynomials satisfying the following properties.
\begin{itemize}
    \item $L_{n-1}(Y) = \prod_{\lambda \in \F_q} (Y - \lambda \alpha_n)(\alpha_{n-1} - \lambda \alpha_n)^{-1}$ is of $q$-degree one with $L_{n-1}(\alpha_{n-1}) = 1$ and $L_{n-1}(\alpha_{n}) = 0$,
    \item $L_{n}(Y) = \prod_{\lambda \in \F_q} (Y - \lambda \alpha_{n-1})(\alpha_{n} - \lambda \alpha_{n-1})^{-1}$ is of $q$-degree one with $L_{n}(\alpha_{n-1}) = 0$ and $L_{n}(\alpha_{n}) = 1$,
    \item $f_{n-1}(X)$ is a $q$-polynomial of $q$-degree at most $n-4$ that interpolates $(\alpha_{4},y_{n-1,4}),\dots,(\alpha_n,y_{n-1,n})$, \emph{i.e.} $f_{n-1}(\alpha_{i_2}) = y_{(n-1,i_2)}$ for each $i_2 \in \llbracket 4,n\rrbracket$,
    \item $f_{n}(X)$ is a $q$-polynomial of $q$-degree at most $n-1$ that interpolates $(\alpha_{1},y_{n,1}),\dots,(\alpha_n,y_{n,n})$, \emph{i.e.} $f_{n}(\alpha_{i_2}) = y_{(n,i_2)}$ for each $i_2 \in \llbracket 1,n\rrbracket$.
\end{itemize}
Since there exists such polynomials, and since the $q$-polynomial $f(X,Y)$ has no monomial of the form $XY$, $X^q Y$ or $XY^q$, its associate codeword $C$ is a codeword in $\CCC(n,3,3;q)$. Therefore, we have $C[i] = f(\alpha_{i_1},\alpha_{i_2})$ for each $i \in \III$. As a consequence, the codimension of $\CCC^\star$ in $\CCC(n,3,3;q)$ is $2n-3$, and so $\dim_{\F_{q^n}}\CCC^\star = n^2 - 2n$.

Let $C \in \CCC^\star$ be a codeword and let $R = C + E$ with a certain $E \in \F_{q^n}^{n\times n}$. Suppose that $\trank_{\F_q}(E) = 1$ and write $E = E^{(1)} + E^{(2)}$ for $E^{(i)} \in \F_{q^n}^{n\times n}$ such that $E^{(1)}[i_1,i_2] = 0$ for each $i \in \llbracket 1,n\rrbracket^2$ with $i_1 \in \{n-1,n\}$ and $E^{(2)}[i_1,i_2] = 0$ for each $i \in \llbracket 1,n\rrbracket \backslash \III$. Then it is possible to recover $(C,E)$ using the following algorithm: one first deletes the entries of $R$ with indices in $\III$ to recover $E^{(2)}$, and then runs Roth's decoding algorithm on the resulting matrix $R' = C + E^{(1)}$ to recover $C$ and $E^{(1)}$, which is possible since $\trank_{\F_q}(E^{(1)}) = 1$. The number of errors that can be corrected by this algorithm is at least $\frac{q^{n + (n-2) + n}}{(q-1)^2} (q^n)^{|\III|} = \frac{q^{3n - 2 + 2n^2 - 3n}}{(q-1)^2} = q^{2n^2 - 2}(q-1)^{-2}$.\end{Rq}

\subsection{Comparison between fibre-wise and radical decoding.}

\begin{Prop}\label{prop:comparisonfibreradical} With a given code $\mathcal{C}_\alpha(\llbracket 0,\mu\rrbracket^2)$, we have the following observations.
    \begin{enumerate}
        \item There exists $E \in \F_{q^n}^{n\times n}$ satisfying equation (\ref{eq:CriterionFibreWiseNaive}) (with a given $j \in \{1,2\}$) but not equation (\ref{eq:ConditionRadicalDecoding}).
        \item If $E \in \F_{q^n}^{n\times n}$ satisfies equation (\ref{eq:ConditionRadicalDecoding}), then there exists a $j \in \{1,2\}$ such that $E$ will satisfy equation (\ref{eq:CriterionFibreWiseNaive}) for this choice of $j$. Similarly, there exists elements $E \in \F_{q^n}^{n\times n}$ satisfying equation (\ref{eq:ConditionRadicalDecoding}) but not equation (\ref{eq:CriterionFibreWiseNaive}) for at least one $j \in \{1,2\}$. 
    \end{enumerate}
\end{Prop}

\begin{proof} Let $\mu \in \llbracket 0,n-2\rrbracket$ and define $a = n -\mu -2$. Let $E \in \F_{q^n}^{n\times n}$ be such that $\wt_{\fibresp_{3}}(E) +  \min_{j \in \{1,2\}}\wt_{\slicesp_j}(E) \leq n - \mu -1$. If $\wt_{\fibresp_{3}}(E) +  \wt_{\slicesp_1}(E) \leq n - \mu -1$, then $\wt_{\fibresp_3}(E) \leq \left\lfloor\frac{n-\mu-1}{2}\right\rfloor$ or $\wt_{\slicesp_1}(E) \leq \left\lfloor\frac{n-\mu-1}{2}\right\rfloor$. In the first case, equation (\ref{eq:CriterionFibreWiseNaive}) is satisfied for both $j$, and in the second case, Lemma~\ref{lemma:maxrankandslicespaceweight} shows that equation (\ref{eq:CriterionFibreWiseNaive}) is satisfied for $j=1$. A similar reasoning when $\wt_{\fibresp_{3}}(E) +  \wt_{\slicesp_2}(E) \leq n - \mu -1$ concludes the proof of first part of the second point. Additionally, consider the following two matrices in $\F_{q^n}^{n\times n}$.
$$E_1 = I_n \quad \text{ and } \quad E_2 = \begin{bmatrix}
    \alpha_1 & \cdots & \alpha_1 \\
    \vdots & \ddots & \vdots \\
    \alpha_a & \cdots & \alpha_a \\
    &&\\
    &0& \\
    &&
\end{bmatrix},$$
where we recall that $\alpha_1, \dots, \alpha_n$ are $\F_q$--linearly independent.
Then we have $\wt_{\slicesp_1}(E_1) = \wt_{\slicesp_2}(E_1) =  n$, $\wt_{\fibresp_3}(E_1) = 1$, and the rank of every row and column of $E_1$ is $1$. Additionally, we have $\wt_{\slicesp_1}(E_2) = a$, $\wt_{\slicesp_2}(E_2) = 1$, $\wt_{\fibresp_3}(E_2) = a$, the rank of every column of $E_2$ is $a$, and the rank of every row of $E_2$ is $1$. 
\end{proof}
\begin{Rq}
    The first part of the second point in the proposition above can be rephrased as the following~: if an error $E$ satisfies equation (\ref{eq:ConditionRadicalDecoding}), and with a fixed $j$, $E$ or its (\emph{resp.} any of its, for the generalised case below) transpose is correctable by Algorithms~\ref{alg:fibrewise},~\ref{alg:fibreWiseTwoDirrs}.
\end{Rq}

In particular, the two algorithms have different sets of correctable errors, and one can verify it experimentally as well. We will now compare the computational asymptotic complexity of both algorithms. The computational complexity of the column-wise decoder (Algorithm~\ref{alg:fibrewise}) is $n$ times the complexity of the Gabidulin decoder for a code of length $n$ which requires a number of computations on $\mathbb{F}_{q}$ of the order of $n^3\log n$ \cite{Wachter-Zeh2013}, yielding a total complexity of $\mathcal{O}(n^4\log n)$ operations on $\F_q$. On the other hand, finding all solutions of Problem~\ref{pb:Linearisedproblem} for a given $\SSS$, $t$ and $R$ corresponds to the resolution of a linear system of $n^2$ equations over $\mathbb{F}_{q^n}$ with $|\mathcal{S} + \llbracket 0,t\rrbracket|$ unknowns in $\mathbb{F}_{q^n}$. One can verify that for $\mathcal{S} = \llbracket 0,\mu \rrbracket^2$ we have $|\mathcal{S} + \llbracket 0,t \rrbracket| = (\mu +1)^2 + (2\mu + 1)t$. Indeed, except the $(\mu+1)^2$ points in $\mathcal{S}$ there are $t$ points on $(2\mu +1)$ segments corresponding to the points $(x,\mu)$ and $(\mu,y)$ in $\mathcal{S}$, with $x,y \in \llbracket 0,\mu\rrbracket$. Hence, solving the problem sums up to solving a system of $n^3$ equations and $n((\mu +1)^2 + (2\mu + 1)t)$ unknowns on $\mathbb{F}_q$. Consequently, solving this system has an asymptotic complexity $\mathcal{O}(n^5\mu^4)$ or $\mathcal{O}(n^7\mu^2)$ respectively if $t \leq \frac{n^2-(\mu+1)^2}{2\mu +1}$ or not, and in particular $\mathcal{O}(n^9)$. Moreover, the search for the smallest $\delta$ in Corollary \ref{corol:findsmallestsupport} in the computed set of solution of the problem is equivalent to compute at most $n$ intersections of $\F_{q^n}$-vector spaces of dimension at most $n^2$ each, the complexity of Algorithms \ref{alg:RadicalDecoding} and \ref{alg:RadicalDecodingFindMin} are both $\OOO(n^9)$; see \cite[Chapter 3]{golub2013matrix}.

\subsection{Discussion on the tensor-rank}

As shown in Section~\ref{subsec:distancesinthecode}, for each $T \in \F_{q^n}^{n\times n}$, the weights $\wt_{\slicesp_1}(T)$, $\wt_{\slicesp_2}(T)$ and $\wt_{\fibresp_3}(T)$ as well as the rank of any column or row of $T$ are all bounded from above by $\trank_{\F_q}$. Similarly, we have $\wt_{\Sigma\!\slicesp}(T)\leq 2\trank_{\F_{q}}(T)$. In particular, if a code has minimum distance at least $d$ when endowed with the metrics $d_{\slicesp_1}$, $d_{\slicesp_2}$, $d_{\fibresp_3}$, $d_{\rk_1}$, or $d_{\rk_2}$, then it has minimum distance at least $d$ when endowed with the tensor-rank as a metric, \emph{i.e.} with respect to the metric $d_{\trank}$ (see Sections~\ref{subsec:distancesinthecode} and~\ref{subsec:decodingfibrewise}). We obtain the same conclusion if a code has minimum distance $2d$ when endowed with  $\wt_{\Sigma\!\slicesp}$ as a metric. In terms of decoding, we obtain the following statement.

\begin{Corol} 
    Let $\mu \in \llbracket 0,n-1\rrbracket$. We have $d_{\trank}(\CCC(\llbracket0,\mu\rrbracket^2) \geq n-\mu$. If $R = C + E$ with $C \in \CCC(\llbracket 0,\mu\rrbracket^2)$ and $E \in \F_{q^n}^{n\times n}$ such that $\trank_{\F_q}E\leq \left\lfloor\frac{n-\mu-1}{2}\right\rfloor$, then Algorithms~\ref{alg:fibrewise},~\ref{alg:fibrewise}',~\ref{alg:fibreWiseTwoDirrs} and~\ref{alg:RadicalDecodingFindMin} all return $C$. \end{Corol}
\begin{proof}
    This is a consequence of Proposition~\ref{prop:weightsandtensorrank}, Corollary~\ref{corol:maxrankslicesvstrank} and the results on the algorithms in Sections~\ref{subsec:decodingfibrewise} and~\ref{subsec:decodinglinearisedproblem}.
 \end{proof}

 By construction, Algorithms~\ref{alg:fibrewise},~\ref{alg:fibrewise}',~\ref{alg:fibreWiseTwoDirrs} and~\ref{alg:RadicalDecodingFindMin} can correct errors above the tensor-rank radius given in the corollary. Proposition~\ref{prop:NumberOfCorrectableErrorsMethods} provides an example illustrating the additional errors corrected by particular cases of Algorithms~\ref{alg:fibrewise},~\ref{alg:fibrewise}' and~\ref{alg:fibreWiseTwoDirrs}. Regarding the additional errors corrected by~\ref{alg:RadicalDecodingFindMin}, with the technical lemma below proven in appendix~\ref{proof:UBontensorofrankatmost} we obtain the asymptotic behaviour stated in Proposition~\ref{prop:assymptoticcountingextraerrors}.

\begin{Lemma} \label{lemma:UBontensorsofrankatmost}
    Let $n \geq 2$ and let $R \geq 2$. Any element in $(\mathbb{F}_q^n)^{\otimes 3}$ of tensor-rank at most $R$ can be written as a sum of $R$ elements of tensor-rank one. In particular, the number of elements in $(\mathbb{F}_q^n)^{\otimes 3}$ of tensor-rank at most $R$ is bounded from above by $\frac{(q^n-1)^{3R}}{(q-1)^{2R}}$. 
\end{Lemma}

\begin{Prop}\label{prop:assymptoticcountingextraerrors}
    Let $(R_n)_{n\in {\N_0}}$ and $(\mu_n)_{n \in {\N_0}}$ be non-negative integer sequences with 
    $\mu_n = \lceil{n-1-2R_n}\rceil$ and $R_n \geq 2$ for each $n \in \N\backslash \{1\}$. Then     \begin{equation}
        \label{eq:MinNumberofTensorswithinssdim}\left| \left\{ E\in\F_{q^n}^{n \times n} \ | \ \wt_{\slicesp_1}(E) \leq R_n, \wt_{\slicesp_2}(E) \leq n,\wt_{\fibresp_3}(E) \leq R_n  \right\}\right| \geq q^{n{R_n}^{2}}.
    \end{equation}
    In particular, assuming that $R_n \geq \frac{3+\sqrt{13}}{2}$ for sufficiently large $n$, using the ``big omega'' notation as in \cite[1.2.11.1 (20)]{knuth1997fundamental}, we have
    \begin{equation}
        \label{eq:AssBhvrAdditionalErrors}\left| \left\{ E\in\F_{q^n}^{n \times n} \ | \ \trank(E) > R_n , \text{$\forall C \in\CCC(\llbracket 0,\mu_n\rrbracket^2) : $ \texttt{Alg~\ref{alg:RadicalDecodingFindMin}}($C+E$) $= C$ with parameter $\mu_n$}\right\}\right| = \Omega(q^{nR_n^2}).
    \end{equation}
\end{Prop}
\begin{proof}
    Note that $\mu_n = \lceil{n-1-2R_n}\rceil$ is equivalent to $R_n = \lfloor \frac{n-\mu_n-1}{2}\rfloor$, which implies $\mu_n \in \llbracket 0,n-1\rrbracket$ for each $n \in \N$. In particular, any element $E$ satisfying $\wt_{\slicesp_1}(E) \leq R_n$ and $\wt_{\fibresp_3}(E) \leq R_n$ is such that $\wt_{\Sigma\slicesp}(E) \leq 2R_n \leq n-\mu_n -1$ and thus is a correctable error with Algorithm~\ref{alg:RadicalDecodingFindMin} on the code $\CCC(\llbracket 0,\mu_n\rrbracket^2)$.
    
     Consider $\SSS := \left\{ E\in\F_{q^n}^{n \times n} \ | \ \wt_{\slicesp_1}(E) \leq R_n,\wt_{\fibresp_3}(E) \leq R_n  \right\}$. Since the transposition map $t:(\mathbb{F}_q^n)^{\otimes 3} \to (\mathbb{F}_q^n)^{\otimes 3}, \Gamma \mapsto (\Gamma[i,\ell,j])_{i,j,\ell}$ swapping the second and third component is a bijection with $\dim \slicesp_1(t(\Gamma)) = \dim \slicesp_1(\Gamma)$, with $\dim \slicesp_2(t(\Gamma)) = \dim \slicesp_3(\Gamma)$ and with $\dim \slicesp_3(t(\Gamma)) = \dim \slicesp_2(\Gamma)$, then the set $\SSS$ as same cardinality as the set $\Tilde\SSS := \left\{ E\in\F_{q^n}^{n \times n} \ | \ \wt_{\slicesp_1}(E) \leq R_n, \wt_{\slicesp_2}(E) \leq R_n  \right\}$.  
     
   Let $s_1,s_2 \in \llbracket 1,R_n\rrbracket$ be integers. Consider the injective map $\phi_{s_1,s_2} : \F_{q^n}^{s_1 \times s_2} \to \F_{q^n}^{R_n \times R_n}, A \mapsto \left[\begin{smallmatrix}
         A&0\\0&0
     \end{smallmatrix}\right]$. Denote by $\AAA_{s_1,s_2} := \left\{ E\in\F_{q^n}^{n \times n} \ | \ \wt_{\slicesp_1}(E) = s_1, \wt_{\slicesp_2}(E) = s_2  \right\}$. Extending the definition of slice-space weight to non-square matrices as the dimension of the $\F_q$-span of rows or columns, we can consider $\BBB_{s_1,s_2} := \left\{ E\in\F_{q^n}^{s_1 \times n} \ | \  \wt_{\slicesp_1}(E) = s_1, \wt_{\slicesp_2}(E) = s_2  \right\}$ and $\CCC_{s_1,s_2} :=  \left\{ E\in\F_{q^n}^{s_1 \times s_2} \ | \  \wt_{\slicesp_1}(E) = s_1, \wt_{\slicesp_2}(E) = s_2  \right\}$. We can then construct the surjective maps $\psi_{s_1}^{(1)} : \AAA_{s_2,s_2} \to \BBB_{s_1,s_2}$, $\psi_{s_2}^{(2)} : \BBB_{s_2,s_2} \to \CCC_{s_1,s_2}$ and $\psi_{s_1,s_2} : \AAA_{s_1,s_2} \to \CCC_{s_1,s_2}$ defined by $\psi_{s_1,s_2} = \psi_{s_2}^{(2)} \circ \psi_{s_1}^{(1)}$ and by the following relations, where the minima below are defined with a fixed total order, e.g. the lexicographical order, on the set of tuples of integers.
     $$\begin{array}{c}
          \forall E \in \AAA_{s_1,s_2}, \forall i\in\llbracket 1,s_1\rrbracket, \forall j \in \llbracket 1,n\rrbracket:\psi_{s_1}^{(1)}(E)[i,j] = E[I_i,j] \phantom{xxxxxx} \\
          \phantom{x}\text{where } I := \min\{I' \in \llbracket 1,n\rrbracket^{s_1} \ | \ I'_1 < \cdots < I'_{s_1}\ , \ (E[I'_1,:],\cdots,E[I'_{s_1},:]) \ \text{linearly independent over $\F_q$}\}.\\[0.2cm]
        \forall E \in \BBB_{s_1,s_2}, \forall i\in\llbracket 1,s_1\rrbracket, \forall j \in \llbracket 1,s_2\rrbracket:\psi_{s_2}^{(2)}(E)[i,j] = E[i,J_j] \phantom{xxxxxx} \\
          \phantom{x}\text{where } J := \min\{J' \in \llbracket 1,n\rrbracket^{s_2} \ | \ J'_1 < \cdots < J'_{s_1}\ , \ (E[:,J'_1],\cdots,E[:,J'_{s_2}]) \ \text{linearly independent over $\F_q$}\}.
     \end{array}$$
     These maps are well-defined. Indeed, with $E \in \AAA_{s_1,s_2}$ and with $I$ as defined above, $(E[I_1',:],\dots,E[I_{s_1}',:])$ is a basis of $\UUU_{1}(E)$ and thus, in $(\mathbb{F}_q^n)^{\otimes 3}$, $\psi_{s_1}^{(1)}$ corresponds to a succession of two operations preserving the slice-spaces dimension : a multiplication by a well-chosen invertible matrix along the first mode moving the considered rows up and deleting the rows below, and cropping the created zero rows, see \cite[Section~3]{ByrneNeriRavagnaniSheekeyTensorRepresentation2019}. Same goes for $\psi_{s_2}^{(2)}$ and the surjectivity of the two maps is immediate if one considers matrices of the form $\left[\begin{smallmatrix}
         B&0\\0&0
     \end{smallmatrix}\right]$ with $B \in \F_{q^n}^{s_1\times s_2}$ with rows and column linearly independent over $\F_q$. Finally, we can consider the map $\theta : \Tilde{\SSS} \to \F_{q^n}^{R_n\times R_n}$ given by $\theta(E) = (\phi_{\wt_{\slicesp_1}(E),\wt_{\slicesp_2}(E)} \circ \psi_{\wt_{\slicesp_1}(E),\wt_{\slicesp_2}(E)})(E)$ for each $E \in \Tilde{\SSS}$ and verify that the map is surjective, proving inequality (\ref{eq:MinNumberofTensorswithinssdim}).

    Let $\EEE_n :=  \left\{ E\in\F_{q^n}^{n \times n} \ | \ \trank(E) > R_n \ , \ \forall C \in\CCC(\llbracket 0,\mu_n\rrbracket^2) : \text{\texttt{Alg\ref{alg:RadicalDecodingFindMin}}}(C+E,\mu_n) = C \right\}$ be the set such that $\nu_n := |\EEE_n|$ is the integer involved in equation (\ref{eq:AssBhvrAdditionalErrors}).
    By Proposition~\ref{prop:weightsandtensorrank}, $\trank_{\F_q}(E) \leq R_n$ implies $\wt_{\Sigma\slicesp}(E) \leq 2R_n$ for each $E \in \F_{q^n}^{n\times n}$. In particular, we have the following relations.
    \begin{align*}
       \EEE_n & \supseteq  \{E \in \F_{q^n}^{n\times n}\  |\  \trank_{\F_q}(E) > R_n\ ,\  \wt_{\Sigma\slicesp}(E) \leq 2R_n\}\\
       & \supseteq  \{E \in \F_{q^n}^{n\times n}\  | \ \wt_{\Sigma\slicesp}(E) \leq 2R_n\} \backslash   \{E \in \F_{q^n}^{n\times n}\  |\  \trank_{\F_q}(E) \leq  R_n\}\\
       & \supseteq \{E \in \F_{q^n}^{n\times n}\  |\  \wt_{\slicesp_1}(E) \leq R_n\ ,\  \wt_{\slicesp_2}(E) \leq n \ ,\  \wt_{\fibresp_3}(E) \leq R_n \} \backslash   \{E \in \F_{q^n}^{n\times n}\  |\  \trank_{\F_q}(E) \leq  R_n\}
    \end{align*}
    Hence, we have the following inequality by Lemma~\ref{lemma:UBontensorsofrankatmost}.
    \begin{align*}
       \nu_n&  \geq |\{E \in \F_{q^n}^{n\times n}\  |\  \wt_{\slicesp_1}(E) \leq R_n\ ,\  \wt_{\slicesp_2}(E) \leq n \ ,\  \wt_{\fibresp_3}(E) \leq R_n \}|-    |\{E \in \F_{q^n}^{n\times n}\  |\  \trank_{\F_q}(E) \leq  R_n\}|\\
       &  \geq \underbrace{q^{nR_n^2} - \frac{(q^n-1)^{3R_n}}{(q-1)^{2R_n}} }_{:= \eta_n}.
    \end{align*}
    Note that $ \frac{(q^n-1)^{3R_n}}{(q-1)^{2R_n}}=o(q^{(3R_n+1)n})$. Since $R_n \geq \frac{3+\sqrt{13}}{2}$ for sufficiently large $n$, we have $nR_n^2 \geq (3R_n+1)n$ and hence $\eta_n \sim q^{nR_n^2}$. In particular, we obtain that $\frac{\nu_n}{q^{nR_n^2}} \geq \frac{\eta_n}{q^{nR_n^2}}$, thus there are positive constants $L$ and $n_0$ such that $\nu_n \geq Lq^{nR_n^2}$ for each $n \geq n_0$. In other words, we have $\nu_n = \Omega(q^{nR_n^{2}})$.
\end{proof}

\begin{Rq}
    It is non-trivial to check if the lower-bound on the tensor-rank of Corollary~\ref{corol:maxrankslicesvstrank} is met or not. It is however the case for $\mu = n-1$ as the evaluation of $a\tr(bX)\tr(cY)$ has tensor-rank one, with $a,b,c\in \mathbb F_q$ with an argument similar to \cite[Prop~14.44]{Burgisser1997ch14}, with \cite[Theorem 2.24]{lidl1994introduction}.
\end{Rq}

Let $\mu \in \llbracket 0,n-1\rrbracket$. We have seen in Section~\ref{subsec:interleaving} that an interleaved decoder can be applied to $\CCC(\llbracket 0,\mu\rrbracket ^2)$. We obtain an algorithm that can correct with high probability any error $E\in\F_{q^n}^{n\times n}$ such that $\wt_{\fibresp_3}(E) \leq \frac{2n}{2n+1}(n - \mu -1)$. With Proposition~\ref{prop:weightsandtensorrank}, it is possible to consider the algorithm as a tensor-rank decoder. In other words, if $E\in \F_{q^n}^{n\times n}$ , the decoder for the interleaved code can decode any received word $R = C + E$ with $C \in \CCC_\alpha(\llbracket 0,\mu\rrbracket^2)$ if $\trank_{\F_q}(E) \leq \frac{1}{2}(n - \mu -1)$, and if $\trank_{\F_q}(E) \leq \frac{2n}{2n+1}(n - \mu -1)$ with high probability. In this way, like the other decoders introduced above, interleaving decoding allows one to correct errors in the tensor-rank metric with additional correctable errors above the decoding radius inherited from Gabidulin codes. Since the distribution of $\wt_{\fibresp_3}(E)$ is unknown too for a chosen uniformly at random $E$ of given tensor-rank, estimating the probability that the interleaved decoder fails for the tensor-rank metric is an open problem.

\begin{Rq}
    As expressed before, there is a significant asymptotic difference between $|\{E' \in \F_{q^n}^{n\times n}, \wt_{\fibresp_3}(E') = t\}|$ and $|\{E' \in \F_{q^n}^{n\times n}, \trank(E')\leq t\}|$ as $n$ tends towards infinity.     Indeed, the first one is lower-bounded by a power of $q$ with an exponent quadratic in $n$, while the second is upper-bounded by a power of $q$ with an exponent linear $n$. Hence, the comparison of the probability of failure for the uniform distributions of randomly chosen elements respectively with fixed weight or upper-bounded tensor-rank is complicated too, especially without information about the number of high tensor-rank elements in the first set.
\end{Rq}

\section{Generalisation to higher order tensor codes.}
\label{sec:Generalisation}

Let $m \in \mathbb{N}_0$ with $m \geq 2$. {We will call $\boldsymbol{(m+1)}$\textbf{-tensors} of size $n \times \cdots \times n$ over the field $\mathbb{F}_q$ any element of the $n^{m+1}$-dimensional $\mathbb{F}_q$-vector space $\bigotimes_{j = 1}^{m+1} \mathbb{F}_q^n = (\mathbb{F}_q^{n})^{\otimes (m+1)}$. An element $T \in (\mathbb{F}_q^{n})^{\otimes (m+1)}$ will be uniquely associated to its expression as an $(m+1)$-dimensional array in the basis $(e_{i_1} \otimes\cdots \otimes e_{i_{m+1}})_{i \in \llbracket 1,n\rrbracket^{m+1}}$ where we denote by $(e_\iota)_{\iota \in \llbracket 1,n\rrbracket}$ the canonical $\mathbb{F}_q$-basis of $\mathbb{F}_q^n$. Namely, a unique tuple  $(T[i])_{i \in \llbracket 1,n\rrbracket^{m+1}}$ is associated to $T$ such that $ T = \sum_{ i \in \llbracket 1,n \rrbracket^{m+1}} T[i] e_{i_1} \otimes \cdots \otimes e_{i_{m+1}}$. 

The \textbf{tensor-rank} of an element of $\Gamma \in (\mathbb{F}_q^n)^{\otimes (m+1)}$ is  the least amount $\mu$ of \emph{elementary tensors}, \emph{i.e.} tensors of the form $u_{r,1} \otimes \cdots \otimes u_{r,m+1}$ with $u_{r,1}, \dots , u_{r,m+1} \in \mathbb{F}_q^n$ for each $r \in \llbracket 1,\mu\rrbracket$, that sum up to $\Gamma$. 

Note that there is an $\F_{q^n}$-isomorphism $\F_{q^n}\otimes (\F_{q}^n)^{\otimes m} \to \prod_{i \in [n]^m}\F_{q^n}$ and thus every element in $\F_{q^n}\otimes (\F_{q}^n)^{\otimes m}$ can be represented by an $m$-dimensional array with entries in $\F_{q^n}$. Like for the case $m =2$, any basis $\omega$ of $\mathbb{F}_{q^n}/\mathbb{F}_q$ allows us to identify every element of $(\F_q^{n})^{\otimes (m+1)}$ to a unique element of $\F_{q^n}\otimes (\F_{q}^n)^{\otimes m}$ with the $\mathbb{F}_q$-isomorphism $\mathfrak{s}_\omega : (\mathbb{F}_q^n)^{\otimes (m+1)} {\longrightarrow}\F_{q^n}\otimes (\F_{q}^n)^{\otimes m}$, $\Gamma \mapsto  \sum_{I_{m+1} = 1}^n \omega_{I_{m+1}} (\Gamma[i_1,\dots,i_m,I_{m+1}])_{i \in \llbracket 1,n\rrbracket^m}$ and same goes for subspaces.
The theory and the algorithms described above can be generalised to higher order tensors. 
One can define the \textbf{Roth-tensor code} of order $m+1$ of set $\mathcal{S} \subseteq \llbracket 0,n-1 \rrbracket^m$ to be the $\F_{q^n}$-subspace of $\F_{q^n}\otimes (\F_{q}^n)^{\otimes m}$ given by 
$$\mathcal{C}_{\alpha}(\mathcal{S}) := \left\{ \left. \big(f(\alpha_{i_1},\dots,\alpha_{i_{m}}) \big)_{i \in \llbracket 1,n\rrbracket^m} \right| f(X_1,\dots,X_m) \in \mathscr{M}_{n,\mathbb{F}_{q^n}}[X_1,\dots,X_m] , \Supp(f) \subseteq  \mathcal{S}\right\},$$
where $\mathscr{M}_{q,\F_{q^n}}[X_1,\dots,X_m]$ is the $\mathbb{F}_{q^n}$-span of monomials of the form $X_1^{q^{r_1}}\cdots X_m^{q^{r_m}}$ where $(r_1,\dots,r_m)$ ranges in $\llbracket 0,n-1\rrbracket^m$, and where the support $\Supp(P)$ of such a \emph{multilinearised $q$-polynomial} $P(X_1,\dots,X_m)$ is the set of points $(r_1,\dots,r_m)$ corresponding to the multi-degrees of its monomials. Similarly, the code $\mathcal{C}_\alpha(\mathcal{S})$ is a code of $\mathbb{F}_{q^n}$-dimension $|\mathcal{S}|$ and corresponds to an $(m+1)$-order tensor code, \emph{i.e.} an $\mathbb{F}_{q}$-vector subspace of $(\mathbb{F}_q^n)^{\otimes (m+1)}$ defined with parity check equations as shown in \cite{ROTHTensorCodesForRankMetric}, and correspond to a polynomial code denoted $\mathcal{C}(\mathcal{S})$, vector subspace of the quotient $\mathscr{M}_{q,\mathbb{F}_{q^n}}[X_1,\dots,X_m]$ by the subspace $U_{X_\bullet}$ spanned by polynomials of the form $f(X_1,\dots,X_{j-1},X_j^{q^n}-X_j,X_{j+1},\dots,X_n)$ where $f\in \mathscr{M}_{q,\mathbb{F}_{q^n}}[X_1,\dots,X_m]$ and $j \in \llbracket 1,m\rrbracket$. The metrics introduced in the sections above are also generalisable for $m \geq 3$ and we can consider the following elements for a given $T \in \F_{q^n}\otimes (\F_{q}^n)^{\otimes m}$.

\begin{itemize}
    \item The \textbf{tensor-rank} of $T$ is the least amount of tensors of the form $\gamma \otimes \bigotimes_{j = 1}^m u_j$, where $u_j \in \mathbb{F}_{q}^n$ and $\gamma \in \mathbb{F}_{q^n}$, that add up to $T$. It corresponds to the tensor-rank of its associated tensor over $\mathbb{F}_q$, and we denote by $\trank_{\mathbb{F}_q}(T)$ this value.
    \item  The $(m+1)$-\textbf{fibre space} of $T$ is the $\F_q$-span of its entries and $\wt_{\fibresp_{m+1}}(T)$ is its $\mathbb{F}_q$-dimension. The entries correspond to the $(m+1)$-\emph{fibres} of its associated tensor in $(\mathbb{F}_q^n)^{\otimes(m+1)}$, \emph{i.e.} the vectors $\Gamma[i_1,\dots,i_m,:] \in \mathbb{F}_{q}^n$ where $i \in \llbracket 1,n\rrbracket^m$. The $(m+1)$-\textbf{fibre space} to be the $\mathbb{F}_q$-span of the former and $\wt_{\fibresp_{m+1}}(T)$ its $\mathbb{F}_q$-dimension. 
    \item We will discuss in the next part about the fibres of $T \in \F_{q^n} \otimes (\F_q^n)^{\otimes m}$ seen as an $m$-dimensional array over $\F_{q^n}$, that is the vectors of length $n$ and over $\F_{q^n}$ obtained by fixing all coordinates of $T$ but one (that corresponds to the rows and columns of the matrices over $\F_{q^n}$ studied above). They do not correspond to fibres of the associated tensors but correspond to certain of its submatrices instead.
    \item The $j$-\textbf{slices} of $T$ are the subtensors of $T$ of the form $T[:,\dots,:,i_j,:,\dots,:] \in  \F_{q^n}\otimes (\F_{q}^n)^{\otimes (m-1)}$ with $i_j \in \llbracket 1,n\rrbracket$, \emph{i.e.} fixing only the $j^{\text{th}}$ coordinate. The $j$-\textbf{slice space} to be the $\mathbb{F}_q$-span of all $j$-slices and $\wt_{\slicesp_{j}}(T)$ its $\mathbb{F}_q$-dimension. 
\end{itemize}

\begin{Prop}
    Let $\SSS \subseteq \llbracket 0,n-1\rrbracket^m$. Let $j \in \llbracket 1,m \rrbracket$. Denote $\SSS^{\,\hat{\!j}} = \{(s_1,\dots,s_{j-1},s_{j+1},\dots,s_m) \ | \ s\in \SSS\}$. Then any $j$-slice of an element $C \in \CCC_\alpha(\SSS)$ is in $\CCC_\alpha(\SSS^{\,\hat{\!j}})$. Consequently, we have the following:
    $$\min_{\substack{{C} \in \CCC_\alpha(\SSS^{\,\hat{\!j}})\\C\neq 0}}\trank_{\F_q}(C) \leq \min_{\substack{{C} \in \CCC_\alpha(\SSS)\\C \neq 0}}\trank_{\F_q}(C). $$
\end{Prop}
\begin{proof}
    Consider the elements above. Let $C \in \CCC_\alpha(\SSS)$ and consider the following minimal tensor-rank form of $C$ where $\tau = \trank_{\F_q}(C)$, and where $\gamma^{(r)} \in \F_{q^n}$ and $u_j^{(r)} \in \F_{q}^n$ for each $r \in \llbracket 1,\tau\rrbracket$ and each $j \in \llbracket 1,m\rrbracket$,
    $$C = \sum_{r = 1}^\tau \gamma^{(r)} \cdot u_{1}^{(r)}\otimes \cdots \otimes u_m^{(r)}.$$
    By construction, there exists a polynomial $f(X_1,\dots,X_m) \in \mathscr{M}_{n,\mathbb{F}_{q^n}}[X_1,\dots,X_m]$ with $\Supp(f) \subseteq  \mathcal{S}$ such that $C[i] = f(\alpha_{i_1},\dots,\alpha_{i_m})$ for each $i\in \llbracket 1,n\rrbracket^m$.\\
    Let $I_j \in \llbracket 1,n\rrbracket$. The $I_j^{\ th}$ $j$-slice of $C$ is the tensor $\Tilde C \in (\F_{q^n}^n)^{\otimes (m-1)}$ given for all $(i_1,\dots,i_{j-1},i_{j+1},\dots,i_m) \in \llbracket 1,n\rrbracket^{m-1}$ by $\Tilde{C}[i_1,\dots,i_{j-1},i_{j+1},\dots,i_m] = C[i_1,\dots,i_{j-1},I_j,i_{j+1},\dots,i_m]$. Clearly, $\Tilde{C}$ is the evaluation tensor of the multilinearised $q$-polynomial with $(m-1)$ variables $f(X_1,\dots,X_{j-1},\alpha_{I_j},X_{j+1},\dots,X_m)$ which has support contained in $\SSS^{\,\hat{\!j}}$ by construction. Hence $\Tilde{C} \in \CCC_{\alpha}(\SSS^{\,\hat{\!j}})$, and we have $\trank_{\F_q}(\Tilde C) \leq \trank_{\F_q}(C)$, which can be deduced from  the following expression,
    $$\Tilde{C} =  \sum_{r = 1}^\tau (u_{j,I_{j}}^{(r)} \gamma^{(r)})\cdot u_{1}^{(r)}\otimes \cdots \otimes u_{j-1}^{(r)} \otimes u_{j+1}^{(r)}  \otimes \cdots \otimes u_{m}^{(r)}.$$
    If $C \neq0$, then there exists an index $I_j$ such that the  $I_j^{\ th}$ $j$-slice of $C$ is non-zero, and in that case $\trank_{\F_q}(C)$ is bounded from below by the minimum tensor-rank of all non-zero elements if $\CCC_{\alpha}(\SSS^{\,\hat{\!j}})$. Since this is true for all non-zero element $C \in \CCC_\alpha(\SSS)$, we have the wanted inequality.
\end{proof}

\begin{Rq}
 For $m = 1$, the tensor-rank of an element in $\F_{q^n}\otimes (\F_{q}^n)^{\otimes m} = \F_{q^n}^n$ is exactly its $\F_q$-rank. Hence, the proposition above can be seen as a generalisation to higher order tensors of Proposition~\ref{prop:GabidulinColumns} and Proposition~\ref{prop:weightsandtensorrank}.
\end{Rq}

\subsection{The fibre-wise decoders}\label{subsec:ordermfibre}
As for for $m=2$, where any row or column of a codeword was a Gabidulin codeword, any fibre (as an $m$-tensor) of any element $T \in \mathcal{C}_\alpha(\mathcal{S}) \subseteq \F_{q^n}\otimes (\F_{q}^n)^{\otimes m}$ is a Gabidulin codeword, \emph{i.e.} for each $j \in \llbracket 1,m\rrbracket$ and each $(i_1,\dots,i_{j-1},i_{j+1},\dots,i_m) \in \llbracket 1,n\rrbracket^{m-1}$ we have $T[i_1,\dots,i_{j-1},:,i_{j+1},\dots,i_m]\in \GG_{\max \pi_j(\mathcal{S}) +1}(\alpha)$. And similarly, when $\SSS = \prod_{j=1}^m \llbracket 0,\mu_j\rrbracket$ for certain $\mu_1,\dots,\mu_m \in \llbracket 0,n-1\rrbracket$, then $\CCC_\alpha(\SSS)$ is isomorphic to the tensor product (over $\F_{q^n}$) of Gabidulin codes $\bigotimes_{j=1}^m \GG_{\mu_j+1}(\alpha)$.  

Consequently, one can correct any received word by fixing an integer $j \in \llbracket 1,m\rrbracket$ and correcting every $j$-fibre of the received word with a Gabidulin decoder, hence generalising Algorithm~\ref{alg:fibrewise} which, with $\mathcal{S} = \prod_{j = 1}^m \llbracket 0,\mu_j\rrbracket$, will be able to correct any error $E \in (\mathbb{F}_{q^n}^n)^{\otimes m}$ such that
\begin{equation}
    \label{eq:ordermfibrewisecriterion}\max_{(i_1,\dots,i_{j-1},i_{j+1},\dots,i_m) \in \llbracket 1,n\rrbracket^{m-1}} \rank_{\mathbb{F}_q} E[i_1,\dots,i_{j-1},:,i_{j+1},\dots,i_m] \leq \left\lfloor \frac{n-\mu_j-1}2 \right\rfloor.
\end{equation}

\refstepcounter{algorithm}\label{alg:fibrewisegeneralonedirr}
Therefore, for each integer $j \in \llbracket 1,m\rrbracket$, we obtain an algorithm that we will denote \textbf{Algorithm \thealgorithm${^{[j]}}$} that, with input $n$, $q$, the basis $\alpha$, the set $\SSS = \prod_{j = 1}^m \llbracket 0,\mu_j\rrbracket$, and a tensor $R \in \F_{q^n} \otimes (\F_q^n)^{\otimes m}$, returns the tensor $C \in \F_{q^n} \otimes (\F_q^n)^{\otimes m}$ such that for each $(i_1,\dots,i_{j-1},i_{j+1},\dots,i_m) \in \llbracket 1,n\rrbracket^{m-1}$ we have
\[
C[i_1,\dots,i_{j-1},:,i_{j+1},\dots,i_m] = \texttt{GabDec}(R[i_1,\dots,i_{j-1},:,i_{j+1},\dots,i_m], \mu_j +1,\alpha).
\]
We will again assume that the Gabidulin decoder returns an unspecified word when the rank of the error in the received message exceeds the decoding radius.

\begin{Prop}
    Denote by $S = \prod_{j = 1}^m \llbracket 0,\mu_j\rrbracket$. For each $j \in \llbracket 1,m\rrbracket$, and each $R \in  \F_{q^n}\otimes (\F_q^n)^{\otimes n}$, denote by \texttt{Alg\ref{alg:fibrewisegeneralonedirr}}${}^{[j]}(R)$ the output of Algorithm \ref{alg:fibrewisegeneralonedirr}${}^{[j]}$ on $R$ with input $n$, $q$, $\alpha$ and $\SSS$. For each sequence $(j_1,\dots,j_t) \in \llbracket 1,m\rrbracket^t$ with $t \in \N$, denote by $\KKK_{(j_1,\dots,j_t)} \subseteq \F_{q^n}\otimes (\F_q^n)^{\otimes n}$ the set of correctable errors after a succession of the corresponding algorithm, \emph{i.e.} 
    $$\KKK_{(j_1,\dots,j_t)} = \left\{E \in   \F_{q^n}\otimes (\F_q^n)^{\otimes n} \ | \ \forall C \in \CCC_\alpha(\SSS) : \text{\texttt{Alg\ref{alg:fibrewisegeneralonedirr}}${}^{[j_t]}\circ\cdots \circ$\texttt{Alg\ref{alg:fibrewisegeneralonedirr}}${}^{[j_1]}(C+E) = C$}\right\}.$$
    Then we have the following properties.
    \begin{enumerate}
        \item For each $j^* \in \llbracket 1,m \rrbracket$, we have $\KKK_{(j_1,\dots,j_t)} \subseteq \KKK_{(j_1,\dots,j_t,j^*)}$.
                \item If $t = 2$ and if $j_1,j_2$ are distinct, then $\KKK_{(j_1,j_2)}$ contains any $E \in \F_{q^n}\otimes (\F_q^n)^{\otimes n}$ such that 
        \begin{equation}
        \label{eq:severalwayscountingcriterion}
        \min_{\substack{\III_{j_2} \subseteq \llbracket 0,n-1\rrbracket \\|\III_{j_2}| = \left\lceil \frac{n+\mu_{j_2}+1}{2}\right\rceil}} \max_{i_{j_2} \in \III_{j_2}} \rank_{\F_q} E[i_1,\dots,i_{j_1},:,i_{j_1+1},\dots,i_m] \leq \left\lfloor \frac{n-\mu_{j_1}-1}{2}\right\rfloor 
        \end{equation}
        for each $(i_j)_{j \neq j_1,j_2} \in \llbracket 0,n-1\rrbracket^{m-2}$.
    \end{enumerate}
\end{Prop}
\begin{proof}
    The first point is clear as Gabidulin decoders (with well chosen parameters) systematically return the input if the input is itself a Gabidulin code (with well chosen parameter).
\end{proof}

\refstepcounter{algorithm}\label{alg:fibrewisegeneralmultidirr}
We will denote by \textbf{Algorithm \thealgorithm} the algorithm that, with input $n$, $q$, the basis $\alpha$, the set $\SSS = \prod_{j = 1}^m \llbracket 0,\mu_j\rrbracket$, and a tensor $R \in \F_{q^n} \otimes (\F_q^n)^{\otimes m}$, returns the tensor $C \in \F_{q^n} \otimes (\F_q^n)^{\otimes m}$ given by 
\[
    C := \text{\texttt{Alg\ref{alg:fibrewisegeneralonedirr}}}{}^{[m]}\circ\cdots \circ \text{\texttt{Alg\ref{alg:fibrewisegeneralonedirr}}}{}^{[1]}(R),
\]
where we denote by \texttt{Alg\ref{alg:fibrewisegeneralonedirr}}${}^{[j]}(R)$ the output of Algorithm \ref{alg:fibrewisegeneralonedirr}${}^{[j]}$ on $R$ with input $n$, $q$, $\alpha$ and $\SSS$.

\subsection{The radical decoder}\label{subsec:ordermrad}
We define the $t^{th}$ radical of an $m$-linear map $f : \mathbb{F}_{q^n} \times \cdots \times \mathbb{F}_{q^n} \to \mathbb{F}_{q^n}$ to be the subspace of $\F_{q^n}$ given by $\Rr\Aaa\Ddd_t(f) :=\{x \in \mathbb{F}_{q^n} \ |  \ \forall (x_j)_{j \in \llbracket 1,m\rrbracket \backslash \{t\}} \in (\mathbb{F}_{q^n})^{ \llbracket 1,m\rrbracket \backslash \{t\}} : f(x_1,\dots,x_m) = 0 \}$. We introduce the generalisation of Problem~\ref{pb:Linearisedproblem} for higher order tensors.

\begin{Pb} \label{pb:LinearisedproblemHOT} Let $\SSS\subseteq \llbracket 0,n-1\rrbracket^m$, $t \in {\mathbb{N}_0}$ and $R \in (\mathbb{F}_{q^n}^{n})^{\otimes m}$. Solve for $(V(Z),N(X_1,\dots,X_m)) \in \mathscr{M}_{q,\mathbb{F}_{q^n}}[Z]\times  \mathscr{M}_{q,\mathbb{F}_{q^n}}[X_1,\dots,X_m]$ satisfying $\qdeg V(Z) \leq t$ and $\Supp(N(X_1,\dots,X_m)) \subseteq \mathcal{S} + \llbracket 0,t\rrbracket$ in 
\begin{equation}
    \label{eq:linearisedProblemHOT}
     \forall i \in \llbracket 1,n\rrbracket^m :  V(R[i]) = N(\alpha_{i_1},\dots,\alpha_{i_m}).
\end{equation}
\end{Pb}

Let $\mu \in \llbracket 0,n-2\rrbracket$ and consider $\mathcal{S} = \llbracket 0,\mu \rrbracket^m$. Let $t \in \mathbb{N}_0$ and let $E \in (\mathbb{F}_{q^n}^{n})^{\otimes m}$ with $\wt_{\fibresp_{m+1}}(E)\leq t$. Then if $R = C + E$ for a given $C \in \mathcal{C}_\alpha(\mathcal{S})$ with associated polynomial $f(X_1,\dots,X_m)$, there exists a polynomial $V(Z) \in \MM_{q,\F_{q^n}}[X_1,\dots,X_m]$ such that $(V(Z),N(X_1,\dots,X_m))$ is a solution of Problem~\ref{pb:Linearisedproblem} with parameters $\SSS$, $t$ and $R$.

Then Theorem~\ref{thm:RadicalCriterion} can be generalised in the following way. The proof of this generalisation is identical to the proof of the initial theorem.

\begin{Thm}\label{thm:RadicalCriterionHOT}
    Let $f(X_1,\dots,X_m) \in \mathcal{C}(\mathcal{S})$ and let $C$ be its corresponding codeword in $\mathcal{C}_{\alpha}(\mathcal{S})$. Consider an error $E \in (\mathbb{F}_{q^n}^{n})^{\otimes m}$ and let $R = E + C$. Let $\Theta \in \llbracket 0,n-\mu-2\rrbracket$ and assume that $\min_{j \in \llbracket 1,m\rrbracket}\wt_{\slicesp_j}(E) \leq \Theta$.
    Then any solution $(V(Z),N(X_1,\dots,X_m))$ of Problem~\ref{pb:LinearisedproblemHOT}  with parameters $\SSS$, $t = n-\mu-1-\Theta$ and $R$ is such that $V(f(X_1,\dots,X_m)) = N(X_1,\dots,X_m)$.     \end{Thm}

\begin{Corol}
    Assume the same notation of Theorem~\ref{thm:RadicalCriterionHOT}. Suppose that
    \begin{equation} \label{eq:RadicalCriterionHOT}
        \wt_{\Sigma\!\slicesp}(E) := \wt_{\fibresp_{3}}(E) +  \min_{j \in \llbracket 1,m\rrbracket} \wt_{\slicesp_j}(E) \leq n - \mu -1.
    \end{equation}
    Then there exists $t \in \llbracket 0,n-1\rrbracket$ such that every solution $(V(Z),N(X_1,\dots,X_m))$ of Problem~\ref{pb:LinearisedproblemHOT} with given input $\SSS$, $t$, and any $R \in \CCC_\alpha(\SSS) + E$ is of the form $N(X_1,\dots,X_m) = V( f(X_1,\dots,X_m))$.\end{Corol}
Therefore, solving the linear system (\ref{eq:linearisedProblemHOT}) like in Algorithms~\ref{alg:RadicalDecoding} and~\ref{alg:RadicalDecodingFindMin} allows the correction of any error such that $\wt_{\fibresp_{m+1}}(E) +  \min_{j \in \llbracket 1,m\rrbracket}\wt_{\slicesp_j}(E) \leq n - \mu -1$. \refstepcounter{algorithm}\label{alg:generalradicaldecoder}
We will then denote by \textbf{Algorithm \thealgorithm} the algorithm that, with input $n$, $q$, the basis $\alpha$, the set $\SSS = \prod_{j = 1}^m \llbracket 0,\mu_j\rrbracket$, and a tensor $R \in \F_{q^n} \otimes (\F_q^n)^{\otimes m}$ that follows the following steps. First, it computes the smallest integer $\delta \in \llbracket 0,n-1\rrbracket$ such that there exists a non-zero solution of Problem~\ref{pb:LinearisedproblemHOT} such that $\qdeg V(Z) \leq \delta$ and $\Supp N(X,Y) \subseteq \SSS + \llbracket 0,\delta\rrbracket$. It returns an error message if the problem has no non-zero solutions. It then picks such a non-zero solution $(V(Z),N(X_1,\dots,X_m))$ of the problem above and computes a left-Euclidean like factorisation as in \textbf{Algorithm \ref{alg:divisionalgorithm}}, which returns $f(X_1,\dots,X_m)$ with $\Supp(f) \subseteq \SSS$ if $V(f(X_1,\dots,X_m)) = N(X_1,\dots,X_m)$. If that is not possible, it returns an error message. Finally, it computes the codeword $C \in \CCC_\alpha(\SSS)$ associated to $f$, and returns it.

\subsection{Comparison}
Following the same principle as in Proposition~\ref{prop:comparisonfibreradical}, we can show that for a given $\mu \in \llbracket 0,n-1\rrbracket$, any element $E \in (\F_{q^n}^n)^{\otimes m}$ satisfying (\ref{eq:RadicalCriterionHOT}) satisfies Equation (\ref{eq:ordermfibrewisecriterion}) for at least one choice of $j$, but not necessarily strictly more than one. Moreover, there exist tensors in $(\F_{q^n}^n)^{\otimes m}$ that satisfy (\ref{eq:ordermfibrewisecriterion}) for at least one choice of $j$ but do not satisfy  (\ref{eq:linearisedProblemHOT}). One can proceed with a generalisation of the elements in the proof of Proposition~\ref{prop:comparisonfibreradical}, but we can also construct such an $E$ satisfying (\ref{eq:linearisedProblemHOT}) satisfy equation (\ref{eq:ordermfibrewisecriterion}) for exactly one choice of $j$. We give such elements in the example below.
\begin{Ex}
   Let $a = n -\mu -2$, the generalisation of the matrices $E_1$ and $E_2$ mentioned in Proposition~\ref{prop:comparisonfibreradical}, respectively into the great-diagonal tensor $\Tilde{E_1} = (\ind{ \{i_1 = \cdots =  i_m\}})_{i \in \llbracket 1,n\rrbracket^m}$ and $\Tilde{E_2} = (\ind{\{i_1 \leq a\}}\alpha_{i_1})_{i \in \llbracket 1,n\rrbracket^m}$, are such counter-examples. Now consider $\Tilde{E_3}  \in (\F_{q^n}^n)^{\otimes m}$ given for each $i \in \llbracket 1,n\rrbracket^m$ by
    $$\Tilde{E_3}[i] = \left\{\begin{array}{rcl}  
\alpha_{i_1} &\text{if}& i_1 \leq a, i_2 = 1,\dots,i_{m-1} = 1,\\
\alpha_{i_2} &\text{if}& i_1 =1 , 2 \leq i_2 \leq a, i_3 = 1 ,\dots,i_{m-1} = 1 ,\\
&\vdots & \\
\alpha_{i_{m-2}} &\text{if}& i_1 =1,\dots, i_{m-3} =1,  2 \leq i_{m-2}\leq a,i_{m-1} =1,\\
\alpha_{i_{m-1}} &\text{if}& i_1 =1, i_2 =1,,\dots, i_{m-2}=1, 2 \leq i_{m-1} \leq a ,\\
0&&else.
    \end{array}\right. $$
With this definition, we can easily see that $\wt_{\fibresp_{m+1}}(\Tilde{E_3})= a$, that $\wt_{\slicesp_m}(\Tilde{E_3}) = 1$ for definition of $\Tilde{E_3}[i]$ does not depend on $i_m$, and that the maximum $\mathbb{F}_q$-rank of the $j$-fibres for $j<m$ is equal to $a$. Indeed, every line of the definition exhibits a fibre achieving this maximum.  It is also possible to check that $\wt_{\slicesp_j}(\Tilde{E_3}) = a$ for $j<m$ and that the maximum $\mathbb{F}_q$-rank of the $m$-fibres is $1$. For $m = 3$, we have the shape illustrated in Figure~\ref{fig:ParticularErrors}.
\end{Ex}

Like in Subsection~\ref{subsec:distancesinthecode}, we can show that for any $T \in\F_{q^n}^{n\times n}$, the weights $\wt_{\slicesp_j}(T), j \in \llbracket 1,m\rrbracket$ and $\wt_{\fibresp_{m+1}}(T)$ are bounded from above by $\trank_{\F_q}(T)$, and that $\wt_{\Sigma\!\slicesp}(T)\leq 2\trank_{\F_{q}}(T)$. In particular, we are in the same configuration as for the case $m = 2$ and we obtain the following corollary.

\begin{Corol}
    The $\mathbb F_{q}$-tensor code $\mathcal C_\alpha(\llbracket 0,\mu\rrbracket^m)$ of dimension $n(\mu +1)^m$ over $\F_q$ endowed with the tensor-rank metric is $\left\lfloor\frac{n-\mu-1}{2}\right\rfloor$-error correcting, and  Algorithms~\ref{alg:fibrewisegeneralonedirr}${}^{[j]}$,~\ref{alg:fibrewisegeneralmultidirr} and~\ref{alg:generalradicaldecoder} can correct errors of tensor-rank at most $\left\lfloor\frac{n-\mu-1}{2}\right\rfloor$.
\end{Corol}

With the same arguments used in section~\ref{sec:Comparison}, we can give the complexity of the generalised algorithms introduced in the subsections above.
\begin{itemize}
    \item For each $j \in \llbracket 1,m\rrbracket$, the complexity of Algorithm \ref{alg:fibrewisegeneralonedirr}${}^{[j]}$ is $n^{m-1}$ times the complexity of the Gabidulin decoder, hence $\mathcal{O}(n^{m+2}\log n)$ computations on $\mathbb{F}_q$. 
    \item The complexity of Algorithm \ref{alg:fibrewisegeneralmultidirr} is $m n^{m-1}$ times the complexity of the Gabidulin decoder, hence $\mathcal{O}(m n^{m+2}\log n)$ computations on $\mathbb{F}_q$. 
    \item Let $\SSS = \llbracket 0,\mu\rrbracket^m$ and let $t \in \llbracket 0,n-1\rrbracket$, the set $\mathcal{S} + \llbracket 0,t\rrbracket$ is the disjoint union of $\mathcal{S}$ and of $t$ points of the form $(x_1,\dots,x_{j-1},\mu,x_{j+1},\dots,x_m)$ for $j \in \llbracket 1,m\rrbracket$ and $x_1,\dots,x_{j-1},x_{j+1},\dots,x_m \in \llbracket 0,\mu \rrbracket$; in particular we have that $|\mathcal{S} + \llbracket 0,t\rrbracket| = (\mu +1)^m + t\ o(\mu^m)$. And since equation (\ref{eq:linearisedProblemHOT}) corresponds to a linear system over $\mathbb{F}_{q^n}$ with $n^m$ equations and $t + 1  + |\mathcal{S} + \llbracket 0,t\rrbracket|$ unknowns. Therefore, the asymptotic number of Algorithm \ref{alg:generalradicaldecoder} is $\mathcal{O}(n^{3(m+1)})$.
\end{itemize}

\section*{Acknowledgements}
This work has emanated from research conducted with the financial support of the European Union MSCA Doctoral Networks, HORIZON-MSCA-2021-DN-01, Project 101072316, the French \emph{Agence Nationale de la Recherche} project ANR-21-CE39-0009-BARRACUDA and by \emph{Plan France 2030} ANR-22-PETQ-0008. We thank the anonymous reviewer of \cite{BCF25} for the interesting Remark~\ref{rq:ContextualisationTabular}.

\addcontentsline{toc}{section}{Bibliography}
\bibliographystyle{IEEEtran}
\bibliography{DecodingTensorCodes_bibio}

\appendix\newpage
\section{Figures}

\begin{sidewaystable}
    \centering
    \small 
    \begin{tabular}{|c|c|c|c|c|c|}
        \hline
        \makecell{\textbf{Algorithm}} & 
        \makecell{\textbf{On code}} & 
        \makecell{\textbf{Redundancy as}\\\textbf{ $\F_{q^n}$-vector space}} & 
        \makecell{\textbf{Best known condition for}\\\textbf{worst-case decoding}} & 
        \makecell{\textbf{Lower-bound on}\\\textbf{tensor-rank minimum}\\\textbf{distance}} & 
        \makecell{\textbf{Algorithm}\\\textbf{complexity}} \\
        \hline\hline
        \makecell{Alg~\ref{alg:fibrewise}\\{\tiny Column-wise dec.}} & $\CCC_\alpha(\llbracket 0,\mu_1\rrbracket \times \llbracket 0,\mu_2\rrbracket)$ & $n^2 - (\mu_1+1)(\mu_2+1)$&\makecell{$\forall i_2 \in \llbracket 1,n\rrbracket : $\\$\rank_{\F_q}E[:,i_2] \leq \frac{n-\mu_1-1}{2}$} &$n-\mu_1$&$O(n^4\log n)$\\\hline
        \makecell{Alg~\ref{alg:fibrewise}'\\{\tiny Row-wise dec.}} & $\CCC_\alpha(\llbracket 0,\mu_1\rrbracket \times \llbracket 0,\mu_2\rrbracket)$ & $n^2 - (\mu_1+1)(\mu_2+1)$&\makecell{$\forall i_1 \in \llbracket 1,n\rrbracket : $\\$\rank_{\F_q}E[i_1,:] \leq \frac{n-\mu_2-1}{2}$} 
        &$n-\mu_1$
        &$O(n^4\log n)$\\\hline
        \makecell{Alg~\ref{alg:fibreWiseTwoDirrs}\\{\tiny Fibre-wise dec.}} 
        & $\CCC_\alpha(\llbracket 0,\mu_1\rrbracket \times \llbracket 0,\mu_2\rrbracket)$ 
        & $n^2 - (\mu_1+1)(\mu_2+1)$
        &\makecell{$\exists J \subseteq \llbracket 1,n\rrbracket, |I| = \left\lceil \frac{n + \mu_2 + 1}{2} \right\rceil$\\$\forall i_2 \in J : $\\$\rank_{\F_q}E[:,i_2] \leq \frac{n-\mu_1-1}{2}$} &$n-\mu_1$
        &$O(n^4\log n)$\\\hline
        \makecell{Alg~\ref{alg:RadicalDecoding}\\{\tiny Fixed wt radical dec.}}&$\CCC_\alpha(\llbracket 0,\mu\rrbracket^2)$& $n^2 - (\mu+1)^2$&\makecell{With specified $t$,\\$\min_{j}\wt_{\slicesp_j}(E)\leq n-\mu-1-t$\\and $\wt_{\fibresp_3}(E)\leq t$}&$2\min\left\{{\begin{smallmatrix}
             t,\\n-\mu-t-1
        \end{smallmatrix}} \right\} +1$ & $O(n^9)$\\\hline
        \makecell{Alg~\ref{alg:RadicalDecodingFindMin}\\{\tiny Radical dec.}}&$\CCC_\alpha(\llbracket 0,\mu\rrbracket^2)$& $n^2 - (\mu+1)^2$&\makecell{$\wt_{\fibresp_3}(E) + \min_{j}\wt_{\slicesp_j}(E)$\\$ \leq n-\mu-1$}&$ n-\mu$ & $O(n^9)$\\\hline\hline
        \makecell{Roth decoder\\{\tiny Trank 1 correction}} & $\CCC(n,3,3;q)$ & $\leq 3$ & $\trank(E) \leq 1$ &3& $O(n^3)$ \\\hline
        \makecell{Roth decoder\\{\tiny Trank $ 2$ correction}} 
        &  $\CCC(n,5,3,q)$ 
        & $\leq 10$ 
        & $\trank(E) \leq 2$ 
        & 5
        &$O(qn^3)$ \\\hline
    \end{tabular}
    \caption{Summary of algorithms and code properties}
    \label{fig:sumupalgorithms}
\end{sidewaystable}

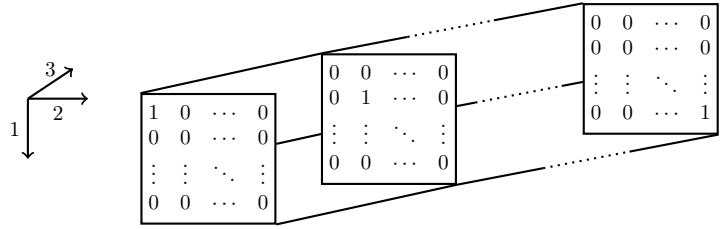
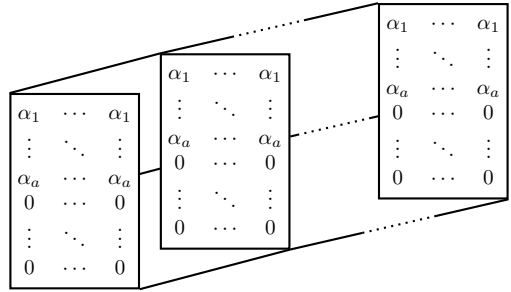
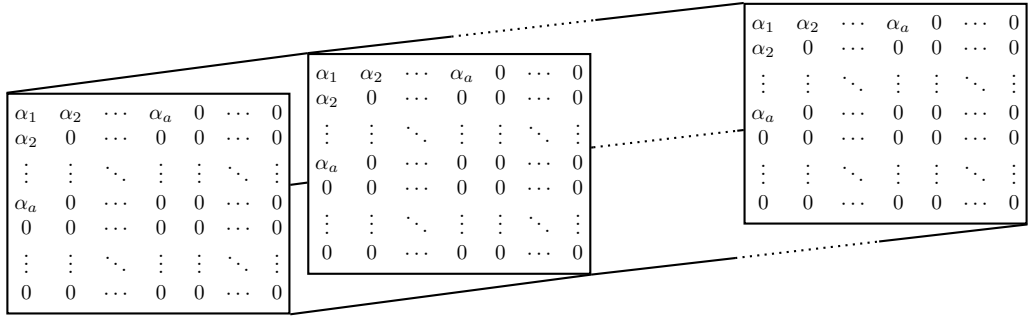
\begin{figure}[H]
    \centering
    \subfloat[Illustration of $\Tilde{E_1}$]{
   \begin{tikzpicture}[thick,scale=0.8, every node/.style={transform shape}]
    
    \draw[thick,->] (-3,1) -- (-3,0) node [midway, left] {1};
    \draw[thick,->] (-3,1) -- (-2,1) node [midway, below] {2};
    \draw[thick,->] (-3,1) -- (-2.25,1.5) node [midway, above] {3};

        \node[draw,rectangle] (M) at (0,0) {$\begin{matrix}
           1 & 0 & \cdots & 0 \\
            0 & 0 & \cdots & 0 \\
            \vdots &  \vdots &\ddots & \vdots \\ 
            0& 0&\cdots & 0
        \end{matrix}$};
        \node[draw,rectangle] (N) at (3,0.66) {$\begin{matrix}
          0 & 0 & \cdots & 0 \\
            0 & 1 & \cdots & 0 \\
            \vdots &  \vdots&\ddots & \vdots \\ 
            0& 0&\cdots & 0
        \end{matrix}$};
        \node[draw,rectangle] (P) at (7.35,1.495) {$\begin{matrix}
          0 & 0 & \cdots & 0 \\
            0 & 0 & \cdots & 0 \\
            \vdots &  \vdots&\ddots & \vdots \\ 
            0& 0&\cdots & 1
        \end{matrix}$};
        \draw (M)--(N);
        \draw (N) -- ($(N)!0.33!(P)$);
         \draw[dotted] ($(N)!0.33!(P)$) -- ($(N)!0.66!(P)$);
        \draw ($(N)!0.66!(P)$) -- (P);
        \draw (M.north west) -- (N.north west);
        \draw (M.south east) -- (N.south east);
        \draw (N.north west) -- ($(N.north west)!0.33!(P.north west)$);
         \draw[dotted] ($(N.north west)!0.33!(P.north west)$) -- ($(N.north west)!0.66!(P.north west)$);
        \draw ($(N.north west)!0.66!(P.north west)$) -- (P.north west);
	\draw (N.south east) -- ($(N.south east)!0.33!(P.south east)$);
         \draw[dotted] ($(N.south east)!0.33!(P.south east)$) -- ($(N.south east)!0.66!(P.south east)$);
        \draw ($(N.south east)!0.66!(P.south east)$) -- (P.south east);
    \end{tikzpicture}
    }\\
    \subfloat[Illustration of $\Tilde{E_2}$]{
\begin{tikzpicture}[thick,scale=0.8, every node/.style={transform shape}]
        \node[draw,rectangle] (M) at (0,0) {$\begin{matrix}
            \alpha_1 & \cdots & \alpha_1 \\
            \vdots & \ddots & \vdots \\ 
            \alpha_a & \cdots &\alpha_a \\
             0& \cdots & 0\\
            \vdots & \ddots & \vdots \\ 
            0& \cdots & 0
        \end{matrix}$};
        \node[draw,rectangle] (N) at (2.5,0.66) {$\begin{matrix}
          \alpha_1 & \cdots & \alpha_1 \\
            \vdots & \ddots & \vdots \\ 
            \alpha_a & \cdots &\alpha_a \\
             0& \cdots & 0\\
            \vdots & \ddots & \vdots \\ 
            0& \cdots & 0
        \end{matrix}$};
        \node[draw,rectangle] (P) at (6.125,1.495) {$\begin{matrix}
             \alpha_1 & \cdots & \alpha_1 \\
            \vdots & \ddots & \vdots \\ 
            \alpha_a & \cdots &\alpha_a \\
             0& \cdots & 0\\
            \vdots & \ddots & \vdots \\ 
            0& \cdots & 0
        \end{matrix}$};
        \draw (M)--(N);
        \draw (N) -- ($(N)!0.33!(P)$);
         \draw[dotted] ($(N)!0.33!(P)$) -- ($(N)!0.66!(P)$);
        \draw ($(N)!0.66!(P)$) -- (P);
        \draw (M.north west) -- (N.north west);
        \draw (M.south east) -- (N.south east);
        \draw (N.north west) -- ($(N.north west)!0.33!(P.north west)$);
         \draw[dotted] ($(N.north west)!0.33!(P.north west)$) -- ($(N.north west)!0.66!(P.north west)$);
        \draw ($(N.north west)!0.66!(P.north west)$) -- (P.north west);
	\draw (N.south east) -- ($(N.south east)!0.33!(P.south east)$);
         \draw[dotted] ($(N.south east)!0.33!(P.south east)$) -- ($(N.south east)!0.66!(P.south east)$);
        \draw ($(N.south east)!0.66!(P.south east)$) -- (P.south east);
    \end{tikzpicture}
    }\\
    \subfloat[Illustration of $\Tilde{E_3}$]{
\begin{tikzpicture}[thick,scale=0.8, every node/.style={transform shape}]
        \node[draw,rectangle] (M) at (0,0) {$\begin{matrix}
            \alpha_1 & \alpha_2 & \cdots & \alpha_a &0 &\cdots & 0\\
            \alpha_2 & 0  & \cdots & 0 & 0  & \cdots & 0 \\ 
            \vdots &  \vdots & \ddots &\vdots &  \vdots & \ddots &\vdots\\
            \alpha_a &0  & \cdots & 0 & 0  & \cdots & 0  \\
             0& 0&\cdots & 0 & 0  & \cdots & 0\\
            \vdots & \vdots & \ddots & \vdots &  \vdots & \ddots &\vdots \\ 
            0&0& \cdots & 0 & 0  & \cdots & 0
        \end{matrix}$};
        \node[draw,rectangle] (N) at (5,0.66) {$\begin{matrix}
            \alpha_1 & \alpha_2 & \cdots & \alpha_a &0 &\cdots & 0\\
            \alpha_2 & 0  & \cdots & 0 & 0  & \cdots & 0 \\ 
            \vdots &  \vdots & \ddots &\vdots &  \vdots & \ddots &\vdots\\
            \alpha_a &0  & \cdots & 0 & 0  & \cdots & 0  \\
             0& 0&\cdots & 0 & 0  & \cdots & 0\\
            \vdots & \vdots & \ddots & \vdots &  \vdots & \ddots &\vdots \\ 
            0&0& \cdots & 0 & 0  & \cdots & 0
        \end{matrix}$};
        \node[draw,rectangle] (P) at (12.25,1.495) {$\begin{matrix}
            \alpha_1 & \alpha_2 & \cdots & \alpha_a &0 &\cdots & 0\\
            \alpha_2 & 0  & \cdots & 0 & 0  & \cdots & 0 \\ 
            \vdots &  \vdots & \ddots &\vdots &  \vdots & \ddots &\vdots\\
            \alpha_a &0  & \cdots & 0 & 0  & \cdots & 0  \\
             0& 0&\cdots & 0 & 0  & \cdots & 0\\
            \vdots & \vdots & \ddots & \vdots &  \vdots & \ddots &\vdots \\ 
            0&0& \cdots & 0 & 0  & \cdots & 0
        \end{matrix}$};
        \draw (M)--(N);
        \draw (N) -- ($(N)!0.33!(P)$);
         \draw[dotted] ($(N)!0.33!(P)$) -- ($(N)!0.66!(P)$);
        \draw ($(N)!0.66!(P)$) -- (P);
        \draw (M.north west) -- (N.north west);
        \draw (M.south east) -- (N.south east);
        \draw (N.north west) -- ($(N.north west)!0.33!(P.north west)$);
         \draw[dotted] ($(N.north west)!0.33!(P.north west)$) -- ($(N.north west)!0.66!(P.north west)$);
        \draw ($(N.north west)!0.66!(P.north west)$) -- (P.north west);
	\draw (N.south east) -- ($(N.south east)!0.33!(P.south east)$);
         \draw[dotted] ($(N.south east)!0.33!(P.south east)$) -- ($(N.south east)!0.66!(P.south east)$);
        \draw ($(N.south east)!0.66!(P.south east)$) -- (P.south east);
    \end{tikzpicture}
    }
    \caption{Illustration of the uncorrectable errors for $m = 3$, \emph{i.e.} 4-tensors on the $\mathbb{F}_q$}
    \label{fig:ParticularErrors}
\end{figure}

\begin{figure}[H]
    \centering \includegraphics[width=\linewidth]{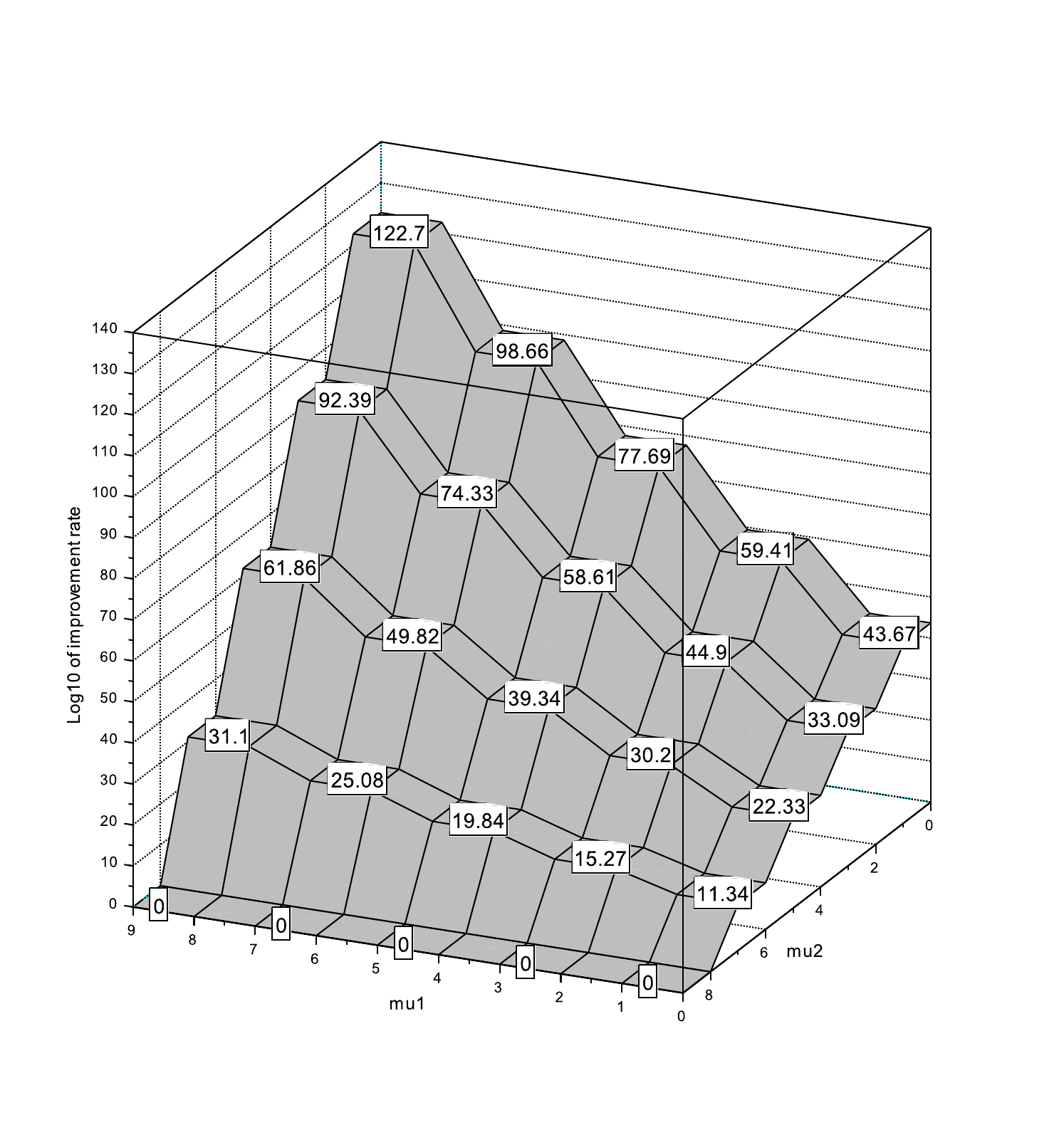} \\[-1cm]
    \begin{minipage}{0.8\linewidth}
        For $q = 2$ and $n = 10$, computes $\log_{10}(N_2/N_1)$ where $N_i$ is the number of correctable errors with Algorithm $i$ for the code $\mathcal{C}_\alpha(\llbracket 0,\mu_1\rrbracket \times \llbracket 0,\mu_2\rrbracket)$.
        \emph{e.g. Algorithm 2 corrects approximately $10^{39.3}$ times more errors than Algorithm 1 on $\mathcal{C}_\alpha(\llbracket 0,5\rrbracket^2)$.} 
    \end{minipage}  
       
    \caption{Improvement rate between the correctable errors of the two fibre-wise technique.}
    \label{fig:ImprovBoundFibreDec}
\end{figure}

\section{Details on the left Euclidean factorisation algorithm}
\label{sec:DetailsFactoringAlgorithm}
We will demonstrate a method to recover a bilinearised $q$-polynomial $f(X,Y)$ from a given pair of $q$-polynomials $(V(Z),N(X,Y))$ satisfying $N(X,Y) = V(f(X,Y))$ where $V(Z) \in \MM_{q,\F_{q^n}}[Z]$ is non-zero and $N(X,Y) \in \MM_{q,\F_{q^n}}[X,Y]$, as an alternative to the Euclidean division algorithm used in classical Gabidulin decoding in the non-commutative ring of linearised $q$-polynomials. For each $u \in \F_{q^n}$ and $\ell \in {\N_0}$, denote by $u^{1/{q^\ell}}$ the preimage of $u$ under the isomorphism $x \mapsto x^{q^\ell}$.

\begin{Lemma}\label{label:AlgoFindf}
Let $\mathcal{S}$ be a non-empty subset of $\llbracket 0,n-1 \rrbracket^2$. Let $V(Z) = \sum_{\ell = 0}^t v_\ell Z^{q^\ell}\in \mathscr{M}_{q,\mathbb{F}_{q^n}}[Z]$ be non-zero and let $f(X,Y) = \sum_{ s \in \mathcal{S}} f_{ s} X^{q^{ s_1}}Y^{q^{ s_2}} \in \mathcal{C}(\mathcal{S})$. Denote the integer $\ell_{min} := \min\{\ell \in \llbracket 0,t\rrbracket \ | \ v_\ell \neq 0 \}$. 

Let $N(X,Y) :=V(f(X,Y)) =\sum_{ r \in \mathcal{S} + \llbracket 0,t\rrbracket} n_{r} X^{q^{r_1}}Y^{q^{r_1}}$ defined by $N(X,Y) = V(f(X,Y))$. We have the following expression of the coefficients of $f(X,Y)$. Moreover, the coefficients of the polynomial $f(X,Y)$ are recursively computable solely from the coefficients of $V(Z)$ and $N(X,Y)$.
\begin{equation}
     \forall  s \in \mathcal{S} : \ \  f_{ s} = \left( \! v_{\ell_{min}}^{-1} \! \left(n_{ s + \ell_{\min}} - \!\!\!\!\!\!\sum_{\substack{\ell \in \llbracket {\ell_{min}} +1,t \rrbracket\\  s + {\ell_{min}}- \ell \in \mathcal{S}}}\!\!\!\!\!\! v_\ell f^{q^\ell}_{ s + \ell_{min}- \ell}\right)\!\!\right)^{\frac 1{q^{\ell_{min}}}}.
    \label{eq:CoefsfWithCoefsVN}
\end{equation}
\end{Lemma}
\begin{proof}
With Lemma~\ref{lemma:CoefCompositionVof}, we have the following for each $ s \in \mathcal{S}$.
$$n_{ s + \ell_{min}} =  \sum_{\substack{\ell \in \llbracket \ell_{\min}, t\rrbracket\\   s + \ell_{min}-\ell \in \mathcal{S}}} v_\ell f_{ s + \ell_{min}- \ell}^{q^\ell} = v_{\ell_{min}}f_{ s}^{q^{\ell_{min}}} +   \sum_{\substack{\ell \in \llbracket \ell_{\min} +1, t\rrbracket\\   s + \ell_{min}-\ell \in \mathcal{S}}} v_\ell f_{ s + \ell_{min}- \ell}^{q^\ell}$$
Since $v_{\ell_{min}}$ is invertible by definition of ${\ell_{min}}$, we obtain equation (\ref{eq:CoefsfWithCoefsVN}). Moreover, for each ${s} \in \mathcal{S}$, 
$f_{{s}}$ is a function of the coefficients in $V(Z)$, in $N(X,Y)$ and of $f_{s -1}, f_{{s}-2},\dots$ while $s-j\in \mathcal S$. In addition, the same expression of the coefficients in $N(X,Y)$ shows that $f_s$ such that $s \in \mathcal S$ and $s-j \not\in \mathcal S$ for all $j  \in \mathbb N$, can be computed knowing only $V(Z)$ and $N(X,Y)$, see Algorithm~\ref{alg:divisionalgorithm}.
\end{proof}

\begin{Ex}
    Given $V(Z) = Z^{q} + Z$, $f(X,Y) = aX^qY^q + bX^qY + cXY^q + dXY$, and $\mathcal{S} = \llbracket0,1\rrbracket^2$, we have $N(X,Y) = a^qX^{q^2}Y^{q^2} + b^qX^{q^2}Y^q + c^qX^{q}Y^{q^2} + (a + d^q)X^qY^q +  bX^qY + cXY^q + dXY$. First, we can recover the coefficient of the monomial $XY$ in $f(X,Y)$ since it is a function of the coefficient in $V(X,Y)$ and the coefficient in $N(X,Y)$ in front of $XY$. The minimum $q$-degree of $V$ is zero, and indeed $d$ is found on the last monomial of $XY$. Then, knowing that $f_{(0,0)} = d$, we can recover $f_{(1,1)}$ using $(a + d^q)$.
\end{Ex}

\begin{algorithm}
\caption{Factoring on the left}\label{alg:divisionalgorithm}
\begin{algorithmic}
\Require $n$ an integer, $q$ a prime power, $\mathcal{S}$ a subset of $\llbracket 0,n-1\rrbracket^2$, $V(Z)$ and $N(X,Y)$ $q$-polynomials such that $V(Z) \neq 0$ and $\Supp(N) \subseteq \mathcal{S} + \llbracket 0,\deg_q(V) \rrbracket $.
\State \phantom{x}
\State $M \gets \max(\max \pi_1(\mathcal{S}),\max \pi_2(\mathcal{S}))$
\State $\theta \gets \deg_q(V)$
\State $f \gets 0\in \mathscr{M}_{q,\mathbb{F}_{q^n}}[X,Y]$
\State $\ell_{\min} \gets \min\{\ell \in \llbracket 0,t \rrbracket \ | \ v_\ell \neq 0\}$
\State \phantom{x} \Comment{The function \texttt{Coef}($F,\Mmm$) returns the coefficient in front of the monomial $\Mmm$ in the polynomial $F$.} 

\If{$N = 0$}
    \State \textbf{return} $f$ 
\EndIf
\For{$ \delta \in \llbracket -M,M\rrbracket$}
\For{$\tau$\textbf{ from $0$ to }$M$}
\State $ s \gets ( \delta+\tau,\tau)$
\If{$s \in \mathcal{S}$}
\State $\Sigma \gets \sum_{\substack{\ell \in \llbracket {\ell_{min}} +1,\theta \rrbracket\\  s + {\ell_{min}}- \ell \in \mathcal{S}}} \texttt{Coef}(V,Z^{q^{\ell}}) \times  \left(\texttt{Coef}\left(f,X^{q^{ s_1+ \ell_{min}- \ell}} Y^{q^{ s_2+ \ell_{min}- \ell}} \right)\right)^{q^\ell} $
\State $c \gets\left( (\texttt{Coef}(V,Z^{q^{\ell_{min}}})^{-1}  \times\left(\texttt{Coef}\left(N,X^{q^{ s_1+ \ell_{min}}} Y^{q^{ s_2+ \ell_{min}}} \right) - \Sigma\right) \right)^{1/q^{\ell_{min}}} $
\State $f\gets f +c X^{q^{ s_1}}  Y^{q^{ s_2}}$
\EndIf
\EndFor
\EndFor
\State \textbf{return} $f$
\end{algorithmic}
\end{algorithm}

\section{Tensor-rank combinatorics}
We denote by $\sgbc abq$ the Gaussian binomial coefficient, \emph{i.e.} the number of $a$-dimensional subspaces of $\F_q^b$, see \cite[Ch 24]{Van_Lint_A_Course_In_Combinatorics}. We consider tensors in $\F_q^{m \times n \times k} =  \F_q^n \otimes \F_q^m \otimes \F_q^k$, for integers $k,m,n$. All the other definitions hold. It is known that the number of elements in $\F_q^{m \times n \times k}$ with tensor-rank precisely one is exactly $\frac{(q^m-1)(q^n-1)(q^k-1)}{(q-1)^2}$. We remind the reader that a rank one matrix is uniquely defined by the choice of a one-dimensional column-space and a one dimensional row-space.

\begin{proof}[Proof of Lemma~\ref{lemma:CountingNbTrank2}] \label{proof:CountingNbTrank2}
    By \cite[Prop~14.45]{Burgisser1997ch14}, let $\Gamma \in \mathbb{F}_q^{m \times n\times k}$, then $\trank(\Gamma) = 2$ if and only if $\slicesp_3(\Gamma)$ is contained in a space spanned by two rank one matrices, but is not contained in a space spanned by only one rank one matrix, so if and only if $\slicesp_3(\Gamma) = \Span_{\mathbb{F}_q} (A_1,A_2)$  where $A_1,A_2 \in \mathbb{F}_q^{m\times n}$ of rank one, or $\slicesp_3(\Gamma) = \mathbb{F}_q \cdot B$ where $B \in \mathbb{F}_q^{m \times n}$ is of rank $2$. 
    
    \begin{enumerate}
    	 \item Let $V$ be a $t$ dimensional subspace of $\mathbb{F}_q^{m \times n}$, with $t \leq k$. The number of tensors $\Gamma \in \mathbb{F}_q^{ m \times n \times k}$ that have $V$ as first slice-space is the number of (ordered) generating families of $V$ of length $k$, in other words, the number of rank $t$ matrices in $\mathbb{F}_q^{k \times t}$. This number is given in  \cite[Theorem 25.2]{Van_Lint_A_Course_In_Combinatorics}.
    
        \item There are exactly $\gbc 2mq \gbc 2nq q(q-1)(q^2-1) = \frac{q(q^m-1)(q^{m-1}-1)(q^n-1)(q^{n-1} -1)}{(q-1)(q^2-1)}$ matrices of rank exactly two in $\mathbb{F}_q^{m \times n}$ (see \cite[Theorem 25.2]{Van_Lint_A_Course_In_Combinatorics}). Thus, the number of spaces spanned by such matrices is
        $$\gbc 2mq \gbc 2nq q(q^2-1) =  \frac{q(q^m-1)(q^{m-1}-1)(q^n-1)(q^{n-1} -1)}{(q-1)^2(q^2-1)}. $$

        \item There are exactly $\frac{q^2}{2}\gbc 1mq \gbc 1nq\gbc 1{m-1}q\gbc 1{n-1}q$ subspaces $V$ of $\mathbb{F}_q^{m \times n}$ such that :
    \begin{itemize}[label = \textbullet]
        \item $V$ is generated by two rank one matrices,
        \item the row space of $V$ is two-dimensional,         \item and the column space of $V$ is two-dimensional,
    \end{itemize}
    where we define the row space and the column space of $V$ to be respectively the sum of the row spaces and column spaces of the elements in $V$. Indeed, to enumerate all the possible spaces, it suffices to fix two pairs $(u_1\mathbb{F}_q,u_2\mathbb{F}_q)$ and $(v_1\mathbb{F}_q,v_2\mathbb{F}_q)$ of distinct one dimensional subspaces respectively in $\mathbb{F}_q^m$ and in $\mathbb{F}_q^n$ and define $V = \Span_{\mathbb{F}_q}(u_1^\top v_1,u_2^\top v_2)$. There are $\frac{(q^m - 1)}{(q-1)}\frac{(q^m - q)}{(q-1)}\frac{(q^n - 1)}{(q-1)}\frac{(q^n - q)}{(q-1)} = q^2\gbc 1mq \gbc 1nq\gbc 1{m-1}q\gbc 1{n-1}q$ possibilities to fix these pairs of spaces. Additionally, given such a pair, since any linear combination with two non-zero scalars of two rank one (independent) matrices with different column space and different row space is a rank two matrix,     the only bases of $V$ spanned by two rank one matrices are $(\lambda u_1^\top v_1,\mu u_2^\top v_2)$ and $(\mu u_2^\top v_2,\lambda u_1^\top v_1)$ where $\lambda,\mu \in \mathbb{F}_q^\times$. In other words, there is only one other possibility in the ennumeration above that results in the same space. In the end, the number of such spaces is half the number of possibilities above, and we have the wanted result.

    \item There are exactly $\frac{1}{q+1}\gbc 1mq \gbc 1nq\gbc 1{m-1}q$ subspaces $V$ of $\mathbb{F}_q^{m \times n}$ such that :
    \begin{itemize}[label = \textbullet]
        \item $V$ is generated by two rank one matrices,
        \item the row space of $V$ is one-dimensional,
        \item and the column space of $V$ is two-dimensional.
    \end{itemize}
    Similarly, to enumerate all possible spaces $V$, it suffices to fix $(u_1\mathbb{F}_q,u_2\mathbb{F}_q)$ a pair of subspaces of $\mathbb{F}_q^m$ and $v\mathbb{F}_q$ a subspace of $\mathbb{F}_q^n$. For any such choice, one can construct the space $V = \Span_{\mathbb{F}_q}(u_1^\top v,u_2^\top v)$. There are $\frac{(q^m-1)(q^{m}-q)}{(q-1)^2}\frac{(q^n-1)}{(q-1)}$ choices to fix those spaces. Additionally, given such elements, then every non-zero matrix in $V$ has rank one and every couple of independent matrices in $V$ is basis $V$. Hence, the configurations that results in $V$ are exactly the configurations of the form $(\tilde{u_1}\F_q,\tilde{u_2}\F_q)$ and $v\F_q$ where $\tilde{u_1}\F_q$ and $\tilde{u_2}\F_q$ are two different lines in $\Span(u_1,u_2)$, \emph{i.e.} $V$ is listed exactly $\frac{q^2-1}{q-1}\frac{q^2-q}{q-1} = q(q+1)$ times.     In conclusion, the number of wanted subspaces is indeed
    $$ \frac{(q^m -1)(q^n -1)(q^{m-1}-1)}{(q-1)^3(q+1)} = \frac{1}{q+1}\gbc 1mq \gbc 1nq\gbc 1{m-1}q.$$
    
    \item Symmetrically, there are exactly $\frac{1}{q+1}\gbc 1mq \gbc 1nq\gbc 1{n-1}q$ subspaces $V$ of $\mathbb{F}_q^{m \times n}$ such that :
    \begin{itemize}[label = \textbullet]
        \item $V$ is generated by two rank one matrices,
        \item the row space of $V$ is two-dimensional,
        \item and the column space of $V$ is one-dimensional.
    \end{itemize}
    \end{enumerate}

To conclude, with the point 1. and the point 2., we know that the number of tensors of tensor-rank two with one dimensional first slice-space is the following:
 $$  \frac{q(q^m-1)(q^{m-1}-1)(q^n-1)(q^{n-1} -1)(q^k -1)}{(q-1)^2(q^2-1)} =\gbc 2mq \gbc 2nq \gbc 1kq q(q-1)(q^2-1). $$
 And with 1., 3., 4., and 5., since the number of rank two matrices in $\mathbb{F}_q^{k\times 2}$ is $\gbc{2}{k}{q} q(q-1)(q^2 -1)$, the number of tensors with tensor-rank one and two-dimensional first slice-space is
 $$q(q-1)(q^2-1)\gbc 1mq \gbc 1nq \gbc 2kq \left(\frac{q^2}{2}\gbc 1{m-1}q\gbc 1{n-1}q  + \frac{\gbc 1{m-1}q + \gbc 1{n-1}q}{q+1} \right). $$
 In the end, summing the two expressions above, one can check that we obtain the wanted number.
 \end{proof}

Consequently, for $k = n = m$ we have 
\begin{align*}
    &\frac{q(q^n-1)^3}{(q-1)^3(q^2-1)} \left(\frac{(q^{n-1} - 1)^3q^2(q+1)}{2} + (q-1)\left(\frac{q^{2n } +q^{2n} + q^{2n}}{q^2} -2 \frac{q^{n} + q^{n} +q^{n}}{q} +3 \right)\right)\\
    =&\frac{q(q^n-1)^3}{(q-1)^3(q^2-1)} \left(\frac{(q^{n-1} - 1)^3q^2(q+1)}{2} + (q-1)\left(3q^{2n-2} -6 q^{n-1}  +3 \right)\right) \\
    =&\frac{q(q^n-1)^3}{(q-1)^3(q^2-1)} \left(\frac{(q^{n-1} - 1)^3q^2(q+1)}{2} + 3(q-1)\left(q^{n-1} -1 \right)^2\right) \\
    =&\frac{q(q^n-1)^3(q^{n-1}-1)^2}{(q-1)^3(q^2-1)} \left(\frac{(q^{n-1} - 1)q^2(q+1)}{2} + 3(q-1)\right). \\
\end{align*}

\begin{proof}[Proof of Lemma~\ref{lemma:UBontensorsofrankatmost}]\label{proof:UBontensorofrankatmost}
    Consider $\Gamma \in \F_{q}^{n\times n\times n}$ such that $\trank(\Gamma) = r \leq R$. Then there exists $\Gamma_1,\dots,\Gamma_r \in \F_q^{n\times n\times n}$ with $\trank(\Gamma_t) = 1$ for each $t \in \llbracket 1,r\rrbracket$ and such that $\Gamma = \sum_{t = 1}^r \Gamma_t$. 

    First assume that $R-r$ is even. Let $\Gamma_0 \in \F_q^{n\times n\times n}$ be such that $\trank(\Gamma_0) = 1$. Then $\Gamma$ can be written as a sum of exactly $R$ elements of tensor-rank one by $\Gamma = \sum_{t = 1}^{r} \Gamma_t + \sum_{t = r+1}^R (-1)^t \Gamma_0$. Now assume that $R -r$ is odd.    
    \begin{itemize}
        \item Assume that $R-r \geq 3$. Notice that the zero tensor can be written as a sum of $3$ elements of tensor-rank one $0 = \Gamma^{(0)}_0 + \Gamma^{(1)}_0 + \Gamma^{(2)}_0$ defined for all $i,j,k \in \llbracket 1,n\rrbracket$ by 
         $$ \begin{array}{l}
              \Gamma^{(0)}_0[i,j,k] = \left\{\begin{array}{rl}
             1 &\text{if } i \in \{1,2\}, j=k = 1 \\
             0 &\text{otherwise;}
        \end{array}\right.  \\[0.5cm]
               \Gamma^{(1)}_0[i,j,k] = \left\{\begin{array}{rl}
             -1 &\text{if } i =j = k = 1 \\
            0 &\text{otherwise;}
        \end{array}\right.\\[0.5cm]
        \Gamma^{(2)}_0[i,j,k] = \left\{\begin{array}{rl}
             -1 &\text{if } i=2, j = k = 1 \\
            0 &\text{otherwise.}
        \end{array}\right.
         \end{array} $$
         Therefore, we have $ \Gamma = \Gamma_0^{(0)} + \Gamma_0^{(1)} + \Gamma_1^{(2)} + \sum_{t = 1}^{r} \Gamma_t + \sum_{t = r+1}^{R-3} (-1)^t \Gamma_0$ hence $\Gamma$ has an expression as a sum of exactly $3+r+(R-3)-(r+1)+1 = R$ tensors of tensor-rank one.
        \item Assume that $R-r = 1$ and that there exist $\tau \in \llbracket 1,r\rrbracket$ and $(i_0,j_0,k_0),(i_1,j_1,k_1)\in \llbracket 1,n\rrbracket^3$ with $(i_1,j_1,k_1)\neq (i_0,j_0,k_0)$ and $\Gamma_\tau[i_0,j_0,k_0] \neq 0$ and $\Gamma_\tau[i_1,j_1,k_1] \neq 0$.         Without loss of generality, assume that $i_0 \neq i_1$. Consider $\Gamma^{(0)}_\tau,\Gamma_{\tau}^{(1)} \in (\mathbb{F}_q^n)^{\otimes 3}$ defined by the following relations.
        $$\forall i,j,k \in \llbracket 1,n\rrbracket: \Gamma^{(0)}_\tau[i,j,k] = \left\{\begin{array}{rcl}
             \Gamma_\tau[i,j,k] &\text{if}& i = i_0 \\
             0 &\text{if}& i \neq i_0 \\
        \end{array}\right. \quad and \quad \Gamma^{(1)}_\tau[i,j,k] = \left\{\begin{array}{rcl}
             0 &\text{if}& i = i_0 \\
             \Gamma_\tau[i,j,k] &\text{if}& i \neq i_0 \\
        \end{array}\right.$$
                Clearly $\Gamma^{(0)}_\tau + \Gamma_{\tau}^{(1)} = \Gamma_\tau$, and notice that $\trank(\Gamma^{(0)}_\tau) = \trank(\Gamma_{\tau}^{(1)})= 1$. Therefore, $\Gamma = \Gamma^{(0)}_\tau + \Gamma_{\tau}^{(1)} + \sum_{t \in \llbracket 1,R-1\rrbracket \backslash\{\tau\}} \Gamma_t  $ can be written as a sum of $R$ elements of tensor-rank one.

        \item Finally, assume that $R-r = 1$ and that each tensor $\Gamma_t$, where $t \in \llbracket 1,R-1\rrbracket$, has a single non-zero entry. Denote by $a = \Gamma_1[i_1,j_1,k_1] \neq 0$ the single non-zero entries of $\Gamma_1$ with a certain $(i_1,j_1,k_1)\in \llbracket 1,n\rrbracket^3$. We can define $\Tilde{\Gamma_{1}},\Tilde{\Gamma_{1.5}}  \in (\mathbb{F}_q^n)^{\otimes 3}$ such that for each $i,j,k \in \llbracket 1,n\rrbracket$ we have          $$ \begin{array}{l}
              \Tilde{\Gamma_{1}}[i,j,k] = \left\{\begin{array}{rl}
             a &\text{if } (i,j,k) = (i_1,j_1,k_1) \\
             1 &\text{if } (i,j,k) = (i_2,j_1,k_1) \\
             0 &\text{otherwise;}
        \end{array}\right.  \\[0.5cm]
            \Tilde{\Gamma_{1.5}}[i,j,k] = \left\{\begin{array}{rl}
             -1 &\text{if } (i,j,k) = (i_2,j_1,k_1) \\
             0 &\text{otherwise;}
        \end{array}\right.
         \end{array} $$        
    with any $i_2\neq i_1$. Therefore, $\Gamma = \Tilde{\Gamma_{1}} + \Tilde{\Gamma_{1.5}} + \sum_{t = 2}^{R-1}\Gamma_t$ is an expression of $\Gamma$ as a sum of $R$ elements of tensor-rank one.    \end{itemize}
\end{proof}

\end{document}